\pdfoutput=1
\documentclass[12pt]{article}

\usepackage{amsfonts}
\usepackage{amsmath}
\usepackage{amssymb}
\usepackage{amsthm}
\usepackage{array}
\usepackage{bm}
\usepackage{booktabs}
\usepackage[font=small,labelfont=sc]{caption}
\usepackage[usenames,dvipsnames]{color, xcolor}
\usepackage{chngcntr}
\usepackage{colortbl}
\usepackage{comment}
\usepackage{dcolumn}
\usepackage{enumitem}
\usepackage{etoolbox}
\usepackage{float}
\usepackage{footnote}
\usepackage{fullpage}
\usepackage{graphicx}
\usepackage[
  bookmarks=true, 
  bookmarksopen=false, 
  pdfauthor={Federico Echenique, Taisuke Imai, Kota Saito}, 
  breaklinks=true, 
  colorlinks=true, 
  citecolor=OliveGreen, 
  linkcolor=OliveGreen, 
  urlcolor=OliveGreen, 
  filecolor=OliveGreen, 
  pagebackref=false, 
  pdfdisplaydoctitle=true,
  pdffitwindow=false, 
  pdfstartview=FitH, 
]{hyperref}
\usepackage{kbordermatrix}
\usepackage{latexsym}
\usepackage{natbib}
\usepackage{pgfplots}
\usepackage{rotating}
\usepackage{sectsty}
\usepackage{setspace}
\usepackage{subcaption}
\usepackage{tikz}
\usepackage{tabularx}
\usepackage{tcolorbox}
\usepackage{textcomp}

\usepackage[T1]{fontenc}
\usepackage[osf,lining]{libertine}
\usepackage[libertine]{newtxmath}
\usepackage[scaled=0.85]{beramono}

\newtheorem{definition}{Definition}

\newtheorem{fact}{Fact}

\newtheorem{proposition}{Proposition}
\newtheorem{remark}{Remark}

\def\citeapos#1{\citeauthor{#1}'s (\citeyear{#1})}

\def\R{\mathbb{R}}

\definecolor{myPaleBlue}{RGB}{187,204,238}
\definecolor{myPaleYellow}{RGB}{238,238,187}
\definecolor{myPaleRed}{RGB}{255,204,204}

\newcolumntype{H}{>{\setbox0=\hbox\bgroup}c<{\egroup}@{}}

\usepackage{multibib}
\newcites{SOMpapers}{References}
\def\citeSOMpapersapos#1{\citeauthor{#1}'s (\citeyear{#1})}

\newlength\mystoreparindent
\newenvironment{myparindent}[1]{%
  \setlength{\mystoreparindent}{\the\parindent}
  \setlength{\parindent}{#1}
}{%
  \setlength{\parindent}{\mystoreparindent}
}

\title{
  Decision Making under Uncertainty: An Experimental Study in Market Settings 
}

\author{
  Federico Echenique 
  \and 
  Taisuke Imai 
  \and 
  Kota Saito 
  \thanks{
  Echenique: Division of the Humanities and Social Sciences, California Institute of Technology, \href{mailto:fede@hss.caltech.edu}{\texttt{fede@hss.caltech.edu}}. 
  Imai: Department of Economics, LMU Munich, \href{mailto:taisuke.imai@econ.lmu.de}{\texttt{taisuke.imai@econ.lmu.de}}. 
  Saito: Division of the Humanities and Social Sciences, California Institute of Technology, \href{mailto:saito@caltech.edu}{\texttt{saito@caltech.edu}}. 
  } 
}

\date{
  April 30, 2021
}

\begin{document}

\maketitle 

\begin{abstract}
We implement nonparametric revealed-preference tests of subjective expected utility theory and its generalizations. We find that a majority of subjects' choices are consistent with the maximization of some utility function. They respond to price changes in the direction subjective expected utility theory predicts, but not to a degree that makes them consistent with the theory. Maxmin expected utility adds no explanatory power. The degree of deviations from the theory is uncorrelated with demographic characteristics. Our findings are essentially the same in laboratory data with a student population and in a panel survey with a general sample of the U.S. population. 
\end{abstract}

\clearpage 
{\footnotesize 
\paragraph{Acknowledgement.}
The authors wish to thank Aur\'{e}lien Baillon, Yoram Halevy, Yves Le~Yaouanq, Olivia Mitchell, Pietro Ortoleva, Maricano Siniscalchi, and Charles Sprenger for helpful discussions. The authors also thank Noriko Imai for developing the software for the experiment, Yimeng Li for research assistance, John Duffy, Michael McBride, and Jason Ralston for supporting our experiments at ESSL, and Tania Gutsche and Bart Orriens at CESR (University of Southern California) for setting up and implementing the survey on the Understanding America Study. 
This research is supported by Grant SES-1558757 from the National Science Foundation and the TIAA Institute and Wharton School's Pension Research Council/Boettner Center. In addition, Echenique thanks the NSF for its support through the grant CNS-1518941, Imai acknowledges financial support by the Deutsche Forschungsgemeinschaft through CRC TRR 190, and Saito thanks the NSF for its support through the grant SES-1919263. The project described received funding from the TIAA Institute and Wharton School's Pension Research Council/Boettner Center. The content is solely the responsibility of the authors and does not necessarily represent official views of the TIAA Institute or Wharton School's Pension Research Council/Boettner Center. The project described in this paper relies on data from surveys administered by the Understanding America Study, which is maintained by the Center for Economic and Social Research (CESR) at the University of Southern California. The content of this paper is solely the responsibility of the authors and does not necessarily represent the official views of USC or UAS. A summary of our research was published in the TIAA Institute's ``Trends and issues,'' and a first draft of this paper on their ``Research dialogue.''}

\AtBeginEnvironment{thebibliography}{\linespread{1}\selectfont}

\clearpage 
\onehalfspacing
\section{Introduction}
\label{section:introduction} 

Subjective expected utility theory \citep[SEU;][]{savage1954} is the standard model of decision making in the face of uncertainty, where objective probabilities about uncertain states of the world are not known to agents. The theory postulates an agent that behaves as if they have a subjective probabilistic belief over states of the world and maximizes expected utility with respect to this belief. 

While SEU is the leading theory of choice under uncertainty, it is well known to face empirical challenges. In an influential paper, \cite{ellsberg1961} argued that many agents would not conform to SEU. The phenomenon he uncovered, known as the ``Ellsberg paradox,'' suggests that agents may seek to avoid betting on uncertain events in ways that cannot be reconciled with a subjective probability. Such avoidance of uncertain bets is termed ambiguity aversion and subsequent empirical literature has identified it in different contexts and in different subject populations \citep{trautmann2015ambiguity}. 

The empirical literature has relied almost exclusively on the thought experiment discussed in \cite{ellsberg1961}, where agents are offered bets on the color of balls drawn from urns whose composition is not fully specified. The simple binary choice structure of \citeauthor{ellsberg1961} makes it easy to identify violations of SEU through violations of the so-called ``sure-thing principle'' \citep[postulates P2 and P4 of][]{savage1954}. However, the artificial nature of the experiment may question the external validity of its findings. Despite its difficulty, designing choice environments that are more ``natural,'' while providing clean identification, is an important step toward deeper empirical understandings of decision making under uncertainty. 

In this study, we present an empirical investigation of SEU and its generalization, maxmin expected utility \citep[MEU;][]{gilboa1989}, from a different angle, combining an experimental paradigm and measurement techniques that are inspired by recent development in revealed preference theory. 

We consider a ``market'' environment in which an agent chooses a portfolio of Arrow-Debreu securities, given state prices and a budget. 
\cite{echenique2015savage} provide a necessary and sufficient condition for an agent's behavior in the market to be consistent with (risk-averse) SEU. Similarly, \cite{chambers2016} provide a condition for MEU when there are two states of the world. \cite{echenique2018approximate} characterize an ``approximate'' version of SEU, allowing for errors and mistakes. These revealed-preference characterizations provide tests for SEU and MEU, as well as a measure quantifying ``how much'' a dataset deviates from these theories. The tests are {\em nonparametric} in the sense that they do not impose any specific functional forms on utility functions, such as CRRA or CARA. They do assume that agents are risk averse or risk neutral (i.e., they impose a concave von~Neumann-Morgenstern utility). 

We bring these nonparametric revealed-preference tests to actual choices people make in the face of uncertainty. Following the spirit of portfolio-choice tasks introduced by \cite{loomes1991evidence} and \cite{choi2007}, and later used in many other studies \citep[e.g.,][]{ahn2014,carvalho2019complexity,choi2014who,hey2014}, subjects were asked to purchase bundles of state-contingent payoffs under varying budget constraints while not knowing the probabilities of the states of the world. 

Our exploratory analysis starts with checking whether subjects are consistent with SEU, MEU, or more general utility maximization; and, if they are not consistent, how large their violations are. 
In order to investigate the effect of the source of uncertainty on behavior, we generate uncertainty from two different sources. The first source is the classical Ellsberg-style ``urns and balls.'' The second one comes from simulated stock prices. 
To understand the robustness of our findings across different subject populations,  we ran experiments in the laboratory, where we recruited undergraduate students, and on a large-scale internet panel, where we recruited subjects from a general sample of the adult U.S.\ population. 
Finally, we compare our measures of degree of deviation from SEU and the standard measure of ambiguity attitude {\`a} la Ellsberg.

\subsection{Overview of Results}
Our main findings are that: (1) subjects are consistent with general utility maximization and \citeapos{machina1992} probabilistic sophistication, but not SEU; (2) MEU adds no explanatory power to SEU; (3) demand responds to price changes in the direction predicted by SEU, but not enough to make the data fully consistent with SEU; (4) subjects in the laboratory and in the panel display similar patterns; and (5) the correlation between the aforementioned results and demographic characteristics are weak. 

The main purpose of our study was to nonparametrically test theories of decision making under uncertainty. We find that most subjects are utility maximizers (they satisfy the Generalized Axiom of Revealed Preference), and satisfy \citeapos{epstein2000probabilities} necessary condition for probabilistic sophistication.\footnote{Since we test a necessary condition for probabilistic sophistication, we can only say that subjects are {\em not inconsistent} with probabilistic sophistication.} 
However, the news is not good for more restrictive theories. In our experiments, the vast majority of subjects, both in the laboratory and on the panel survey, do not conform to SEU. This finding would be in line with the message of the Ellsberg paradox, except that pass rates for MEU are just as low as for SEU. In fact, in all of our sample, there is only one subject whose choice is consistent with MEU but not SEU. 

One might conjecture that the theories could be reconciled with the data if one allows for  mistakes, but our measures of the distance from the theory do not suggest so. A more forgiving test is to check if prices are negatively correlated with quantities: we refer to this property as ``downward-sloping demand,'' and it bears a close connection to SEU (see \cite{echenique2018approximate} for details). The vast majority of subjects exhibit the downward-sloping demand property, at least to some degree, but not to the extent needed to make them fully consistent with SEU. 

Our panel experiment allows us to connect the distance to SEU with subjects'  sociodemographic characteristics. We find that the distance to SEU is weakly correlated with financial literacy, with more financially-literate subjects being closer to SEU than less literate subjects. A notable finding is the absence of a significant correlation with factors that have been shown to matter for related theories of choice \citep{choi2014who,echenique2018approximate}. In particular, older subjects, subjects with lower educational backgrounds, and subjects with lower cognitive ability, do not necessarily exhibit lower degrees of compliance with SEU. 

One final implication of our results is worth discussing. Our experiments included a version of the standard Ellsberg questions. The distance to SEU, or the degree of compliance with downward-sloping demand, are not related to the answers to the Ellsberg questions, but the variability of uncertainty in our market experiment is. Our between-subject experimental design included a treatment on the variability of the uncertain environment, specifically the variability in the sample paths of the stock price whose outcomes subjects were betting on. Subjects who were exposed to more variable uncertainty seem less ambiguity averse (in the sense of Ellsberg) than subjects who were exposed to less variable uncertainty.

\subsection{Related Literature}
Starting with an influential thought experiment by \cite{ellsberg1961}, many studies have tested SEU and related models of decision making under uncertainty using data from laboratory experiments. \cite{trautmann2015ambiguity} provide an overview of this large but still growing empirical literature. Typical experiments involve ``urns and balls'' following \citeapos{ellsberg1961} original thought experiment, and individual's attitude towards ambiguity is inferred by looking at valuations or beliefs elicited through a series of binary choices \citep[e.g.,][]{abdellaoui2011source, baillon2015ambiguity, chew2017partial, epstein2019ambiguous, halevy2007}. 

Other studies try to estimate parameters of the models of decision making under uncertainty \citep[e.g.,][]{ahn2014, dimmock2015estimating, hey2010, hey2014}. Unlike these studies, our approach is nonparametric, imposing no assumptions on functional form other than risk-aversion.

While the use of artificially generated ambiguity as in Ellsberg-style urns and balls has attractive features that make the interpretation of choice behavior, and experimental implementation, simple, it has been argued that researchers should not rely too much on a paradigm that uses an artificial source of ambiguity. Instead, one should study more ``natural'' sources of ambiguity.\footnote{For example, \cite{camerer1992} note that: ``Experimental studies that do not directly test a specific theory should contribute to a broader understanding of betting on natural events in a wider variety of conditions where information is missing. There are diminishing returns to studying urns!'' (p.~361). Similarly, \cite{gilboa2009} writes: ``David Schmeidler often says, `Real life is not about balls and urns.' Indeed, important decisions involve war and peace, recessions and booms, diseases and cures'' (p.~136).}
In response to these concerns, several studies use non-artificial sources of ambiguity such as stock market indices and temperature \citep{abdellaoui2011source, baillon2015ambiguity, baillon2018learning}. \cite{baillon2018measuring} introduce a method that elicits ambiguity attitudes for natural events while controlling for unobservable subjective likelihoods. \cite{anantanasuwong2019ambiguity} apply the methodology of \cite{baillon2018measuring} to elicit ambiguity perceptions and attitudes from a sample of Dutch investors. 

It is also important to note that there are several studies that try to understand the relationship between sociodemographic characteristics, ambiguity attitudes, and real-world behavior (especially financial).\footnote{\cite{trautmann2015ambiguity} note the importance of this direction: ``Interestingly, the empirical literature has so far provided little evidence linking individual attitudes toward ambiguity to behavior outside the laboratory. Are those agents who show the strongest degree of ambiguity aversion in some decision task also the ones who are most likely to avoid ambiguous investments?'' (p.~89).} This is a subset of a growing empirical literature that seeks to understand the common foundation of a wide class of (behavioral) preferences and to relate cross-/within-country heterogeneity and cultural or sociodemographic characteristics \citep[e.g.,][]{bianchi2019ambiguity, bonsang2015risk, dimmock2015estimating, dimmock2016portfolio, dimmock2016ambiguity, dohmen2018relationship, falk2018global, huffman2019time, sunde2016aging, tymula2013like}. 

Finally, the analysis of our data uses theoretical tools developed and discussed in \cite{chambers2016}, \cite{echenique2015savage}, and \cite{echenique2018approximate}. They require coupling SEU and MEU with risk-aversion. The methods in \cite{polissonquahrenou} avoid the assumption of risk-aversion, but are computationally hard to implement in the case of SEU (their paper contains an application to objective EU, for which their method is efficient). \citeauthor{polissonquahrenou} also develop a test for first-order stochastic dominance in models with known (objective) probabilities. Their test could be seen as a first step towards an understanding of probabilistic sophistication.

\section{Revealed Preferences}
\label{section:theoretical_background} 

We introduce our notions of rationality and ways to test them nonparametrically. The discussion in this section serves to motivate our experimental design (Section~\ref{section:design}) as well as our strategies for data analysis (Section~\ref{section:results}). 

Let $S$ be a finite set of {\em states}. 
Let $\Delta_{++} = \{ \mu \in \R_{++}^S : \sum_{s=1}^S \mu_s = 1 \}$ denote the set of strictly positive probability measures on $S$. 
In the models we consider below, the objects of choice are state-contingent monetary payoffs, or simply {\em monetary acts}, which are vectors in $\R_+^S$. 

A {\em dataset} is a finite collection $(p^k, x^k)_{k=1}^K$, where each $p^k \in \R_{++}^S$ is a vector of strictly positive prices and each $x^k \in \R_+^S$ is a monetary act. $K$ indicates the number of observations. 
The interpretation of a dataset is that each pair $(p^k, x^k)$ consists of a monetary act $x^k$ chosen from the budget $B (p^k, p^k \cdot x^k) = \{ x \in \R_+^S : p^k \cdot x \leq p^k \cdot x^k \}$ of affordable acts. 
We now introduce several concepts of rationalization of the dataset, ordered from the most restrictive to the least. 

Following \cite{echenique2015savage}, we say that a dataset $(p^k, x^k)_{k=1}^K$ is {\em subjective expected utility (SEU) rational with risk aversion} if there exist $\mu \in \Delta_{++}$ and a concave and strictly increasing function $u : \R_+ \to \R$ such that, for all $k$, 
\begin{eqnarray*}
y \in B (p^k, p^k \cdot x^k) \Longrightarrow \sum_{s \in S} \mu_s u (y_s) \leq \sum_{s \in S} \mu_s u (x_s^k) . 
\end{eqnarray*}

\cite{gilboa1989} suggest that an agent in Ellsberg's example may have too little information to form a unique subjective belief, and hence entertains multiple subjective probabilities. Being ambiguity averse, the agent maximizes the minimal expected utility over all possible subjective probabilities she entertains. The resulting theory is called maxmin expected utility.

Following \cite{chambers2016}, we say that a dataset $(p^k, x^k)_{k=1}^K$ is {\em maxmin expected utility (MEU) rational with risk aversion} if there exist a convex set $\Pi \subseteq \Delta_{++}$ and a concave and strictly increasing function $u : \R_+ \to \R$ such that, for all $k$, 
\begin{eqnarray*}
y \in B (p^k, p^k \cdot x^k) \Longrightarrow \inf_{\pi \in \Pi} \sum_{s \in S} \pi_s u (y_s) \leq \inf_{\pi \in \Pi} \sum_{s \in S} \pi_s u (x_s^k) . 
\end{eqnarray*}

\cite{echenique2015savage} and \cite{chambers2016} develop behavioral axiomatic characterizations of risk-averse SEU and risk-averse MEU with two states, which they term the Strong Axiom of Revealed Subjective Expected Utility and Strong Axiom of Revealed Maxmin Expected Utility, respectively.\footnote{\cite{chambers2016} have results for MEU with more states than two, but only under the assumption of risk neutrality.}
Discussing these axioms is beyond the scope of the paper, but roughly speaking, they say that prices and quantities must be inversely related, subject to certain qualifications. We term this {\em downward-sloping demand} property. 

We are able to check whether a given dataset is consistent with SEU or MEU by solving the linear program that is equivalent to the corresponding axiom characterizing each model.  

\begin{fact}\label{fact:seu_rp_test}
A dataset $(x^k, p^k)_{k=1}^K$ is SEU rational with risk aversion if and only if there are strictly positive numbers $v_s^k$, $\lambda^k$, and $\mu_s$ for $s = 1, \dots, S$ and $k = 1, \dots, K$ such that 
  \[ \mu_s v_s^k = \lambda^k p_s^k , \;\; x_s^k > x_{s'}^{k'} \Longrightarrow v_s^k \leq v_{s'}^{k'} . \]
\end{fact}

\begin{fact}\label{fact:meu_rp_test}
Given a dataset $(x^k, p^k)_{k=1}^K$, let $K^0 = \{ k : x^k_1 = x^k_2 \}$, $K^1 = \{ k : x^k_1 < x^k_2 \}$ and $K^2 = \{ k : x^k_1 > x^k_2 \}$. 
A dataset $(x^k, p^k)_{k=1}^K$ is MEU rational with risk aversion and two states if and only if there are strictly positive numbers $\underline{\pi}$, $\bar{\pi}$, $\pi^k$, $v_s^k$, and $\lambda^k$ for $s = 1, 2$ and $k = 1, \dots, K$ such that 
  \[ \pi^k v_s^k = \lambda^k p_s^k , \;\; \bar{\pi} \geq \underline{\pi} , \;\; x_s^k > x_{s'}^{k'} \Longrightarrow v_s^k \leq v_{s'}^{k'} , \]
where $\pi^k = \bar{\pi}$ if $k \in K^1$, $\pi^k = \underline{\pi}$ if $k \in K^2$, and $\pi^k \in [\underline{\pi}, \bar{\pi}]$ if $k \in K^0$. 
\end{fact}

Facts~\ref{fact:seu_rp_test} and~\ref{fact:meu_rp_test} stem from the first-order conditions for the maximization of SEU and MEU. Thanks to these facts, testing for SEU or MEU rationality boils down to finding numbers like $v_s^k$, $\lambda^k$, and $\mu_s$. See Online Appendix~\ref{appendix:implementation} for details. 

When imposed on a dataset, requiring that a decision-maker maximizes expected utility {\em exactly}, without errors, may be too demanding. In order to capture situations where the model holds {\em approximately}, \cite{echenique2018approximate} relax the previous definition of SEU rationality by ``perturbing'' some elements of the model.\footnote{\cite{echenique2018approximate} introduce perturbation of utilities, prices, and beliefs and show that these three sources of perturbations are equivalent. We assume price perturbations here since this source is best suited to our empirical applications.} 

Let $e \in \R_+$ be a number that controls the size of permissible perturbations. We say that a dataset $(x^k,p^k)_{k=1}^K$ is {\em $e$-price-perturbed SEU rational with risk aversion} if there exist $\mu \in \Delta_{++}$, a concave and strictly increasing function $u : \R_+ \rightarrow \R$, and $\varepsilon^k \in \R_{+}^S$ for each $k \in K$ such that, for all $k$, 
  \[ y \in B (\tilde{p}^k, \tilde{p}^k \cdot x^k) \Longrightarrow \sum_{s \in S} \mu_s u (y_s) \leq \sum_{s \in S} \mu_s u (x^k_s) , \]
where for all $k \in K$ and $s \in S$, 
  \[ \tilde{p}^k_s = p^k_s \varepsilon^k_s , \]
and for all $k, l \in K$ and $s, t \in S$, 
  \[ \frac{\varepsilon^k_s/\varepsilon^k_t}{\varepsilon^l_s/\varepsilon^l_t} \leq 1+e . \]
The idea behind the model is that prices are measured, or perceived, with error. We consider the multiplicative form $p^k_s \varepsilon^k_s$ for mathematical convenience. 
As above, we can check this notion of ``approximate'' rationality by setting up a linear programming problem. 

\begin{fact}\label{fact:seu_approx_test}
Given $e \in \R_+$, a dataset $(x^k, p^k)_{k=1}^K$ is $e$-price-perturbed SEU rational with risk aversion if and only if there are strictly positive numbers $v_s^k$, $\lambda^k$, $\mu_s$, and $\varepsilon_s^k$ for $s = 1, \dots, S$ and $k = 1, \dots, K$ such that 
\[ \mu_s v_s^k = \lambda^k \varepsilon_s^k p_s^k , \;\; x_s^k > x_{s'}^{k'} \Longrightarrow v_s^k \leq v_{s'}^{k'} , \]
and for all $k, l \in K$ and $s, t \in S$, 
\[ \frac{\varepsilon_s^k/\varepsilon_t^k}{\varepsilon_s^l/\varepsilon_t^l} \leq 1+e . \]
\end{fact}

Note that price-perturbed SEU with $e = 0$ corresponds to the exact SEU rationality as discussed above, and any dataset becomes $e$-price-perturbed SEU rational if we set $e$ large enough. We are thus interested in the {\em smallest~$e$} for which the dataset becomes $e$-price-perturbed SEU rational. We term this number {\em minimal~$e$} and denote it simply by $e_*$. In the sequel, minimal~$e$ will be our notion of distance between the observed dataset and SEU. 
Using Fact~\ref{fact:seu_approx_test}, we can compute $e_*$ by setting up a constrained minimization problem as follows. 

\begin{fact}\label{fact:minimal_e}
Minimal~$e$ for SEU is a solution to the following problem: 
\begin{equation*}
\begin{aligned}
\min_{(\mu_s, v^k_s, \lambda^k, \varepsilon^k_s)_{k,s}} & \max_{k, l \in K, s, t \in S}  \frac{\varepsilon_s^k/\varepsilon_t^k}{\varepsilon_s^l/\varepsilon_t^l} \\
\rm{ s.t. } & \;\; \mu_s v^k_s = \lambda^k \varepsilon^k_s p^k_s ,  \;\; x^k_s > x^{k'}_{s'} \Longrightarrow v^k_s \leq v^{k'}_{s'} . 
\end{aligned}
\end{equation*}
\end{fact}

\cite{echenique2018approximate} study perturbed versions of objective and subjective expected utility. We can extend their framework to define $e$-price-perturbed MEU and obtain minimal~$e$ for MEU in a similar manner. Given a dataset $(x^k, p^k)_{k=1}^K$, let us define $K^0 = \{ k : x^k_1 = x^k_2 \}$, $K^1 = \{ k : x^k_1 < x^k_2 \}$, and $K^2 = \{ k : x^k_1 > x^k_2 \}$ as in Fact~\ref{fact:meu_rp_test}. 

\begin{fact}\label{fact:minimal_e_meu}
Given $e \in \R_+$, a dataset $(x^k, p^k)_{k=1}^K$ is $e$-price-perturbed MEU rational with risk aversion and two states if and only if there are strictly positive numbers $\underline{\pi}$, $\bar{\pi}$, $\pi^k$, $v_s^k$, $\lambda^k$, and $\varepsilon_s^k$ for $s = 1, 2$ and $k = 1, \dots, K$ such that 
  \begin{equation}\label{eq:e_meu_constraints}
  \pi^k v_s^k = \lambda^k \varepsilon_s^k p_s^k , \;\; \bar{\pi} \geq \underline{\pi} , \;\; x_s^k > x_{s'}^{k'} \Longrightarrow v_s^k \leq v_{s'}^{k'} ,
  \end{equation}
where $\pi^k = \bar{\pi}$ if $k \in K^1$, $\pi^k = \underline{\pi}$ if $k \in K^2$, and $\pi^k \in [\underline{\pi}, \bar{\pi}]$ if $k \in K^0$, and for all $k, l \in K$ and $s, t \in S$, 
\[ \frac{\varepsilon_s^k/\varepsilon_t^k}{\varepsilon_s^l/\varepsilon_t^l} \leq 1+e . \]
Minimal~$e$ for MEU is a solution to the following problem: 
\begin{equation*}
\begin{aligned}
\min_{(\underline{\pi}, \bar{\pi}, \pi^k, v^k_s, \lambda^k, \varepsilon^k_s)_{k,s}} & \max_{k, l \in K, s, t \in S}  \frac{\varepsilon_s^k/\varepsilon_t^k}{\varepsilon_s^l/\varepsilon_t^l} \\
\rm{ s.t. } & \;\; \rm{constraints~\eqref{eq:e_meu_constraints}} . 
\end{aligned}
\end{equation*}
\end{fact}

We now  turn to the most basic Bayesian model of decision under uncertainty. \cite{machina1992} postulate that agents may have a unique subjective probability, but not necessarily decide according to the expected utility with respect to this probability.\footnote{\cite{machina1992} were motivated by paradoxes of choice under risk, not uncertainty.} 
An agent is {\em probabilistically sophisticated} if $x \in \R_+^S$ is evaluated by the distribution it induces given some prior $\mu \in \Delta_{++}$. \cite{epstein2000probabilities} proposes the following necessary condition. 
\begin{fact}\label{fact:epstein_ps}
If a dataset $(x^k, p^k)_{k=1}^K$ is probabilistically sophisticated, then there cannot exist $k, k' \in K$ and $s, t \in S$ such that 
\begin{enumerate}
\item $p_t^k \geq p_s^k$ and $p_s^{k'} \geq p_t^{k'}$, with at least one inequality being strict, and 
\item $x_t^k > x_s^k$ and $x_s^{k'} > x_t^{k'}$. 
\end{enumerate}
\end{fact}

Finally, we know, from \cite{afriat1967} and \cite{varian1982}, that the Generalized Axiom of Revealed Preference (GARP) is a necessary and sufficient condition for a dataset to be consistent with maximization of a well-behaved utility function. 
We say that a bundle $x^k$ is {\em directly revealed preferred to} another bundle $x$, denoted $x^k \succeq^R x$, if $p^k \cdot x^k \geq p^k \cdot x$, and is {\em strictly directly revealed preferred to} $x$, denoted $x^k \succ^R x$, if $p^k \cdot x^k > p^k \cdot x$. 

\begin{fact}\label{fact:garp}
A dataset $(x^k, p^k)_{k=1}^K$ satisfies GARP if and only if for any sequence $((x^{k_1}, p^{k_1}), \dots, (x^{k_L}, p^{k_L}))$, 
  \[ x^{k_1} \succeq^R x^{k_2} , \;\; x^{k_2} \succeq^R x^{k_3} , \;\; \dots, \;\; x^{k_{L-1}} \succeq^R x^{k_L} \;\; \Longrightarrow \;\; \text{not } x^{k_L} \succ^R x^{k_1} . \]
\end{fact}

When a dataset does not satisfy GARP, we are interested in measuring how severe this violation is. Most of the existing studies applying revealed preference methods use the measure called Critical Cost Efficiency Index, inspired by \citeapos{afriat1967} observation that the violation of GARP disappears if expenditures at each observation are deflated.\footnote{CCEI is not without problems: see \cite{echenique2011} for a discussion and a proposed alternative.}

\begin{fact}\label{fact:ccei}
Given a dataset $(x^k, p^k)_{k=1}^K$ and a number $e \in [0, 1]$, define a pair of modified revealed preference relations $\langle \succeq^{R (e)}, \succ^{R (e)} \rangle$ by $x^k \succeq^{R (e)} x$ if $e p^k \cdot x^k \geq p^k \cdot x$ and $x^k \succ^{R (e)} x$ if $e p^k \cdot x^k > p^k \cdot x$. 
We say that a dataset $(x^k, p^k)_{k=1}^K$ satisfies $\text{GARP} (e)$ if and only if for any sequence $((x^{k_1}, p^{k_1}), \dots, (x^{k_L}, p^{k_L}))$, 
  \[ x^{k_1} \succeq^{R (e)} x^{k_2} , \;\; x^{k_2} \succeq^{R (e)} x^{k_3} , \;\; \dots, \;\; x^{k_{L-1}} \succeq^{R (e)} x^{k_L} \;\; \Longrightarrow \;\; \text{not } x^{k_L} \succ^{R (e)} x^{k_1} . \]
Critical Cost Efficiency Index (CCEI) is the supremum over all the numbers $e$ such that $(x^k, p^k)_{k=1}^K$ satisfies $\text{GARP} (e)$: 
  \[ \text{CCEI} = \sup \left\{ e \in [0, 1] : (x^k, p^k)_{k=1}^K \text{ satisfies GARP} (e) \right\} . \]
\end{fact}

\section{Experimental Design}
\label{section:design} 

The goal of our experiment is to nonparametrically test models of decision making under uncertainty, measure the degree of consistency of the data with the models, and relate this degree to the standard measure of ambiguity attitude, as well as subjects' demographic characteristics. Our design mirrors the environment described in Section~\ref{section:theoretical_background}. 

We conducted experiments at the Experimental Social Science Laboratory at the University of California, Irvine (hereafter {\em the laboratory}), and on the Understanding America Study (UAS) panel, a longitudinal survey platform (hereafter {\em the panel}).\footnote{Our experiment was approved by the Institutional Review Board of California Institute of Technology (\#15-0478). It was then reviewed and approved by the director of ESSL and the board of UAS. The module number of our UAS survey is~116 (\href{https://uasdata.usc.edu/survey/UAS+116}{\url{https://uasdata.usc.edu/survey/UAS+116}}).} 
The general structure of tasks in the laboratory and on the panel was the same, but there were several differences between the two. We shall first describe the basic tasks in Section~\ref{section:design_tasks}. Then, in Section~\ref{section:design_implementation}, we turn to the specific features of each implementation\textemdash{} such as recruiting procedures, treatment variations, and incentives. Further details and instructions appear in Online Appendices~\ref{appendix:design_detail} and~\ref{appendix:instruction}.

\subsection{Tasks}
\label{section:design_tasks}

We first describe two tasks used in our experiments: the market task (also referred to as the allocation task), and the Ellsberg two-urn choice task. The market task has two versions, depending on the source of uncertainty. The exact set of tasks differed somewhat depending on the platform: the laboratory or the panel. Table~\ref{table:exp_task_structure} presents an overview of the laboratory and the panel experiments. 

\begin{table}[!t]
\centering
\caption{Structure of the experiment.}
\label{table:exp_task_structure}
\resizebox{\textwidth}{!}{%
\begin{tabular}{lccccc}
\toprule 
 & Treatment & Task~1 & Task~2 & Task~3 & Task~4 \\
\midrule 
Laboratory & Large volatility & Market-stock & Market-Ellsberg & Standard Ellsberg & Survey \\
 & Small volatility & Market-stock & Market-Ellsberg & Standard Ellsberg & Survey \\
\midrule 
Panel & Large volatility & Market-stock & Standard Ellsberg & --- & --- \\
 & Small volatility & Market-stock & Standard Ellsberg & --- & --- \\
\bottomrule 
\end{tabular}%
}
\end{table}

\paragraph{Market task.} 
In the market task, a subject chooses among portfolios of Arrow-Debreu commodities given state prices and a budget. 
The dataset we intend to collect in this task is of the form $(x^k, p^k)_{k=1}^K$, as introduced in Section~\ref{section:theoretical_background}. 
Experimental implementations of similar portfolio-choice problems were introduced by \cite{loomes1991evidence} and \cite{choi2007}, and later used in \cite{ahn2014}, \cite{choi2014who}, and \cite{hey2014}, among others. 

Uncertainty is represented through an underlying three-state {\em state space} $\Omega = \{ \omega_1, \omega_2, \omega_3 \}$. The probabilities of these states are unknown to the subjects. For each choice problem, there are two relevant {\em events}, denoted by $E_s$, $s = 1, 2$. This three-state, two-events, design is part of the methodological innovation in our paper; its purpose will be clear below. Events are sets of states, which are lumped together in ways that will be clear below. The events $E_1$ and $E_2$ are mutually exclusive (i.e., a partition of $\Omega$). Subjects are endowed with~100 (divisible) tokens in each round. An event-contingent payoff may be purchased at a price, which experimentally is captured through an ``exchange value.'' Exchange values, denoted $z_s$, $s = 1, 2$, relate tokens allocated to an event, and monetary outcomes. Given a pair of exchange values $(z_1, z_2)$, subjects are asked to decide on the allocation of tokens, $(a_1, a_2)$, between the two events. A subjects who decides on an allocation $(a_1, a_2)$ earns $x_s = a_s \times z_s$ if event $E_s$ occurs. The sets of exchange values $(z_1, z_2)$ used in the experiments are presented in Table~\ref{table:set_budgets} in the Online Appendix. 

An allocation $(a_1, a_2)$ of tokens is equivalent to buying a $x_s$ units of an Arrow-Debreu security that pays $\$1$ per unit if event $E_s$ holds, from a budget set satisfying $p_1 x_1 + p_2 x_2 = I$, where prices and income $(p_1, p_2, I)$ are determined by the token exchange values $(z_1, z_2)$ in the round.\footnote{We set $p_1 = 1$ (normalization) and $p_2 = z_1 / z_2$. Then, the income is given by $I = 100 \times z_1$.} 

Our design deviates from the other studies mentioned above by introducing a novel event structure. 
There are three underlying states of the world $(\omega_1, \omega_2, \omega_3)$ and we introduce two {\em types} of questions. In Type~1 questions, event~1 is $E_1^1 = \{\omega_1\}$ and event~2 is $E_2^1 = \{\omega_2, \omega_3\}$. In Type~2 questions,  event~1 is $E_1^2 = \{\omega_1, \omega_2\}$ and event~2 is $E_2^2 = \{\omega_3\}$. See Figure~\ref{fig:design_event_structure} for an illustration. This event structure requires SEU decision makers to behave consistently not only within each type of questions but also across two types of questions.\footnote{\citeapos{hey2014} design is the closest to ours. In their experiment, uncertainty was generated by the colors of balls in a Bingo Blower, and subjects were asked to make~76 allocation decisions in two different types. In the first type of problems, subjects were asked to allocate between two of the colors. In the second type, they were asked to allocate between one of the colors and the other two. Note that the motivation of \cite{hey2014} is a parametric estimation of leading models of ambiguity aversion. We test SEU and its generalization nonparametrically.} 

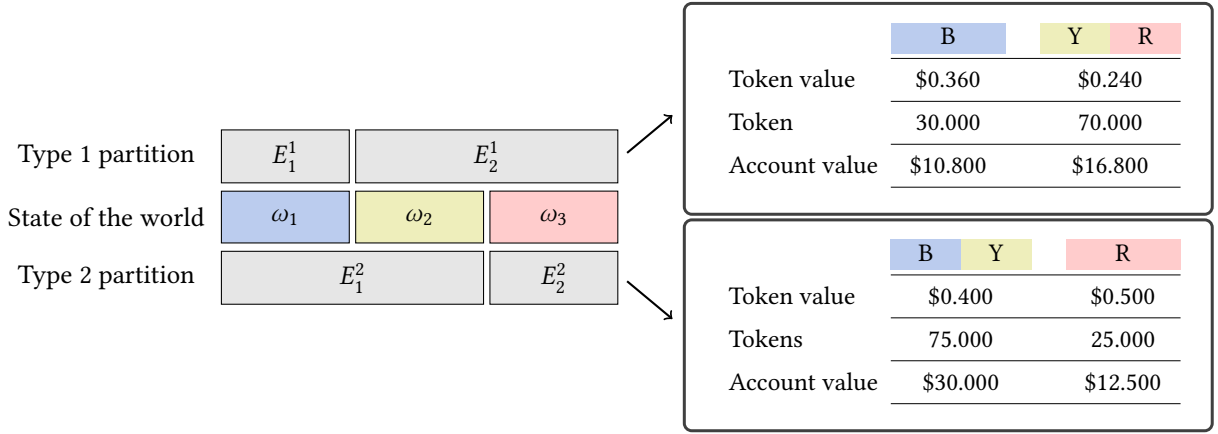
\begin{figure}[!t]
\centering 
\begin{minipage}[c]{0.54\textwidth}
\begin{tikzpicture}[scale=0.85, every node/.style={scale=0.85},
node distance=2.1cm,
textbox/.style={rectangle, minimum width=2cm, minimum height=0.75cm, inner sep=8pt},
state1/.style={rectangle, minimum width=2cm, minimum height=0.75cm, inner sep=8pt, text centered, draw=black, fill=myPaleBlue},
state2/.style={rectangle, minimum width=2cm, minimum height=0.75cm, inner sep=8pt, text centered, draw=black, fill=myPaleYellow},
state3/.style={rectangle, minimum width=2cm, minimum height=0.75cm, inner sep=8pt, text centered, draw=black, fill=myPaleRed},
event1/.style={rectangle, minimum width=2cm, minimum height=0.75cm, inner sep=5pt, text centered, draw=black, fill=gray!20},
event2/.style={rectangle, minimum width=4.1cm, minimum height=0.75cm, inner sep=5pt, text centered, draw=black, fill=gray!20},
]
\node (s1) [state1, align=left] {$\omega_1$};
\node (s2) [state2, right of=s1] {$\omega_2$};
\node (s3) [state3, right of=s2] {$\omega_3$};
\node (e11) [event1, above of=s1, yshift=-1.15cm] {$E_1^1$};
\node (e12) [event2, above of=s3, xshift=-1.05cm, yshift=-1.15cm] {$E_2^1$};
\node (a00) [right of=s3, xshift=-1.2cm, yshift=0.8cm] {$\phantom{1}$};
\node (a01) [right of=s3, xshift=0cm, yshift=1.8cm] {$\phantom{1}$};
\draw [thick,->] (a00) --  (a01);
\node (e21) [event2, below of=s1, xshift=1.05cm, yshift=1.15cm] {$E_1^2$};
\node (e22) [event1, below of=s3, yshift=1.15cm] {$E_2^2$};
\node (b00) [right of=s3, xshift=-1.2cm, yshift=-0.8cm] {$\phantom{1}$};
\node (b01) [right of=s3, xshift=0cm, yshift=-1.8cm] {$\phantom{1}$};
\draw [thick,->] (b00) --  (b01);
\node (text0) [textbox, left of=s1, xshift=-0.7cm] {State of the world};
\node (text1) [textbox, above of=text0, xshift=-0cm, yshift=-1.15cm] {Type~1 partition};
\node (text2) [textbox, below of=text0, xshift=-0cm, yshift=1.15cm] {Type~2 partition};
\end{tikzpicture}
\end{minipage}
\begin{minipage}[c]{0.45\textwidth}
\newcolumntype{C}{>{\centering\arraybackslash}p{0.1\linewidth}}
\newcolumntype{F}{>{\centering\arraybackslash}p{0.2\linewidth}}
\centering 
\resizebox{0.95\textwidth}{!}{%
\tcbox[colback=white]{
\begin{tabular}{l F p{0.02\linewidth} C C}
 & \cellcolor{myPaleBlue} B & & \cellcolor{myPaleYellow} Y & \cellcolor{myPaleRed} R \\
\cmidrule{2-5}
Token value & \multicolumn{1}{c}{\$0.360} & & \multicolumn{2}{c}{\$0.240} \\
\cmidrule{2-5}
Token & \multicolumn{1}{c}{30.000} & & \multicolumn{2}{c}{70.000} \\
\cmidrule{2-5}
Account value & \multicolumn{1}{c}{\$10.800} & & \multicolumn{2}{c}{\$16.800} \\
\cmidrule{2-5} 
\end{tabular}
}}
\resizebox{0.95\textwidth}{!}{%
\tcbox[colback=white]{
\begin{tabular}{l C C p{0.02\linewidth} F}
 & \cellcolor{myPaleBlue} B & \cellcolor{myPaleYellow} Y & & \cellcolor{myPaleRed} R \\
\cmidrule{2-5}
Token value & \multicolumn{2}{c}{\$0.400} & & \multicolumn{1}{c}{\$0.500} \\
\cmidrule{2-5}
Tokens & \multicolumn{2}{c}{75.000} & & \multicolumn{1}{c}{25.000} \\
\cmidrule{2-5}
Account value & \multicolumn{2}{c}{\$30.000} & & \multicolumn{1}{c}{\$12.500} \\
\cmidrule{2-5}
\end{tabular}
}}
\end{minipage}
\caption{(Left) Event structure in two types of questions. (Right) Illustration of the allocation table for a type~1 question (top) and a type~2 question (bottom).} 
\label{fig:design_event_structure}
\end{figure}

The design allows us to examine a very basic aspect of SEU rationality: monotonicity of probability. The monotonicity follows from the fact that SEU-rational agent should consider event $E_1^2 = \{ \omega_1, \omega_2 \}$ is (weakly) more likely than event $E_1^1 = \{ \omega_1 \}$ and, hence, the agent should allocate more tokens on event $E_1^2$ than on event $E_1^1$ if the prices and income are the same. We term this property {\em event monotonicity}. In the experiment, we introduced two consecutive questions that have the same budget set, but with different event structures, to test for event monotonicity. Note that these two questions are asked consecutively, meaning that a severe violation of event monotonicity can be attributed to a lack of understanding of the task or inattention, rather than limited memory. 

Subjects in the experiment make decisions through a computer interface. The {\em allocation table} on the screen contains all the information subjects need to make their decisions in each question; see right panels in Figure~\ref{fig:design_event_structure}. The allocation table displays exchange values $(z_1, z_2)$ for the current question, their current allocation of tokens $(a_1, a_2)$, and implied monetary value of each account, referred to as the ``account value,'' $(a_1 \times z_1, a_2 \times z_2)$. 
Subjects can allocate tokens between two events using a slider at the bottom of the screen; every change in allocation is instantaneously reflected in the allocation table.\footnote{Tokens are divisible (the slider moves in the increment of $0.01$). This ensures that the point on the budget line which equalizes the payouts in the two events (i.e., on the 45-degree line) is technically feasible.}

An important feature of our design is that we implement the task under two different sources of uncertainty. Subjects face two versions of the market task, as we change the source of uncertainty. In the first version, called ``market-Ellsberg,'' uncertainty is generated with an Ellsberg urn. In the second version, termed ``market-stock,'' uncertainty is generated through a stochastic process that resembles the uncertain price of a financial asset, or a market index. The market-Ellsberg version follows \cite{ellsberg1961}, and the empirical literature on ambiguity aversion \citep{trautmann2015ambiguity}. Subjects are presented with a bag containing~30 red, yellow, and blue chips, but they are not told anything about the composition of the bag. 
The three states of the world are then defined by the color of a chip drawn from the bag: state~1 ($\omega_1$) corresponds to drawing a blue chip, state~2 ($\omega_2$) corresponds to drawing a yellow chip, and state~3 ($\omega_3$) corresponds to drawing a red chip.  

In the market-stock task, uncertainty is generated through the realization of simulated stock prices. 
Subjects are presented with a history of stock prices, as in Figure~\ref{fig:design_path_and_budgets}, panel~A.\footnote{We used a Geometric Brownian Motion to simulate~100 stock price paths that share the common starting price and the time horizon. After visually inspecting the pattern of each price path, we handpicked~28 paths and then asked workers on Amazon Mechanical Turk what they believed the future price of each path would be. The elicited belief distributions were then averaged across subjects. Some price paths, especially those with clear upward or downward trends, tend to be associated with skewed elicited belief distributions. Others have more symmetric distributions. We thus selected two relatively ``neutral'' ones from the latter set for the main experiment. See Online Appendix~\ref{appendix:design_detail_price_path}.}  
The chart shows the evolution of a stock price for $300$ periods; the next $200$ periods are unknown, and left blank. Subject are told that prices are determined through a model used in financial economics to approximate real world stock prices. They are told that the chart represents the realized stock price up to period $300$, and that the remaining periods will be determined according to the same model from financial economics. Let the price at period $300$ be the ``starting value'' and the price at period $500$ be the ``target value.'' 
We define three states, given some threshold $R \in (0, 1)$: $\omega_1 = (R, +\infty)$, in which the target value rises by more than $100 R \%$ compared to the starting value (see the blue region in the figure), $\omega_2 = [-R, R]$, in which the price varies by at most $100 R\%$ between the starting value and the target value (the yellow region in the figure), and $\omega_3 = [-1, -R)$, in which the target value falls by more than $100 R\%$ compared to the starting value (the red region in Figure~\ref{fig:design_path_and_budgets}, panel~A). 

\begin{figure}[!t]
\centering 
\includegraphics[width=0.9\textwidth]{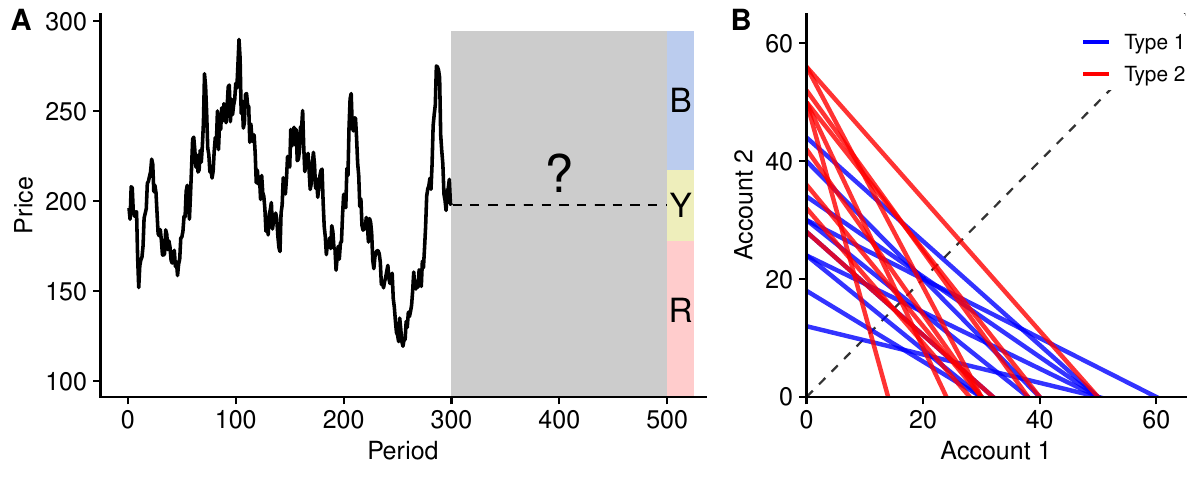}
\caption{(A) Source of uncertainty in the market-stock task. (B) Set of 20 budgets.}
\label{fig:design_path_and_budgets}
\end{figure}

We chose token exchange values $(z_1, z_2)$ for each question to increase the power of our tests. After running several choice simulations to calculate the power of our tests, we select~20 budgets (10 for type~1, 10 for type~2) shown in Figure~\ref{fig:design_path_and_budgets}, panel~B (and Table~\ref{table:set_budgets} in the Online Appendix). Note that event~1 is ``more likely'' in type~2 decision problems since $\{ \omega_1 \} = E_1^1 \subseteq E_1^2 = \{ \omega_1, \omega_2 \}$. In constructing budget sets, we made assets in account~1 relatively more expensive than assets in account~2 in type~2 questions. This is reflected in the steeper slopes for the budget lines presented in Figure~\ref{fig:design_path_and_budgets}, panel~B. 

Several remarks about our experimental design are in order. First, we allowed subjects to make fractional allocations of tokens (up to the third decimal points) between accounts.\footnote{The allocation table (Figure~\ref{fig:design_event_structure}) also displayed account values up to the third decimal place, but subjects were informed that the amount below one cent would be rounded up.} Our fractional allocation design sought to mimic choices from a continuous budget line as much as possible, as in the theoretical models we try to test. Second, we asked two types of allocation decisions. This makes our task demanding for subjects, but it creates a powerful environment for our revealed preference analysis.

\paragraph{Ellsberg two-urn choice task.} 
In addition to the market task described above, we presented our subjects with a standard two-urn version of \citeapos{ellsberg1961} binary choice question. The purpose of including this standard task is to compare the behavior of subjects in the different designs (allocation vs. choice). 
Using this comparison, we can investigate how traditional evaluations of ambiguity aversion via binary choices relate to the conclusions drawn from allocation decisions in a market setting. 

Subjects confront two bags: bag~A and bag~B, each of which contains~20 chips. They receive the following information: Bag~A contains~10 orange chips and~10 green chips. Bag~B contains~20 chips, each of which is either orange or green. The number of chips of each color in bag~B is unknown to them, so there can be anywhere from~0 to~20 orange chips, and anywhere from~0 to~20 green chips, as long as the total number of orange and green chips sums to~20. 

Subjects were offered choices between bets on the color of the chip that would be drawn at the end of the experiment. 
Before choosing between bets, subjects were first asked to choose a fixed color (orange or green; called ``Your Color'') for which they would be paid if they chose certain bets.  They were then asked three questions.\footnote{We adopted the three-question setting akin to \cite{epstein2019ambiguous}, as a way of identifying strict ambiguity preferences. The typical Ellsberg-style experiment would ask only one question, namely the second one.}

The first question asks to choose between a bet that pays \$$X+b$ if the color of the ball drawn from bag~A is ``Your Color'' (and nothing otherwise), and a bet that pays \$$X$ if the color of a ball drawn from bag~B is ``Your Color'' (and nothing otherwise). 
Similarly, the second question asks to choose between a bet that pays \$$X$ if the color of the ball drawn from bag~A is ``Your Color,'' and a bet that pays \$$X$ if the color of a ball drawn from bag~B is ``Your Color''. 
Finally, the third question asks to choose between a bet that pays \$$X$ if the color of the ball drawn from bag~A is ``Your Color'' and a bet that pays \$$X+b$ if the color of a ball drawn from bag~B is ``Your Color''. The payoff $X$ and the bonus $b$ depended on the platform: $(X, b) = (10, 0.5)$ in our laboratory study and $(X, b) = (100, 5)$ in the panel. In our laboratory experiments, the content of bag~B had already been determined at the beginning of the experiment by an assistant. The timing is important to ensure that there is no incentive to hedge \citep{baillon2015ris,epstein2019ambiguous,saito2015preference}. 
The subjects were allowed to inspect the content of each bag after completing the experiment. 

\paragraph{Post-experiment survey.} 
In the laboratory experiment, subjects were asked to fill out a short survey asking for their age, gender, major in college, the three-item cognitive reflection test \citep[CRT;][]{frederick2005cognitive}, and strategies they employed in the allocation tasks if any (see Online Appendix~\ref{appendix:design_detail_post-exp_survey}). 
In the panel study, before exiting the survey module, subjects answered how interesting or uninteresting the survey was and they were also asked to leave any comments if they wished. This is a standard questionnaire that the Understanding America Study (UAS) asks of all its panelist households. The demographic characteristics of the households were already recorded in the previous survey run by the UAS. We could also access datasets from previous surveys that other researchers conducted on the UAS to create additional cognitive and behavioral measures.

\subsection{Implementation}
\label{section:design_implementation}

\paragraph{Interface.} 
We prepared an experimental interface that runs on a web browser. In the panel study, our interface was embedded in the survey page of the UAS. Therefore, subjects in both the laboratory and panel experiments interacted with the exact same interface.

\paragraph{Recruiting and sampling.} 
Subjects for our laboratory study were recruited from a database of undergraduate students enrolled in the University of California at Irvine. The recruiting methodology for the UAS survey is described in detail in the survey website.\footnote{\url{https://uasdata.usc.edu/index.php}.} Within the UAS sample, we drew a stratified random sub-sample with the aim of obtaining a balanced sample of subjects in different age cohorts. 
In particular, we recruited subjects in three age groups: from~20 to~39, from~40 to~59, and~60 and above, randomly from the pool of survey participants.\footnote{The choice of age as the stratification variable is based on the result in \cite{echenique2018approximate}, which shows that the degree of conformity to objective expected utility theory is negatively correlated (younger subjects are closer to the theory than elder subjects) with risk in a similar portfolio choice under known probabilities implemented on several nationally-representative panels.} 

\paragraph{Treatments.} 
In the market-stock task, we prepared two simulated paths of stock prices with different degree of volatility, so that one path seems relatively more volatile than the other, while keeping the general trend in prices as similar as possible between the two paths. Since the perception of volatility is only relative, we embed each path in the common market ``context'' as shown in Figure~\ref{fig:design_stock_var}. Here, the bold black lines indicate the stock under consideration, and the other lines in the background are the same in the two treatments. 

\begin{figure}[t]
\centering 
\includegraphics[width=\textwidth]{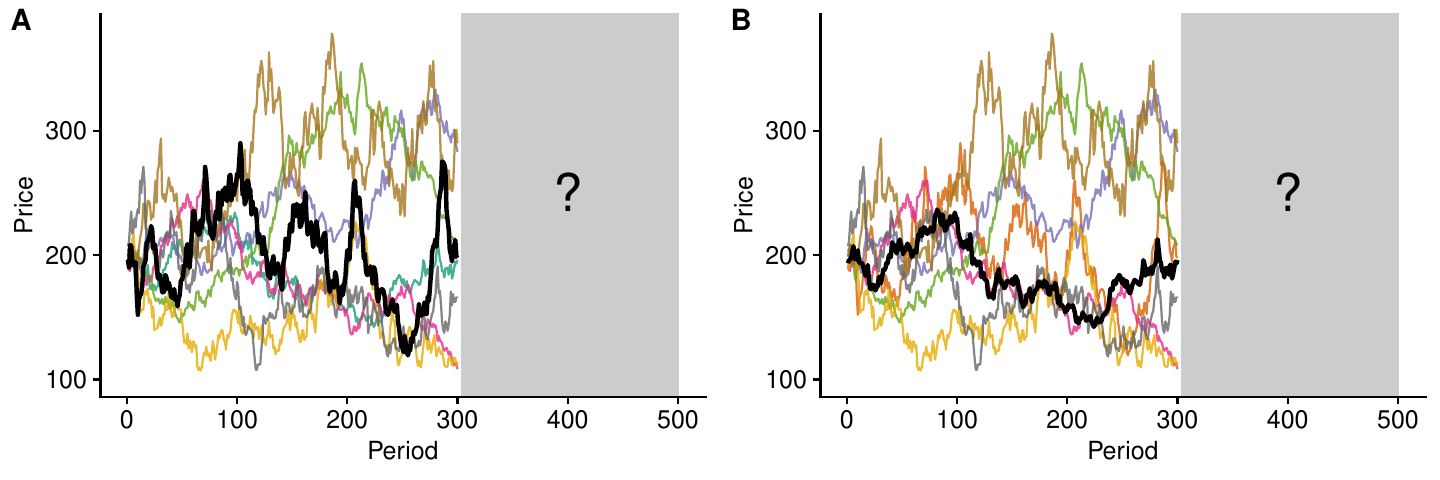}
\caption{Context of market information. (A) Large volatility (B) Small volatility. {\em Notes}: One of these two figures is included in the instructions, depending on the treatment.}
\label{fig:design_stock_var}
\end{figure}

Our treatment variation is the perceived volatility of simulated stock prices (we call the two treatments Large volatility and Small volatility). The subjects were randomly assigned to either a large volatility condition (Figure~\ref{fig:design_stock_var}, panel~A), or a small volatility condition (panel~B).\footnote{In the laboratory study, random assignment to one of the two treatments was done at the session level, meaning that all subjects in the same session were shown the same price path.} 
The instructions for the market-stock task included one of the two charts of Figure~\ref{fig:design_stock_var}, depending on the treatment (see Online Appendix~\ref{appendix:instruction}).

\paragraph{Order of the tasks.} 
Subjects in the laboratory study performed three tasks in the following order: market-stock, market-Ellsberg, and standard Ellsberg. Subjects in the Panel study performed two tasks, market-stock and standard Ellsberg, but due to time constraints we did not implement market-Ellsberg in the panel. Table~\ref{table:exp_task_structure}, which has a summary of the structure of the experiments and treatments, lists the order in which the tasks were completed.

\paragraph{Incentives.} 
In the laboratory study, we used the standard incentive structure of paying-one-choice-at-random. Subjects received a sealed envelope when they entered the laboratory room. The envelope contained a piece of paper, on which two numbers were written. The first number indicated the task number, and the second number indicated the question number in that task. Both numbers were randomly selected beforehand. At the end of the experiment, subjects brought the envelope to the experimenter's computer station. If the selected task was the market task with stock price information, the simulated ``future'' price path was presented on the screen. If, on the other hand, the selected task involved the Ellsberg urn, the subject was asked to pick one chip from the relevant bag. All subjects received a \$7 showup fee. 

In the panel study, four subjects were randomly selected to receive the bonus payment based on their choices in the experiments. Unlike the laboratory study, the bonus payment for these subjects was determined by a randomization implemented by the computer program, but payments were of a much larger scale. All subjects received a participation fee of \$10 by completing the entire survey.

\section{Results}
\label{section:results} 

We present results from the laboratory and panel experiments separately, but our data analysis follows the same structure. 
We shall first discuss the basic patterns of subjects' choices, and then proceed to the revealed-preference tests that were discussed in Section~\ref{section:theoretical_background} above. 
More precisely, we apply the ``exact'' revealed-preference tests for general utility maximization (GARP; Fact~\ref{fact:garp}), SEU (Fact~\ref{fact:seu_rp_test}), MEU (Fact~\ref{fact:meu_rp_test}), as well as the necessary condition for probabilistic sophistication (Fact~\ref{fact:epstein_ps}). 
After observing that most of the subjects' datasets fail the tests, we quantify the severity of violations by CCEI (Fact~\ref{fact:ccei}) and minimal~$e$ (Facts~\ref{fact:seu_approx_test} and~\ref{fact:minimal_e}). 

We also discuss the relationship between the degree of consistency with the models and subjects' demographic characteristics. Finally, we look at the subjects' distance from SEU rationality and their attitude toward ambiguity measured with a simple binary choice task commonly employed in the ambiguity literature. 

All statistical tests reported in this section are two-sided unless otherwise noted.

\subsection{Results from the Laboratory} 
\label{section:results_lab}

We conducted seven sessions at the Experimental Social Science Laboratory of the University of California, Irvine. 
A total of~127 subjects (62 in the small volatility treatment and 65 in the large volatility treatment; $\text{mean age} = 20.16$, $\text{SD} = 1.58$; 35\% male) participated in the study.\footnote{Three additional subjects participated in the study, but we excluded their data from the analysis. One subject accidentally participated in two sessions (thus, the data from the second appearance was excluded). Two subjects spent a significantly longer time on each decision than anyone else. We distributed the instructions for each task of the experiment just before they were to perform that task, meaning that each subject would have to wait until all the other subjects in the session completed the task. We had to ``nudge'' two extremely slow subjects to make decisions more quickly, and hence eliminated their choices from our data.} 
Each session lasted about an hour, and subjects earned on average \$21.3 (including a \$7 showup fee; $\text{SD} = 9.21$).

\paragraph{Allocation decisions in the market tasks.} 
Subjects faced budgets in random order, with one exception, which is related to event monotonicity discussed in Section~\ref{section:design_tasks}. We fixed two consecutive questions, questions~\#5 and~\#6, that had the same budget set, but with different event structures. These were the only questions that were not presented in random order. The purpose of having these questions in fixed order was to check that subjects had a basic understanding of the task. The 5th question was presented as a type~1 question while the 6th question was presented as a type~2 question (recall the terminology from Section~\ref{section:design}). 
Since the event upon which the first account pays off is a larger set in question~\#6 than in question~\#5 ($\{ \omega_1 \} = E_1^1 \subseteq E_1^2 = \{ \omega_1, \omega_2 \}$ by construction), while prices and budget remain the same, subjects should allocate more to the first account in question~\#6 than in question~\#5. 

More than 70\% of subjects satisfy event monotonicity, and this number increases to 90\% if we allow for a small margin of error of five tokens. Moreover, choices are clustered around the allocation which equalizes payout from the two accounts, which may reflect subjects' ambiguity aversion. See Figure~\ref{fig:uci_monotonicity} in the Online Appendix. 

The empirical content of expected utility is captured in part by a negative relation between state prices and allocations as \cite{echenique2018approximate} discuss in depth: a property that can be thought of as ``downward-sloping demand.'' 
We thus look at how subjects' choices responded to price variability between budgets; in particular, we focus on the relation between log price ratios, $\log{(p_2 / p_1)}$, and allocation shares, $x_2/(x_1+x_2)$, pooling choices from all subjects. 
Figure~\ref{fig:uci_dsd_aggregate} shows a negative relation between these two quantities, confirming the downward-sloping demand property at the aggregate level. It holds for both types of questions (type~1 and type~2 event partitions) and in both tasks (market-stock and market-Ellsberg). 

\begin{figure}[!t]
\centering 
\includegraphics[width=\textwidth]{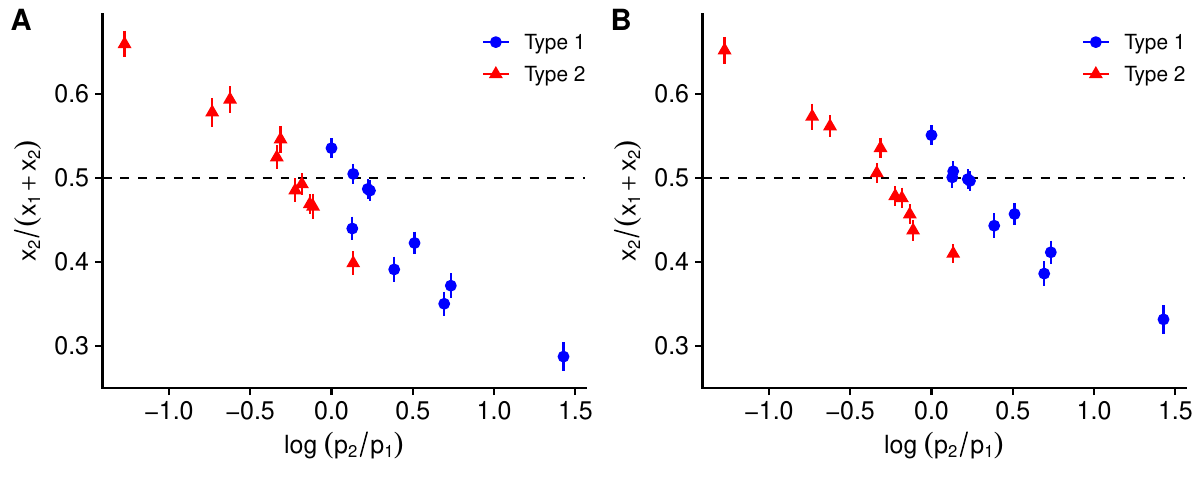}
\caption{Downward-sloping demand at the aggregate level. (A) Market-stock task. (B) Market-Ellsberg task. {\em Notes}: Each point represents mean $x_2/(x_1+x_2)$ at each $\log{(p_2 / p_1)}$ and bars indicate standard error of means.}
\label{fig:uci_dsd_aggregate}
\end{figure}

We also quantify the degree of compliance with the downward-sloping demand property at the individual level by calculating the  correlation $\rho^\text{dsd}$ between $\log{(p_2 / p_1)}$ and $x_2/(x_1+x_2)$.\footnote{We first calculate (Spearman's) correlation coefficient $\rho_t$ for each type ($t = 1, 2$) of questions. To obtain the ``average'' correlation coefficient $\rho^\text{dsd}$, we first convert correlation coefficients to $z$-values by Fisher's transformation, take the average, and convert it back to a correlation coefficient. This procedure is summarized as $\rho^\text{dsd} = \text{tanh} \left( \sum_{t=1}^2 \text{tanh}^{-1} (\rho_t) / 2 \right)$.} 
A significant majority of the subjects (92.1\% in the market-stock task and 88.2\% in the market-Ellsberg task) made choices that responded to prices negatively ($\rho^\text{dsd} < 0$; Figure~\ref{fig:uci_dsd_individual} in the Online Appendix). 

Individual-level data exhibit heterogeneous choice patterns. Figure~\ref{fig:uci_sample_allocation} presents the relationship between $\log{(p_2 / p_1)}$ and $x_2/(x_1+x_2)$ for five selected subjects. As in prior studies \citep[e.g.,][]{ahn2014,choi2007}, there are subjects who responded to price changes smoothly (panels~A and~B), partially or fully ``hedged'' uncertainty by choosing bundles on or close to the 45-degree line (panels~C and~D), and chose bundles all over the space (panel~E). 

\begin{figure}[!t]
\centering 
\includegraphics[width=\textwidth]{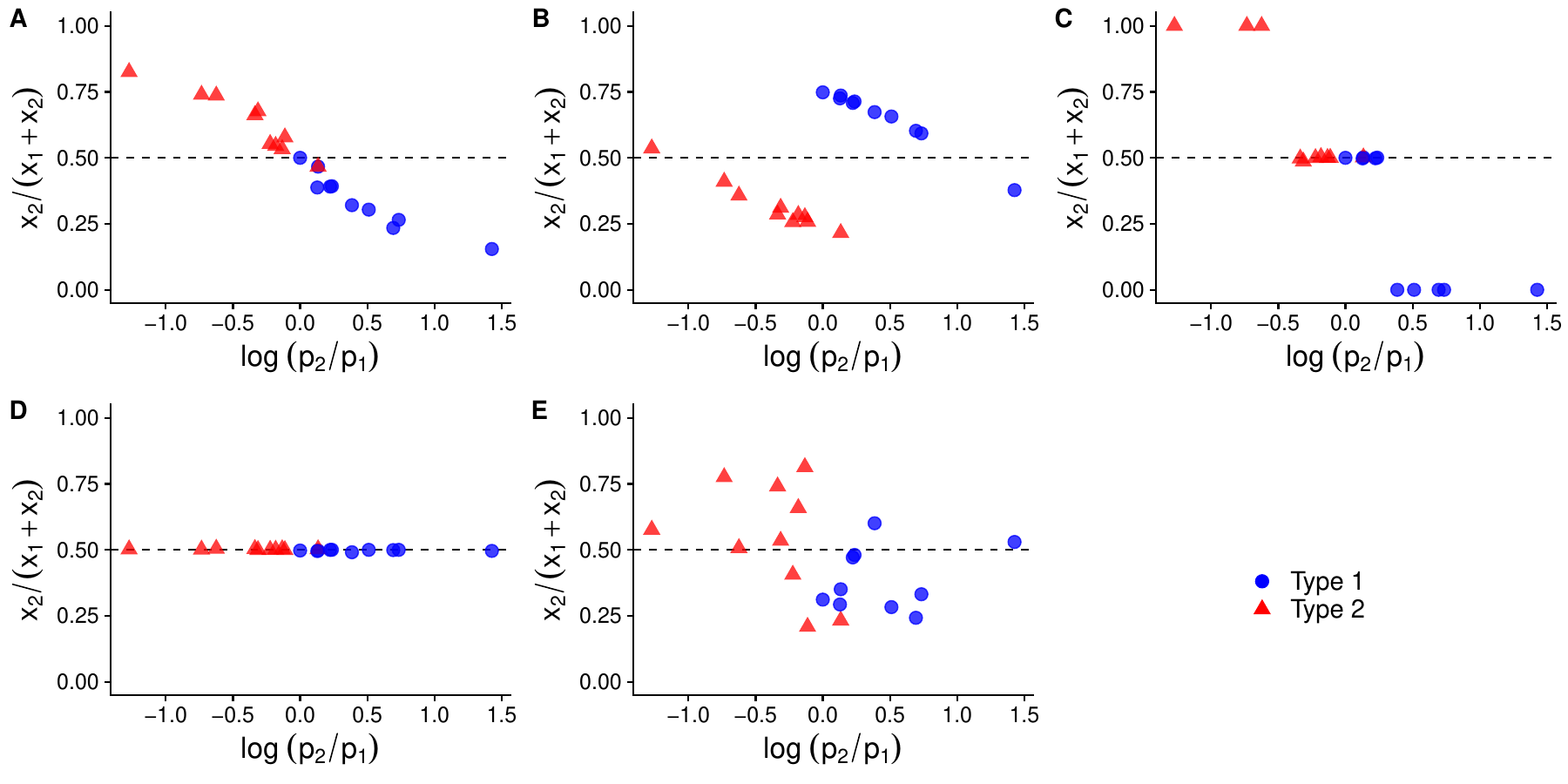}
\caption{The relationship between the log-price ratio $\log (p_2/p_1)$ and the allocation share $x_2/(x_1+x_2)$ for selected subjects.}
\label{fig:uci_sample_allocation}
\end{figure}

We proceed to ask the question: Did the subjects in our experiment make choices that are consistent with basic economic models of utility maximization, including the standard subjective expected utility (SEU) theory?

\paragraph{Revealed-preference tests.} 
We implement nonparametric, revealed-preference tests on each individual subject's choice data. These tests include: GARP, probabilistic sophistication \citep[hereafter PS;][]{machina1992}, SEU \citep[based on and extended from][]{echenique2015savage}, and MEU \citep[based on and extended from][]{chambers2016}. As discussed in Section~\ref{section:theoretical_background}, we can test whether a given dataset is consistent with SEU or MEU by solving the linear program implied by the axiom that characterizes the model. We say that a dataset {\em passes the test} if there is a solution to the problem. 

Recall that, depending on how we partition the state space, we have two types of decision problems. For GARP and PS, we first test each type of problem separately and then combine the results. We say that a subject's data satisfies GARP if it passes the GARP test for both types. Similarly, we say that a subject's data is not inconsistent with PS if it is not inconsistent with PS in the sense of \citeapos{epstein2000probabilities} condition for both types, and also satisfies event monotonicity. For SEU and MEU, we implement the test directly on the data combining the two types of problems. It is, at first glance, not obvious that this can be done. That the two types of problems can be combined, effectively testing the three-state design using bets on two events at a time, is one of the methodological contribution of our paper: see Online Appendix~\ref{appendix:implementation} for details.

\begin{table}[t]
\centering 
\caption{Pass rates.}
\label{table:uci_pass_rates}
\begin{tabular}{l cccc}
\toprule 
 & GARP & PS & SEU & MEU  \\
\midrule 
Market-stock & 0.5827 & 0.4803 & 0.0000 & 0.0000 \\ 
Market-Ellsberg & 0.6693 & 0.6220 & 0.0157 & 0.0157 \\ 
\bottomrule 
\end{tabular}
\caption*{\footnotesize {\em Notes}: $N = 127$. Since \citeapos{epstein2000probabilities} condition is only necessary for probabilistic sophistication, the numbers reported here capture the fraction of the subjects who are {\em not inconsistent} with probabilistic sophistication. Pass rates for each type of questions separately are presented in Table~\ref{table:uci_pass_rates_all} in the Online Appendix. The power of these revealed-preference tests are discussed in Online Appendix~\ref{appendix:power_calc}.} 
\end{table}

Table~\ref{table:uci_pass_rates} presents the {\em pass rate} of each test. That is, the fraction of subjects (out of~127) who passed each test. We find that a majority of subjects satisfy GARP, meaning that their choices are consistent with the maximization of {\em some} utility function. On the contrary, subjects clearly did not make choices that are consistent with SEU. The SEU pass rates are below~0.1, and not a single agent passed the SEU test in the market-stock task.\footnote{Along similar lines, \cite{echenique2018approximate} find that only five out of more than 3,000 participants in three online surveys \citep{carvalho2016poverty, carvalho2019complexity, choi2014who} make choices that are consistent with {\em objective} expected utility theory.}

Perhaps surprisingly, allowing for multiple priors via MEU does not change the result. Pass rates for MEU are the same as for SEU, implying that {\em MEU does not capture violations of SEU in our experiment.} These findings are consistent with data from the experiment in \cite{hey2014}: see \cite{chambers2016}, which performs the same kind of analysis as we do in the present paper for \citeapos{hey2014} data.

Finally, we look at PS to investigate whether observed behavior is (in)consistent with preferences being based on probabilities, using the necessary condition proposed by \cite{epstein2000probabilities} and checking event monotonicity in questions~\#5 and~\#6. We find that 48\% of subjects in the market-stock task and 62\% of subjects in the market-Ellsberg task are not inconsistent with PS. 

Testing for the exact compliance with the model may be too demanding. It is possible that small mistakes could account for a subjects' deviation from SEU or MEU. We now turn to quantifying the degree of compliance with the models, using CCEI and minimal~$e$ as described in Section~\ref{section:theoretical_background}.

\paragraph{Distance measures.}
The Critical Cost Efficiency Index (CCEI) is a measure of the degree of compliance with GARP that is widely used in the recent experimental literature \citep[e.g.,][]{choi2014who}. In our laboratory data, the average CCEI is above~0.98, which implies that on average budget lines needed to be shifted down by about two percent to eliminate a subject's GARP violations (Table~\ref{table:uci_distance}). The CCEI scores reported in Table~\ref{table:uci_distance} are substantially higher than those reported in \cite{choi2014who}, but close to the CCEI scores in \cite{choi2007}. This would seem to indicate a higher level of compliance with utility maximizing behavior than in the experiment by \cite{choi2014who}, and about the same as the experiment by \cite{choi2007}. 
Note, however, that there are several substantial differences in the settings and the designs between the two aforementioned studies and ours. We had two types of events (other studies typically have one fixed event structure), each type involved~10 budgets (i.e., total~20 budgets) while the cited studies had~25 and~50 budgets respectively. Most importantly, objective probabilities were not provided in our experiment. 

\begin{table}[t]
\centering 
\caption{Distance measures.}
\label{table:uci_distance}
\resizebox{\textwidth}{!}{%
\begin{tabular}{l ccc ccc ccc}
\toprule 
 & \multicolumn{3}{c}{CCEI} & \multicolumn{3}{c}{$e_*$ (SEU)} & \multicolumn{3}{c}{$e_*$ (MEU)} \\
 \cmidrule(lr){2-4} \cmidrule(lr){5-7} \cmidrule(lr){8-10}  
Task & Mean & Median & SD & Mean & Median & SD & Mean & Median & SD \\
\midrule 
Market-stock & 0.9805 & 1.0000 & 0.0450 & 1.5066 & 1.2791 & 0.9169 & 1.4949 & 1.2588 & 0.9224 \\ 
Market-Ellsberg & 0.9892 & 1.0000 & 0.0317 & 1.3094 & 1.0000 & 0.9108 & 1.3038 & 1.0000 & 0.9105 \\ 
\bottomrule 
\end{tabular}%
}
\end{table}

We use $e_*$ \citep[minimal~$e$; proposed by][]{echenique2018approximate} as a measure of the degree of deviation from SEU. Remember that the number $e_*$ is a perturbation to the model that allows SEU to accommodate the observed choices. It is zero when data are consistent with SEU, meaning that no perturbation is needed to rationalize the data by means of SEU, but takes a positive value if data violate SEU. The larger is $e_*$, the larger is the size of the perturbation needed to rationalize data by means of a perturbed version of SEU.

We find that $e_*$ in the market-stock task is significantly higher than in the market-Ellsberg task (paired-sample $t$-test; $t (126) = 2.635$, $p = 0.009$). See also Figure~\ref{fig:uci_e_comparison} panel~A. This finding suggests that subjects made choices that were closer to SEU when the source of information was an Ellsberg urn than when the source was a stock price, but the result has to be qualified because the order of the two market tasks was not counterbalanced. 

\begin{figure}[!t]
\centering 
\includegraphics[width=\textwidth]{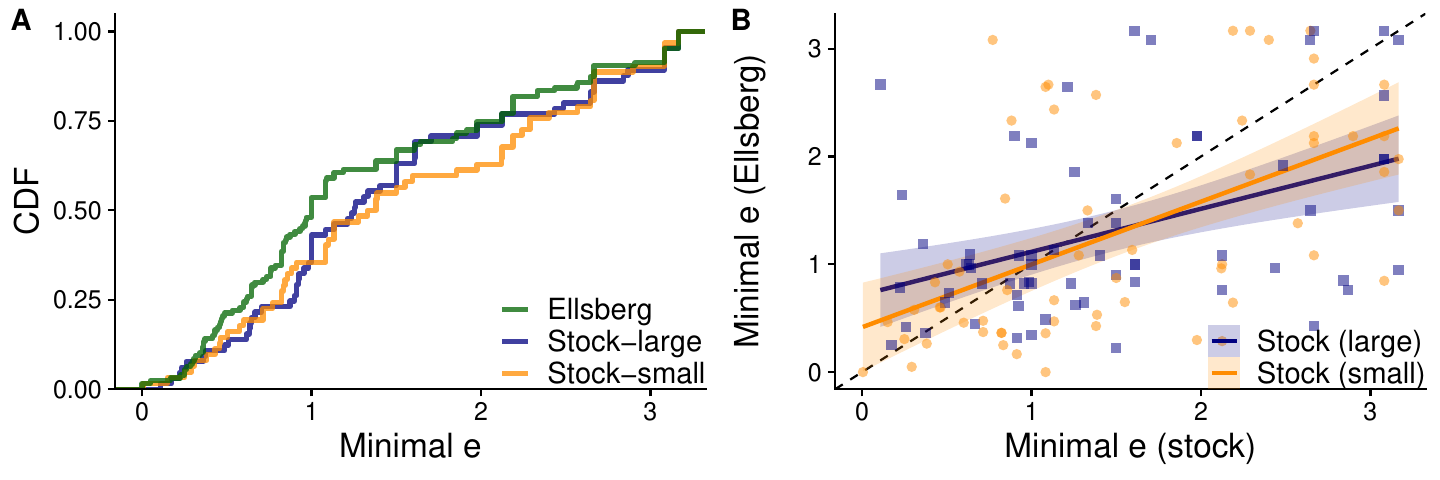}
\caption{$e_*$ from the market tasks. (A) Comparing $e_*$ from market tasks with different sources of uncertainty (Ellsberg urn, stock price with large volatility, and stock price with small volatility). (B) Correlation between $e_*$ from market-stock and market-Ellsberg tasks. Each dot represents a subject.}
\label{fig:uci_e_comparison}
\end{figure}

In the two market tasks, subjects faced the same set of~20 budgets in random order, with the exception of two budgets for which the order was fixed (see above). The choices made by about three-quarters of the subjects are positively correlated between the two tasks (Figure~\ref{fig:uci_allocation_corr_by_var} in the Online Appendix), and~36\% of those subjects exhibit statistically significant positive correlation (one-sided, at the~5\% significance level). 
This correlation is reflected in the degree of violation of SEU\textemdash{}Figure~\ref{fig:uci_e_comparison} panel~B shows that $e_*$ from the two tasks are highly correlated (Spearman's correlation coefficient: $r = 0.406$ for treatment Large, $r = 0.583$ for treatment Small). 

Table~\ref{table:uci_distance} also shows that the data is not much closer to MEU than to SEU. The MEU model has little added explanatory power beyond SEU. In other words, the way in which subjects' choices deviate from SEU is not captured by the MEU model. In  MEU, agents' beliefs can depend on choices, as in the perturbation of the SEU model behind our calculation of $e_*$. However, in MEU, the dependency is specific: beliefs are chosen so as to minimize expected utility. Our finding suggests that subjects' beliefs may depend on choices, but are not determined pessimistically. Therefore, the MEU model cannot explain the subjects' choices better than SEU; the size of perturbation required for MEU is not much lower than that for SEU.

We do not observe gender differences on $e_*$ but there is an effect of cognitive ability as measured with the three-item Cognitive Reflection Test \citep[CRT;][]{frederick2005cognitive}. Subjects who answered all three questions correctly exhibit lower $e_*$ than those who answered none of them correctly. This effect, however, is statistically significant only in the $e_*$ from the market-stock task (two-sample $t$-tests; Market-stock: $t(57)=1.50$, $p = 0.140$; Market-Ellsberg: $t(57)=3.24$, $p = 0.002$). See Figure~\ref{fig:uci_e_demographics} in the Online Appendix.

\paragraph{Ambiguity attitude.} 
Finally, we look at the relation between behavior in the market tasks and subjects' attitudes toward ambiguity, measured using a standard Ellsberg-paradox design. As explained in Section~\ref{section:design_tasks}, we asked three questions regarding choices between an ambiguous bet and a risky bet to identify subjects' attitude toward ambiguity. Figure~\ref{fig:uci_ellsberg_choice_freq} shows the frequency with which subjects preferred to bet on the risky urn, for each question. 

\begin{figure}[t]
\centering 
\includegraphics[width=\textwidth]{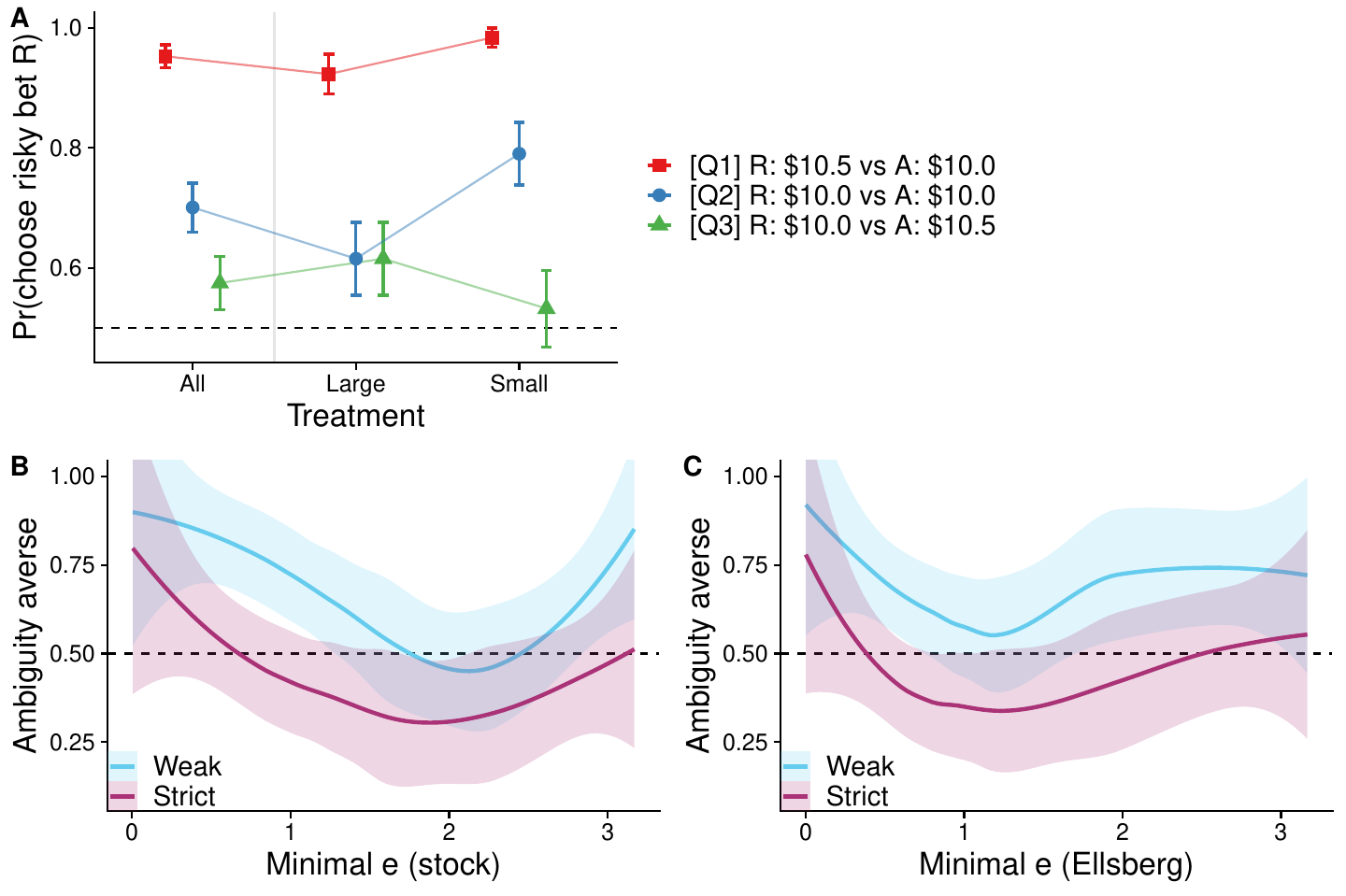}
\caption{(A) Frequency of choosing a risky bet in each question in the standard-Ellsberg task in the laboratory data. Bars indicate standard errors of means. (BC) LOESS curves relating $e_*$ and ambiguity attitude.}
\label{fig:uci_ellsberg_choice_freq}
\end{figure}

In the first question, the risky bet pays an additional \$0.5 in case of winning. This bonus made almost all (95.3\%) subjects choose the risky bet. The third question has instead a bonus for choosing the ambiguous bet, which then pays an additional \$0.5 in case of winning. A little more than half of the subjects (61.5\% in the Large treatment, 53.2\% in the Small treatment) preferred the risky bet, but the difference from 50\% (i.e., indifference at the aggregate level) is not significantly large ($z$-test for proportion; $p = 0.063$ in the Large treatment and $p = 0.612$ in the Small treatment). In the second question, which pays the equal winning prize in the two bets (as in many other Ellsberg-style studies), subjects in the Small treatment chose the risky bet more frequently than those in the Large treatment (61.5\% in the Large treatment and 73.0\% in the Small treatment; two-sample $z$-test for proportion, $p = 0.031$). 

We classify subjects as {\em weakly ambiguity averse} if they chose the risky bet, both in the first and in the second question (68.5\% of the subjects). Similarly, we classify subjects as {\em strictly ambiguity averse} if they chose the risky bet in all three questions (44.1\% of the subjects). In order to connect the deviation from SEU captured by $e_*$ and a measure of ambiguity attitude standard in the literature, we nonparametrically estimate how the probability of being classified as ambiguity averse depends on $e_*$. Figure~\ref{fig:uci_ellsberg_choice_freq}BC suggest a weak but quadratic relationship between these two. Ambiguity aversion is the leading explanation for violations of SEU, so our finding may seem counter-intuitive. One might instead expect a monotonic relation between $e_*$ and ambiguity aversion. It is, however, important to emphasize that $e_*$ captures {\em any} deviation from SEU. Not only those that could be traced to ambiguity aversion.

\subsection{Results from the Panel} 

A total of 764 subjects ($\text{mean age} = 50.26$, $\text{SD} = 15.45$; 50.4\% male) completed the study. 
The median survey length was 29.1 minutes. In addition to \$10 baseline payment for completing the survey, four randomly selected subjects received additional payment from one of the choices they made during the survey (average \$137.56). 

We tried to get subjects to do our experiment on a desktop or laptop computer, but many of them took it with their mobile devices---such as smartphones or tablets. These devices have screens that are smaller than desktop/laptop computers, which makes it quite difficult to understand our experiments, and perform the tasks we request them to complete. 
We thus analyze the data consisting of subjects who used desktop or laptop computer (66\%) as our ``core'' sample.  Table~\ref{table:uas_sociodemographic} provides distributions of individual sociodemographic characteristics in the entire sample as well as the core sample and the excluded sample (those who did not use desktop or laptop computers). It is evident that the type of device used is correlated with some of the demographic variables (age, education level, and income level; chi-squared tests in the last column in Table~\ref{table:uas_sociodemographic}). The sub-samples of subjects exhibited markedly different patterns of behavior as well (Online Appendix~\ref{appendix:additional_panel_sample_comparison}). 
Throughout the rest of the paper, we analyze data from the core sample.\footnote{Results from the same set of analyses on the entire subjects, or comparison across sub-samples, are available upon request.}

\begin{table}[!t]
\centering
\caption{Sociodemographic information.}
\label{table:uas_sociodemographic}
\resizebox{\textwidth}{!}{%
\begin{tabular}{l ccc r}
\toprule 
 & & \multicolumn{2}{c}{Device} \\
\cmidrule(l){3-4}
Variable & All & Desktop/laptop & Tablet/mobile phone & Test \\
\midrule 
\emph{Gender} \\
\hspace{3mm} Female & 0.496 & 0.471 & 0.544 & $\chi^2 (1) = 3.36$ \\ 
\hspace{3mm} Male & 0.504 & 0.529 & 0.456 & $p = 0.0669$ \\ 
\emph{Age group} \\
\hspace{3mm} 20-39 & 0.319 & 0.279 & 0.395 &  \\ 
\hspace{3mm} 40-59 & 0.353 & 0.345 & 0.369 & $\chi^2 (2) = 17.79$ \\ 
\hspace{3mm} 60- & 0.327 & 0.375 & 0.236 & $p = 0.0001$ \\ 
\emph{Education level} \\
\hspace{3mm} Less than high school & 0.258 & 0.190 & 0.388 \\
\hspace{3mm} Some college & 0.219 & 0.200 & 0.255 \\
\hspace{3mm} Assoc./professional degree & 0.187 & 0.200 & 0.163 & $\chi^2 (3) = 53.7$ \\ 
\hspace{3mm} College or post-graduate & 0.336 & 0.410 & 0.194 & $p < 0.0001$ \\ 
\emph{Household annual income} \\
\hspace{3mm} -- \$25k & 0.211 & 0.148 & 0.331 \\
\hspace{3mm} \phantom{-- }\$25k -- \$50k & 0.258 & 0.246 & 0.281 \\
\hspace{3mm} \phantom{-- }\$50k -- \$75k & 0.202 & 0.230 & 0.148 \\
\hspace{3mm} \phantom{-- }\$75k -- \$150k & 0.262 & 0.297 & 0.194 & $\chi^2 (4) = 43.97$ \\ 
\hspace{3mm} \phantom{-- }\$150k -- & 0.068 & 0.080 & 0.046 & $p < 0.0001$ \\ 
\emph{Occupation type} \\
\hspace{3mm} Full-time & 0.497 & 0.509 & 0.475 \\
\hspace{3mm} Part-time & 0.102 & 0.100 & 0.106 & $\chi^2 (2) = 0.78$ \\ 
\hspace{3mm} Not working & 0.401 & 0.391 & 0.418 & $p = 0.6759$ \\ 
\emph{Marital status} \\
\hspace{3mm} Married/live with partner & 0.690 & 0.713 & 0.646 & $\chi^2 (1) = 3.23$ \\ 
\hspace{3mm} Other & 0.310 & 0.287 & 0.354 & $p = 0.0724$ \\ 
\midrule 
\# of observations in the sample & 764 & 501 & 263 \\
\bottomrule
\end{tabular}%
}
\end{table}

The set of~20 budgets used in the market task is the 10-times scaled-up version of the one used in the laboratory (Figure~\ref{fig:design_path_and_budgets}, panel~B). This keeps the relative prices the same between two studies, making the distance measure $e_*$ comparable between data from the laboratory and the panel.

We start by checking event monotonicity, along the lines of our discussion for the laboratory experiment. Subjects' choices on questions~\#5 and~\#6 are informative about how attentive they are when they perform the tasks in our experiment. We find that about 60\% of subjects satisfy event monotonicity, and that this number jumps to 78\% if we allow for a  margin of error of five tokens (see Figure~\ref{fig:uas_monotonicity} in the Online Appendix). There is no treatment difference. Our subjects also made choices that are, to some extent, responsive to  underlying price changes: Figure~\ref{fig:uas_dsd} reports the degree of compliance with the downward-sloping demand property. 

\begin{figure}[t]
\centering 
\includegraphics[width=0.5\textwidth]{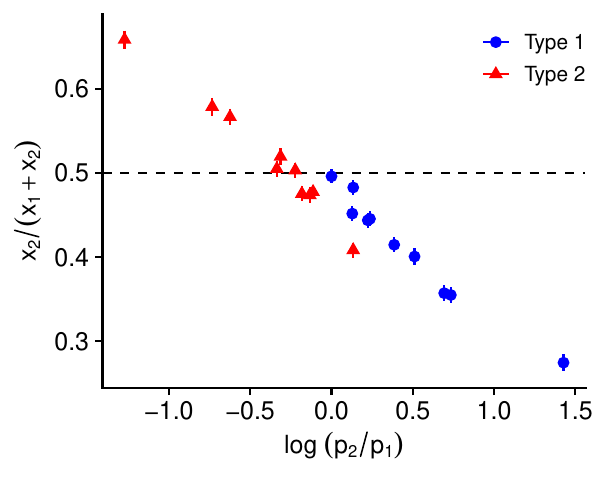}
\caption{Downward-sloping demand at the aggregate level. {\em Notes}: Each point represents mean $x_2/(x_1+x_2)$ at each $\log{(p_2 / p_1)}$ and bars indicate standard error of means.}
\label{fig:uas_dsd}
\end{figure}

\paragraph{Revealed-preference tests, distance measures, and ambiguity attitude.} 
Pass rates for GARP, SEU, and MEU presented in Table~\ref{table:uas_pass_rates} are similar to those of the laboratory data presented above. We find high GARP pass rates, but very low rates for SEU and MEU. Importantly, MEU again does not have more explanatory power than SEU: there is no room for additional rationalizations by allowing for multiple priors. Only one non-SEU subject is rationalized by MEU. High compliance with GARP pushes the average CCEI score to~$0.97$ (Table~\ref{table:uas_distance}). The average $e_*$ of~$1.610$ is not statistically significantly different from the average~$1.507$ in the laboratory data (two-sample $t$-test, $t (626) = 1.133$, $p = 0.258$).

\begin{table}[t]
\centering 
\caption{Pass rates.}
\label{table:uas_pass_rates}
\begin{tabular}{ll cccc}
\toprule 
Treatment & $N$ & GARP & PS & SEU & MEU \\
\midrule 
Large volatility & 245 & 0.4367 & 0.3959 & 0.0122 & 0.0122 \\ 
Small volatility & 256 & 0.4492 & 0.3945 & 0.0234 & 0.0273 \\ 
\midrule 
Combined & 501 & 0.4431 & 0.3952 & 0.0180 & 0.0200 \\ 
\bottomrule 
\end{tabular}
\caption*{\footnotesize {\em Notes}: Since \citeapos{epstein2000probabilities} condition is only necessary for probabilistic sophistication, the numbers reported here capture the fraction of the subjects who are {\em not inconsistent} with probabilistic sophistication. Pass rates for each type of questions separately are presented in Table~\ref{table:uas_pass_rates_all} in the Online Appendix. The power of these revealed-preference tests are discussed in Online Appendix~\ref{appendix:power_calc}.} 
\end{table}

\begin{table}[t]
\centering 
\caption{Distance measures.}
\label{table:uas_distance}
\resizebox{\textwidth}{!}{%
\begin{tabular}{ll ccc ccc ccc}
\toprule 
 & & \multicolumn{3}{c}{CCEI} & \multicolumn{3}{c}{$e_*$ (SEU)} & \multicolumn{3}{c}{$e_*$ (MEU)} \\
 \cmidrule(lr){3-5} \cmidrule(lr){6-8} \cmidrule(lr){9-11}  
Treatment & $N$ & Mean & Median & SD & Mean & Median & SD & Mean & Median & SD \\
\midrule 
Large volatility & 245 & 0.9720 & 0.9950 & 0.0509 & 1.6194 & 1.5369 & 0.9057 & 1.6097 & 1.5000 & 0.9107 \\ 
Small volatility & 256 & 0.9688 & 0.9958 & 0.0552 & 1.6002 & 1.4750 & 0.9243 & 1.5969 & 1.4750 & 0.9261 \\ 
\midrule 
Combined & 501 & 0.9704 & 0.9954 & 0.0531 & 1.6096 & 1.5000 & 0.9144 & 1.6032 & 1.5000 & 0.9177 \\ 
\bottomrule 
\end{tabular}%
}
\end{table}

The pattern of choices in the standard-Ellsberg task is also similar to what we observed in the laboratory data, but the overall frequency with which the risky bet is chosen is smaller. In particular, only~70\% of subjects (regardless of treatment) chose the risky bet in the first question, in which the risky bet pays a~\$5 more than the ambiguous bet in case of winning (note that almost everybody chose the risky bet in the laboratory, albeit with a reward magnitude that is $1/10$th of what we used in the panel). There are thus~44\% (26\%) of subjects who are weakly (strictly) ambiguity averse (Figure~\ref{fig:uas_ellsberg_choice_freq}). These numbers are lower than in the laboratory data. 
Now, using this classification, we look at the relationship between ambiguity aversion and $e_*$. Unlike Figure~\ref{fig:uci_ellsberg_choice_freq} panels~B and~C, using laboratory data, Figure~\ref{fig:uas_ellsberg_choice_freq} panel~A exhibits a decreasing relation between the two (there is a slight indication of reflection around $e_* \approx 1.2$, but it is not as strong as Figure~\ref{fig:uci_ellsberg_choice_freq}BC). Combining these two observations, we can see that subjects with small $e_*$ (close to SEU) are not necessarily less ambiguity averse in the standard Ellsberg task.

\begin{figure}[t]
\centering 
\includegraphics[width=\textwidth]{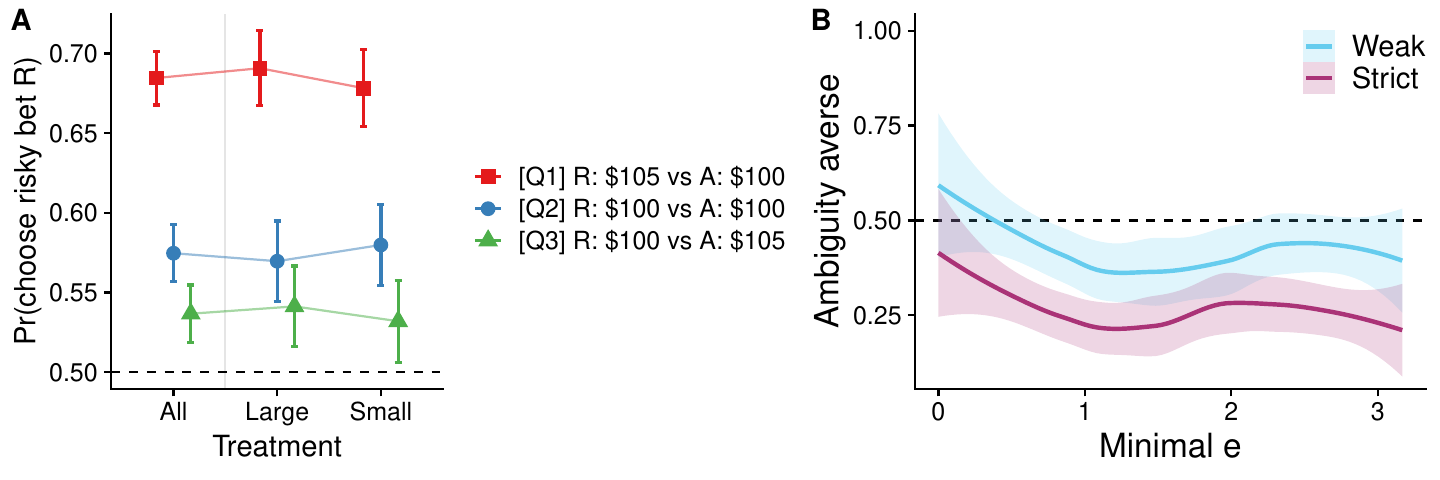}
\caption{(A) Probability of choosing a risky bet in each question in the standard-Ellsberg task in the panel data. (B) LOESS curves relating $e_*$ and ambiguity attitude.}
\label{fig:uas_ellsberg_choice_freq}
\end{figure}

\paragraph{Sociodemographic correlation.} 
One of the great advantages of using the UAS survey is that registered researchers can access datasets from past surveys, and use subject responses in related surveys and experiments. In particular, we use basic demographic information collected through the survey, as well as measures of cognitive ability, financial literacy, and other behavioral data from other experiments.\footnote{The cognitive ability measure is taken from survey module~\#1. Two financial literacy measures are taken from modules~\#1 and~\#6, which asked both the basic and sophisticated financial literacy questions in \cite{lusardi2017ordinary}. One caveat to this approach is the time lag between previous the surveys and ours. For example, the first survey module~\#1 was administered in May~2014.}

We estimated a linear model 
  \[ y_i = \boldsymbol{X}_i \boldsymbol{\beta} + \varepsilon_i , \]
where the dependent variable $y_i$ is subject~$i$'s value of $e_*$ or downward-sloping demand measured by correlation $\rho^\text{dsd}$ between $\log{(p_2/p_1)}$ and $x_2/(x_1+x_2)$, and $\boldsymbol{X}_i$ is a vector of sociodemographic characteristics. These explanatory variables include: age group (omitted category is ``20-39 years old''), above-median financial literacy (measured in UAS modules~\#1 and~\#6; omitted category is ``below-median score''), cognitive ability measured with CRT (omitted category is ``score is~0''), education level (omitted category is ``high school or less''), annual income group (omitted category is ``less than \$25,000''), gender, and employment status. The model is estimated by OLS with robust standard errors. We also estimate logistic regressions where the dependent variable $y_i$ is event monotonicity ($=1$ if monotonicity is violated with a margin of five tokens) and ambiguity attitude in the sense of standard Ellsberg ($=1$ if choices indicate weak ambiguity aversion).

\begin{table}[!p] 
\centering 
\caption{Relation between demographic characteristics and measures for several aspects of behavior in the experiment.} 
\label{table:uas_reg_e_demog} 
\resizebox*{!}{0.94\textheight}{%
\begin{tabular}{@{\extracolsep{0pt}}l D{.}{.}{-3} D{.}{.}{-3} D{.}{.}{-3} D{.}{.}{-3}} 
\toprule 
 & \multicolumn{2}{c}{OLS} & \multicolumn{2}{c}{logistic regression} \\ 
\cmidrule(lr){2-3} \cmidrule(lr){4-5} 
 & \multicolumn{1}{c}{(1)} & \multicolumn{1}{c}{(2)} & \multicolumn{1}{c}{(3)} & \multicolumn{1}{c}{(4)} \\ 
Dependent variable & \multicolumn{1}{c}{$e_*$} & \multicolumn{1}{c}{$\rho^\text{dsd}$} & \multicolumn{1}{c}{Violate mon.} & \multicolumn{1}{c}{Weak AA} \\ 
\midrule 
Treatment: Large & 0.032 & 0.026 & 0.003 & 0.182 \\ 
  & (0.083) & (0.030) & (0.233) & (0.198) \\ 
Age: 40-59 & -0.025 & -0.013 & -0.109 & -0.134 \\ 
  & (0.110) & (0.040) & (0.316) & (0.263) \\ 
Age: 60+ & 0.054 & -0.036 & 0.365 & -0.249 \\ 
  & (0.116) & (0.042) & (0.319) & (0.288) \\ 
Financial literacy (UAS \#1): High & 0.093 & 0.034 & -0.291 & 0.307 \\ 
  & (0.103) & (0.036) & (0.263) & (0.250) \\ 
Financial literacy (UAS \#6): High & -0.244 & -0.043 & 0.067 & 0.204 \\ 
  & (0.100) & (0.035) & (0.268) & (0.247) \\ 
CRT score (UAS \#1): 1 correct answer & -0.028 & -0.019 & -0.362 & 0.436 \\ 
  & (0.098) & (0.034) & (0.264) & (0.232) \\ 
CRT score (UAS \#1): 2 or 3 correct answers & -0.152 & -0.006 & -0.609 & 0.711 \\ 
  & (0.122) & (0.044) & (0.362) & (0.286) \\ 
Education: Some college & 0.112 & 0.005 & 0.142 & -0.070 \\ 
  & (0.133) & (0.047) & (0.342) & (0.331) \\ 
Education: Assoc. or pro. degree & -0.253 & -0.087 & -0.167 & -0.026 \\ 
  & (0.130) & (0.046) & (0.374) & (0.324) \\ 
Education: College or postgraduate & -0.019 & -0.032 & -0.478 & 0.574 \\ 
  & (0.121) & (0.043) & (0.346) & (0.299) \\ 
Income: 25,000-49,999 & 0.244 & 0.098 & -0.055 & 0.470 \\ 
  & (0.138) & (0.046) & (0.368) & (0.335) \\ 
Income: 50,000-74,999 & 0.418 & 0.126 & 0.635 & 0.071 \\ 
  & (0.142) & (0.048) & (0.374) & (0.353) \\ 
Income: 75,000-149,999 & 0.338 & 0.120 & -0.112 & 0.226 \\ 
  & (0.142) & (0.049) & (0.414) & (0.349) \\ 
Income: 150,000+ & 0.290 & 0.141 & -0.087 & 0.675 \\ 
  & (0.201) & (0.078) & (0.614) & (0.484) \\ 
Male & -0.129 & -0.034 & -0.130 & 0.277 \\ 
  & (0.087) & (0.030) & (0.247) & (0.210) \\ 
Working & 0.053 & 0.003 & -0.311 & -0.309 \\ 
  & (0.097) & (0.034) & (0.299) & (0.250) \\ 
Constant & 1.482 & -0.483 & -0.779 & -1.237 \\ 
  & (0.163) & (0.057) & (0.431) & (0.430) \\ 
\midrule
Observations & \multicolumn{1}{c}{490} & \multicolumn{1}{c}{490} & \multicolumn{1}{c}{490} & \multicolumn{1}{c}{490} \\ 
Adjusted $R^{2}$ / Log likelihood & \multicolumn{1}{c}{0.031} & \multicolumn{1}{c}{0.003} & \multicolumn{1}{c}{-239.470} & \multicolumn{1}{c}{-309.069} \\ 
\bottomrule
\end{tabular}%
}
\caption*{{\footnotesize {\em Notes}: Robust standard errors are presented in parentheses.}}
\end{table}

Regression results are presented in the first two columns of Table~\ref{table:uas_reg_e_demog}. First, there is no effect of age on $e_*$. Cognitive ability as measured by CRT is negatively associated with $e_*$ but the effect is not strong. The financial literacy variable measured in UAS module~\#6 is negatively correlated with $e_*$ (i.e., subjects with higher financial literacy are closer to SEU). Subjects in higher income brackets have larger $e_*$ (i.e., further away from SEU), compared to those in the lowest bracket in our sample. Educational background has an effect in the expected direction, but only in the category ``associate or professional degree,'' not in ``college or post-graduate degree.''\footnote{In contrast to these observations, \cite{echenique2018approximate} find that older subjects have larger $e_*$ for objective expected utility (OEU) (i.e., further away from OEU, not SEU) than younger subjects. This is a robust finding in the sense that it holds across data from three different panel surveys \citep{choi2014who, carvalho2016poverty, carvalho2019complexity}.} 
Demographic characteristics do not capture variation in the compliance with the downward-sloping demand property (column~2), but a similar effect of income is observed. Two other measures, violation event monotonicity and ambiguity attitude in the sense of Ellsberg, also exhibit non-significant association with demographic characteristics (except that high CRT score subjects tend to be ambiguity averse compared to low CRT score counterpart). 

Finally, we compare the distribution of $e_*$ in the laboratory and panel data. We can make this comparison because the same set of prices was used in the two experiments. Budgets were very different, but $e_*$ is about relative prices and not about budgets. It is evident from Figure~\ref{fig:uas_e_uas-uci} that there is no difference in distributions of $e_*$ ($p$-values from two-sample Kolmogorov-Smirnov tests are all larger than~0.1). As a basic check to compare that subjects' decisions are at least different than what random choices would offer, we compared the observed distributions to what purely random choices would give rise to: the two distributions are significantly different from the distribution of $e_*$ when simulated subjects make uniformly random choices ($p$-values from two-sample Kolmogorov-Smirnov tests are all below~0.01). 

\begin{figure}[t]
\centering 
\includegraphics[width=0.65\textwidth]{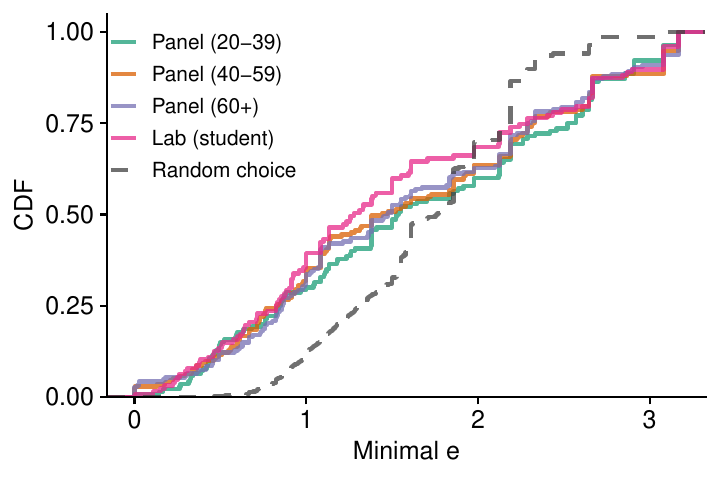}
\caption{Comparing distributions of $e_*$ from the panel study and the laboratory study.}
\label{fig:uas_e_uas-uci}
\end{figure}

\section{Conclusion}
\label{section:conclusion} 

Motivated by recent theoretical advances that provide revealed-preference characterizations of expected utility theory, we design and implement a novel experimental test of the theory of decision making under uncertainty. We find that subjects are mostly consistent with utility maximization, and respond to price changes in the expected direction: they satisfy the downward-sloping demand property, at least to some degree, but not enough to make their choices consistent with SEU. Our findings are the same, regardless of whether we look at laboratory or panel data. In fact, there is a striking similarity in how SEU is violated across the two studies. The subject populations are very different but look very similar in terms of the distribution of the degree of violation of SEU.

Motivated also by the literature on ambiguity aversion, we study the possibility that violations of SEU are due to ambiguity aversion, and look at whether maxmin expected utility (MEU) can explain the data. MEU adds no explanatory power to SEU: with a single exception, {\em all} subjects who fail to satisfy SEU also fail MEU. It is possible that other models of ambiguity aversion could do a better job of accounting for our experimental data. We are restricted to MEU because it is the only model for which there exist nonparametric tests of the kind that we use in our paper; it is also arguably the best known, and most widely applied, model in the ambiguity literature. The testable implications of other models of ambiguity-averse choice is an interesting direction for future research.

Finally, the results in our experiments are markedly unaffected by some of the demographic characteristics that other studies (on choice under risk, not uncertainty) have found significant. Older subjects do not seem to violate SEU to a larger degree than younger subjects. Neither do we see higher degrees of SEU violations in our broad sample of the U.S.\ population, compared to our laboratory experiment conducted on undergraduate students. There are modest effects of income and education. Financial literacy is correlated with subjects' distance to SEU. 

Further studies are necessary to fully understand the behavior in  environments that are more ``natural'' than traditional artificial Ellsberg-style settings. Our nonparametric revealed-preference tests and the empirical approach driven by these theories should hopefully be a useful tool to collect more evidence in this direction.

\bibliographystyle{ecta}
\bibliography{seu_exp}

\clearpage 
\setcounter{page}{1}
\appendix 

\numberwithin{equation}{section}
\counterwithin{figure}{section}
\counterwithin{table}{section}

\begin{center}
{\LARGE \bf Online Appendix}
\end{center}

\section{Implementation of Revealed Preference Tests}
\label{appendix:implementation} 

\subsection{Exact Revealed Preference Tests} 
\label{appendix:implementation_exact} 

We are able to check whether a given dataset is consistent with SEU or MEU by solving a linear programming problem that is equivalent to the axiom characterizing the model under consideration. 
The construction of linear programming problems closely follows the argument in the proofs of Theorems that appear in \citeSOMpapers{echenique2015savage} and \citeSOMpapers{chambers2016}. 

For example, \citeSOMpapers{echenique2015savage} prove in Lemma~7 that a dataset $(x^k, p^k)_{k=1}^K$ is SEU rational if and only if there are strictly positive numbers $v_s^k$, $\lambda^k$, and $\mu_s$ for $s = 1, \dots, S$ and $k = 1, \dots, K$ such that 
\[ \mu_s v_s^k = \lambda^k p_s^k , \;\; x_s^k > x_{s'}^{k'} \Longrightarrow v_s^k \leq v_{s'}^{k'} , \]
or equivalently,  
\[ \log v_s^k + \log \mu_s - \log \lambda^k - \log p_s^k = 0 , \;\; 
x_s^k > x_{s'}^{k'} \Longrightarrow \log v_s^k \leq \log v_{s'}^{k'} , \]
in a log-linearlized form. 

Testing SEU rationality boils down to checking for existence of a solution to the above system, which can be expressed as a system of linear equalities and inequalities: 
\begin{equation}\label{eq:rp_system}
\begin{cases}
A \cdot z = 0 \\
B \cdot z \geq 0 \\
E \cdot z > 0 
\end{cases} . 
\end{equation}

\paragraph{A system of linear equalities and inequalities.}
We now construct three key ingredients of the system~\eqref{eq:rp_system}, matrices $A$, $B$, and $E$ for testing SEU. 

The first matrix $A$ has $K \times S$ rows and $K \times S + S + K + 1$ columns, defined as follows: We have one row for every pair $(k, s)$; one column for every pair $(k, s)$; one column for every $s$, one column for each $k$; and one last column. 
In the row corresponding to $(k, s)$ the matrix has zeroes everywhere with the following exceptions: it has a $1$ in the column for $(k, s)$; it has a $1$ in the column for $s$; it has a $-1$ in the column for $k$; and $-\log p^k_s$ in the very last column. 
This finalizes the construction of $A$. 
The resulting matrix looks as follows: 
\begin{eqnarray*}
\kbordermatrix{
 & (1,1) & \cdots & (k,s) & \cdots & (K,S) &  & 1 & \cdots & s & \cdots & S & & 1 & \cdots & k & \cdots & K &  & p \\
(1,1) & 1 & \cdots & 0 & \cdots & 0 & \vrule & 1 & \cdots & 0 & \cdots & 0 & \vrule & -1 & \cdots & 0 & \cdots & 0 & \vrule & - \log p_1^1 \\
\vdots  & \vdots &  & \vdots &  & \vdots & \vrule & \vdots & & \vdots & & \vdots & \vrule & \vdots &  & \vdots &  & \vdots & \vrule & \vdots \\
(k,s) & 0 & \cdots & 1 & \cdots & 0 & \vrule & 0 & \cdots & 1 & \cdots & 0 & \vrule & 0 & \cdots & -1 & \cdots & 0 & \vrule & - \log p_s^k \\
\vdots  & \vdots &  & \vdots &  & \vdots & \vrule & \vdots & & \vdots & & \vdots & \vrule & \vdots &  & \vdots &  & \vdots & \vrule & \vdots \\
(K,S) & 0 & \cdots & 0 & \cdots & 1 & \vrule & 0 & \cdots & 0 & \cdots & 1 & \vrule & 0 & \cdots & 0 & \cdots & -1 & \vrule & - \log p_S^K \\
} . 
\end{eqnarray*}

Next, we construct matrix $B$ that has $K \times S + S + K + 1$ columns and there is one row for every pair $(k, s)$ and $(k', s')$ for which $x^k_s > x^{k'}_{s'}$. 
In the row corresponding to $x^k_s > x^{k'}_{s'}$ we have zeroes everywhere with the exception of a $-1$ in the column for $(k, s)$  and a $1$ in the column for $(k', s')$. 

Finally, we prepare a matrix that captures the requirement
that the last component of a solution be strictly positive. 
The matrix $E$ has a single row and $K \times S + S + K + 1$ columns. 
It has zeroes everywhere except for $1$ in the last column.

In order to test MEU, we need to modify matrices $A$, $B$, and $E$ appropriately, following the characterization provided by \citeSOMpapers{chambers2016} for the case of two states of the world. 
Let $K^0 = \{ k : x^k_1 = x^k_2 \}$, $K^1 = \{ k : x^k_1 < x^k_2 \}$ and $K^2 = \{ k : x^k_1 > x^k_2 \}$. 
The first-order conditions are: 
\[ \mu^k_s v^k_s = \lambda^k p^k_s , \] 
for $s = 1, 2$ and $k \in \{ 1, \dots, K \}$, where $\mu^k_1 = \bar{\mu}_1$ if $k \in K^1$, $\mu^k_1 = \underline{\mu}_1$ if $k \in K^2$, $\mu^k_1 \in [ \underline{\mu}_1, \bar{\mu}_1]$ if $k \in K^0$. 
We now use $\pi = \mu_1 / \mu_2$ instead of $\mu_1$. Then we can rewrite the first-order conditions: 
\[ \pi^k v^k_1 = \lambda^k p^k_1 \;\; \text{and} \;\; v^k_2 = \lambda^k p^k_2 , \] 
for $k\in\{1,\ldots,K\}$, where $\pi^k = \bar{\pi}$ if $k\in K^1$, $\pi^k = \underline{\pi}$ if $k \in K^2$, $\pi^k \in [\underline{\pi}, \bar{\pi}]$ if $k \in K^0$.

Let $A$ be a matrix with $2 K + 2 + | K^0 | + K + 1$ columns. The first $2 K$ columns are labeled with a different pair $(k, s)$. The next two columns are labeled $\bar{\pi}$ and $\underline{\pi}$. The next $|K^0|$ columns are for choices on the 45-degree line. The next $K$ columns are labeled with $k$. Finally the last column is labeled $p$. 

For each $(k, 2)$, $A$ has a row with all zero entries with the following exception: It has a $1$ in the column labeled $(k, 2)$; It has a $-1$ in the column labeled $k$; It has $- \log p_s^k$ in the column labeled $p$. 
For each $(k, 1)$ with $k \in K^1$, $A$ has a row with all zero entries with the following exception: It has a $1$ in the column labeled $(k, 1)$; It has a $-1$ in the column labeled $k$; it has $- \log p_s^k$ in the column labeled $p$; It has a $1$ in the column labeled $\bar{\pi}$. 
For each $(k, 1)$ with $k \in K^2 \cup K^0$, $A$ has a row defined as above. 
The only difference is that it has a~$1$ in the column labeled $\underline{\pi}$ if $k \in K^2$ and in the column labeled $\pi^k$ if $k \in K^0$, instead of having a~$1$ in the column labeled $\bar{\pi}$. 

The resulting matrix $A$ looks as follows: 
\begin{eqnarray*}
\kbordermatrix{
 & (1,1) & \cdots & (k,s) & \cdots & (K,S) & & \bar{\pi} & \underline{\pi} & \pi^k & \cdots & & 1 & \cdots & k & \cdots & K &  & p \\
\vdots  & \vdots &  & \vdots &  & \vdots & \vrule & \vdots & \vdots & \vdots & & \vrule & \vdots & & \vdots & & \vdots & \vrule & \vdots \\
(k,s) \in (K,2)  & 0 & \cdots  & 1 & \cdots & 0 & \vrule & 0 & 0 & 0 & \cdots & \vrule & 0 & \cdots & -1 & \cdots & 0 & \vrule & - \log p_s^k \\
(k,s) \in (K^1,1)  & 0 & \cdots  & 1 & \cdots & 0 & \vrule & 1 & 0 & 0 & \cdots & \vrule & 0 & \cdots & -1 & \cdots & 0 & \vrule & - \log p_s^k \\
(k,s) \in (K^2,1)  & 0 & \cdots  & 1 & \cdots & 0 & \vrule & 0 & 1 & 0 & \cdots & \vrule & 0 & \cdots & -1 & \cdots & 0 & \vrule & - \log p_s^k \\
(k,s) \in (K^0,1)  & 0 & \cdots  & 1 & \cdots & 0 & \vrule & 0 & 0 & 1 & \cdots & \vrule & 0 & \cdots & -1 & \cdots & 0 & \vrule & - \log p_s^k \\
\vdots  & \vdots &  & \vdots &  & \vdots & \vrule & \vdots & \vdots & \vdots & & \vrule & \vdots & & \vdots & & \vdots & \vrule & \vdots \\
} . 
\end{eqnarray*}

Let $B$ be a matrix with the same number of columns as $A$. 
The columns of $B$ are labeled like those of $A$. 
$B$ has a row for each pair $(x_s^k, x_{s'}^{k'})$ with $x_s^k > x_{s'}^{k'}$. 
The row for $x_s^k > x_{s'}^{k'}$ has all zeroes except for a $1$ in column $(k', s')$ and a $-1$ in column $(k, s)$. 
Finally, $B$ has a row which has a~$1$ in the column for $\bar{\pi}$ and a $-1$ in the column for $\underline{\pi}$ and $2 |K^0|$ additional rows to capture $\pi^k \in [\underline{\pi}, \bar{\pi}]$ for $k \in K^0$. 
The resulting matrix $B$ looks as follows: 
\begin{eqnarray*}
\kbordermatrix{
 & (1,1) & \cdots & (k,s) & \cdots & (k',s') & \cdots & (K,S) & & \bar{\pi} & \underline{\pi} & \pi^k & \cdots & & 1 & \cdots & k & \cdots & K &  & p \\
\vdots  & \vdots &  & \vdots &  & \vdots &  & \vdots & \vrule & \vdots & \vdots & \vdots & & \vrule & \vdots & & \vdots & & \vdots & \vrule & \vdots \\
x_s^k > x_{s'}^{k'}  & 0 & \cdots  & -1 & \cdots & 1 & \cdots & 0 & \vrule & 0 & 0 & 0 & \cdots & \vrule & 0 & \cdots & 0 & \cdots & 0 & \vrule & 0 \\
\vdots  & \vdots &  & \vdots &  & \vdots &  & \vdots & \vrule & \vdots & \vdots & \vdots & \cdots & \vrule & \vdots & & \vdots & & \vdots & \vrule & \vdots \\
\bar{\pi} \geq \underline{\pi} & 0 & \cdots  & 0 & \cdots & 0 & \cdots & 0 & \vrule & 1 & -1 & 0 & \cdots & \vrule & 0 & \cdots & 0 & \cdots & 0 & \vrule & 0 \\
\pi^k \geq \underline{\pi} & 0 & \cdots  & 0 & \cdots & 0 & \cdots & 0 & \vrule & 0 & -1 & 1 & \cdots & \vrule & 0 & \cdots & 0 & \cdots & 0 & \vrule & 0 \\
\pi^k \leq \bar{\pi} & 0 & \cdots  & 0 & \cdots & 0 & \cdots & 0 & \vrule & 1 & 0 & -1 & \cdots & \vrule & 0 & \cdots & 0 & \cdots & 0 & \vrule & 0 \\
} . 
\end{eqnarray*}

\paragraph{Solve the system.}
Our task is to check if there is a vector $z$ that solves the system~\eqref{eq:rp_system} of linear inequalities corresponding to model $\text{M} \in \{ \text{SEU}, \text{MEU} \}$. 
If there is a solution $z$ to this system, we say that the dataset is ``M rational.''

\paragraph{Extension.} 
There are three underlying states of the world, $S = \{ \omega_1, \omega_2, \omega_3 \}$, in the experiments. 
There are two types of questions: in type~1 questions, two events are $s_1 = \{ \omega_1 \}$ and $s_{23} = \{ \omega_2, \omega_3 \}$; in type~2 questions, two events are $s_{12} = \{ \omega_1, \omega_2 \}$ and $s_3 = \{ \omega_3 \}$. 
Let $S_1 = \{ s_1, s_{23} \}$ denote the set of events in type~1 questions and $S_2 = \{ s_{12}, s_3 \}$ denote the set of events in type~2 questions. 
See Figure~\ref{fig:state_spece_structure_in_exp}. 

Suppose we have $K$ observations in the data $(x^k, p^k)_{k=1}^K$. 
Let $k_1 \in K_1$ and $k_2 \in K_2$ denote indices for type~1 and type~2 questions, respectively (thus $K = K_1 \cup K_2$). 
Note that there is no type~1 question with state $s_3$. 
Therefore, indices for observations ($k$) and states ($s$) need to be consistent, i.e., $(k, s) \in K_i \times S_i$ for each $i = 1, 2$.

\begin{figure}[!t]
\centering 
\begin{tikzpicture}[scale=1, every node/.style={scale=1},
node distance=2.1cm,
textbox/.style={rectangle, minimum width=2cm, minimum height=0.75cm, inner sep=8pt},
state1/.style={rectangle, minimum width=2cm, minimum height=0.75cm, inner sep=8pt, text centered, draw=black, fill=myPaleBlue},
state2/.style={rectangle, minimum width=2cm, minimum height=0.75cm, inner sep=8pt, text centered, draw=black, fill=myPaleYellow},
state3/.style={rectangle, minimum width=2cm, minimum height=0.75cm, inner sep=8pt, text centered, draw=black, fill=myPaleRed},
event1/.style={rectangle, minimum width=2cm, minimum height=0.75cm, inner sep=5pt, text centered, draw=black, fill=gray!20},
event2/.style={rectangle, minimum width=4.1cm, minimum height=0.75cm, inner sep=5pt, text centered, draw=black, fill=gray!20},
]
\node (s1) [state1, align=left] {$\omega_1$};
\node (s2) [state2, right of=s1] {$\omega_2$};
\node (s3) [state3, right of=s2] {$\omega_3$};
\node (e11) [event1, above of=s1, yshift=-1.15cm] {$s_1$};
\node (e12) [event2, above of=s3, xshift=-1.05cm, yshift=-1.15cm] {$s_{23}$};
\node (e21) [event2, below of=s1, xshift=1.05cm, yshift=1.15cm] {$s_{12}$};
\node (e22) [event1, below of=s3, yshift=1.15cm] {$s_3$};
\node (text0) [textbox, left of=s1, xshift=-0.7cm] {State of the world};
\node (text1) [textbox, above of=text0, xshift=-0cm, yshift=-1.15cm] {Type~1 question};
\node (text2) [textbox, below of=text0, xshift=-0cm, yshift=1.15cm] {Type~2 question};
\end{tikzpicture}
\caption{Event structure in the experiment.}
\label{fig:state_spece_structure_in_exp}
\end{figure}

In order to test SEU in this environment, we use the following proposition. 

\begin{proposition}\label{prop:seu_test_in_exp}
There exist strictly positive numbers $\mu_{1}$, $\mu_{23}$, $\mu_{12}$, $\mu_{3}$, $v_1^{k_1}$, $v_{23}^{k_2}$, $v_{12}^{k_1}$, $v_3^{k_1}$, $\lambda^{k_1},\lambda^{k_2}$ such that 
\begin{align}
\mu_{1} v^{k_1}_{1} = \lambda^{k_2} p^{k_1}_{1} \;\; \text{for each} \;\; k_1, \tag{FOC1} \label{seu-foc1} \\
\mu_{23} v^{k_1}_{23} = \lambda^{k_2} p^{k_1}_{23} \;\; \text{for each} \;\; k_1, \tag{FOC2} \label{seu-foc2} \\
\mu_{12} v^{k_2}_{12} = \lambda^{k_1} p^{k_2}_{12} \;\; \text{for each} \;\; k_2, \tag{FOC3} \label{seu-foc3} \\
\mu_{3} v^{k_2}_{3} = \lambda^{k_1} p^{k_2}_{3} \;\; \text{for each} \;\; k_2, \tag{FOC4} \label{seu-foc4} \\
\mu_{12} + \mu_3 = 1, \tag{NO1} \label{seu-no1} \\
\mu_{1} + \mu_{23} = 1, \tag{NO2} \label{seu-no2} \\
\mu_{12} \geq \mu_1, \tag{MO1} \label{seu-mo1} \\
\mu_{23} \geq \mu_3, \tag{MO2} \label{seu-mo2} \\
\mu_{12}-\mu_{1} = \mu_{23}-\mu_3. \tag{EQ} \label{seu-eq}
\end{align}
if and only if there exist strictly positive numbers $\tilde{\mu}_{23}$, $\tilde{\mu}_{12}$, $v_1^{k_1}$, $v_{23}^{k_2}$, $v_{12}^{k_2}$, $v_3^{k_2}$, $\tilde{\lambda}^{k_1}, \tilde{\lambda}^{k_2}$ such that 
\begin{align}
v^{k_1}_{1} = \tilde{\lambda}^{k_1} p^{k_1}_{1} \;\; \text{for each} \;\; k_1, \tag{FOC1$'$} \label{seu-foc1'} \\
\tilde{\mu}_{23} v^{k_1}_{23} = \tilde{\lambda}^{k_1} p^{k_1}_{23} \;\; \text{for each} \;\; k_1, \tag{FOC2$'$} \label{seu-foc2'} \\
\tilde{\mu}_{12} v^{k_2}_{12} = \tilde{\lambda}^{k_2} p^{k_2}_{12} \;\; \text{for each} \;\; k_2, \tag{FOC3$'$} \label{seu-foc3'} \\
v^{k_2}_{3} = \tilde{\lambda}^{k_2} p^{k_2}_{3} \;\; \text{for each} \;\; k_2, \tag{FOC4$'$} \label{seu-foc4'} \\
\tilde{\mu}_{23} \tilde{\mu}_{12} \geq 1. \tag{MO1$'$} \label{seu-mo1'}
\end{align}
\end{proposition}
\begin{proof}
Define $\mu_1 = 1 / (1 + \tilde{\mu}_{23})$, $\mu_{23} = \tilde{\mu}_{23} / (1 + \tilde{\mu}_{23})$, $\mu_{12} = \tilde{\mu}_{12} / (1 + \tilde{\mu}_{12})$, $\mu_3 = 1 / (1 + \tilde{\mu}_{12})$, $\lambda^{k_1} = \tilde{\lambda}^{k_1} / (1 + \tilde{\mu}_{23})$, and $\lambda^{k_2} = \tilde{\lambda}^{k_2} / (1 + \tilde{\mu}_{12})$. 
Then, conditions \ref{seu-foc1}-\ref{seu-foc1} are equivalent to \ref{seu-foc1'}-\ref{seu-foc4'} since, for example, 
\begin{equation*}
\mu_{12} v^{k_2}_{12} = \lambda^{k_2} p^{k_2}_{12} 
\;\; \iff \;\;
\frac{\tilde{\mu}_{12}}{1 + \tilde{\mu}_{12}} v^{k_2}_{12} = \frac{\tilde{\lambda}^{k_2}}{1 + \tilde{\mu}_{12}} p^{k_2}_{12} 
\;\; \iff \;\;
\tilde{\mu}_{12} v^{k_2}_{12} = \tilde{\lambda}^{k_2} p^{k_2}_{12} . 
\end{equation*}
Condition \ref{seu-mo1} is equivalent to \ref{seu-mo1'} since 
\begin{equation*}
\mu_{12} \geq \mu_1 
\;\; \iff \;\; 
\frac{\tilde{\mu}_{12}}{1 + \tilde{\mu}_{12}} \geq \frac{1}{1 + \tilde{\mu}_{23}} 
\;\; \iff \;\; 
\tilde{\mu}_{12} + \tilde{\mu}_{12} \tilde{\mu}_{23} \geq 1 + \tilde{\mu}_{12} 
\;\; \iff \;\; 
\tilde{\mu}_{12} \tilde{\mu}_{23} \geq 1 , 
\end{equation*}
and similarly \ref{seu-mo2} is equivalent to \ref{seu-mo1'}. 
Condition \ref{seu-eq} is satisfied since 
\begin{equation*}
\mu_{12} - \mu_1 
= \frac{\tilde{\mu}_{12}}{1 + \tilde{\mu}_{12}} - \frac{1}{1 + \tilde{\mu}_{23}} 
= \frac{\tilde{\mu}_{12} \tilde{\mu}_{23} - 1}{(1 + \tilde{\mu}_{12}) (1 + \tilde{\mu}_{23})} 
= \frac{\tilde{\mu}_{23}}{1 + \tilde{\mu}_{23}} - \frac{1}{1 + \tilde{\mu}_{12}} 
= \mu_{23} - \mu_3 . 
\end{equation*}
\end{proof}

In order to implement the test, we first assemble matrices $A$ (capturing $v$, $\mu$, and $\lambda$; equality constraints in the linear programming problem) and $B$ (capturing concavity of $u$; weak inequality constraints). 
The above proposition has two implications: 
(i) We need to find only two strictly positive numbers capturing subjective beliefs, $\tilde{\mu}_{23}$ and $\tilde{\mu}_{12}$, instead of four numbers $\mu_1$, $\mu_{23}$, $\mu_{12}$, and $\mu_3$. 
(ii) We need to add one row in $B$ to take care of additional weak inequality constraint $\tilde{\mu}_{23} \tilde{\mu}_{12} \geq 1$ (or equivalently, $\log \tilde{\mu}_{23} + \log \tilde{\mu}_{12} \geq 0$). 

Let us now consider MEU. 
Suppose that in the first $m_1$ observations we have a partition $\{ \{ \omega_1 \}, \{ \omega_2, \omega_3 \} \}$ (i.e., type~1 questions), and in the second $m_2$ observations we have a partition $\{ \{ \omega_1, \omega_2 \}, \{ \omega_3 \} \}$  (i.e., type~2 questions). 
\begin{itemize}
\item Partition 1: $\{ \{ \omega_1 \}, \{ \omega_2, \omega_3 \} \}$. 
Let $O^0 = \{ k : x^k_{1} = x^k_{23} \}$, $O^1 = \{ k : x^k_{1} < x^k_{23} \}$ and $O^2 = \{ k : x^k_{1} > x^k_{23} \}$. 
Afriat inequalities are now: 
\[ \theta^k v^k_{1} = \lambda^k p^k_{1} \;\; \text{and} \;\; v^k_{23} = \lambda^k p^k_{23}, \] 
for $k \in \{ 1, \dots, m_1 \}$, where $\theta^k = \bar{\theta}$ if $k \in O^1$, $\theta^k = \underline{\theta}$ if $k \in O^2$, and $\theta^k \in [\underline{\theta}, \bar{\theta}]$ if $k \in O^0$. 

\item Partition 2: $\{ \{ \omega_1, \omega_2 \}, \{ \omega_3 \} \}$. 
Let $T^0 = \{ k : x^k_{12} = x^k_3\}$, $T^1 = \{ k : x^k_{12} < x^k_3\}$ and $T^2 = \{ k : x^k_{12} > x^k_3 \}$. 
Afriat inequalities are now: 
\[ \pi^k v^k_{12} = \lambda^k p^k_{12} \;\; \text{and} \;\; v^k_3 = \lambda^k p^k_3, \] 
for $k \in \{ m_1+1, \dots, m_1+m_2 \}$, where $\pi^k = \bar{\pi}$ if $k \in T^1$, $\pi^k = \underline{\pi}$ if $k \in T^2$, and $\pi^k \in [\underline{\pi}, \bar{\pi}]$ if $k \in T^0$. 
\end{itemize}

The unknowns are
\[ \underline{\theta}, \bar{\theta}, \theta^k, v^k_s, \lambda^k \] 
for all $k = 1, \dots, m_1$, and $s \in \{ \{ \omega_1 \}, \{ \omega_2, \omega_3 \} \}$, and 
\[ \underline{\pi}, \bar{\pi}, \pi^k, v^k_s, \lambda^k \] 
for all $k = m_1+1, \dots, m_1+m_2$, $s \in \{ \{ \omega_1, \omega_2 \}, \{ \omega_3 \} \}$. 
The system of inequalities are: 
\begin{equation*}
\begin{cases}
\theta^k v^k_{1} = \lambda^k p^k_{1} & \text{if } k \in O^0 \\
\bar{\theta} v^k_{1} = \lambda^k p^k_{1} & \text{if } k \in O^1 \\
\underline{\theta} v^k_{1} = \lambda^k p^k_{1} & \text{if } k \in O^2 \\
v^k_{23} = \lambda^k p^k_{23} & \text{if } k \in O^0 \cup O^1 \cup O^2 \\
\underline{\theta} \leq \theta^k \leq \bar{\theta} &\text{if } k \in O^0 \cup O^1 \cup O^2
\end{cases} , \;\;\;\; 
\begin{cases}
\pi^k v^k_{12} = \lambda^k p^k_{12} & \text{if } k \in T^0 \\
\bar{\pi} v^k_{12} = \lambda^k p^k_{12} & \text{if } k \in T^1 \\
\underline{\pi} v^k_{12} = \lambda^k p^k_{12} & \text{if } k \in T^2 \\
v^k_3 = \lambda^k p^k_3 & \text{if } k \in T^0 \cup T^1 \cup T^2 \\
\underline{\pi} \leq \pi^k \leq \bar{\pi} & \text{if } k \in T^0 \cup T^1 \cup T^2 \\
\end{cases}
\end{equation*}
and, in addition, the constraints $\underline{\theta} \leq \bar{\theta}$, $\underline{\pi} \leq \bar{\pi}$, $\underline{\theta} \leq \underline{\pi}$, and $\bar{\theta} \leq \bar{\pi}$. 

Note that this system of inequalities is linear after we take the $\log$ of each variable. In particular the constraint that  $\pi^k \in [\underline{\pi}, \bar{\pi}]$ is written as $\log{(\underline{\pi})} \leq \log{(\pi^k)} \leq \log{(\bar{\pi})}$. 

\begin{proposition}
A solution to the previous Afriat inequalities gives a solution to the FOCs. 
\end{proposition}
\begin{proof}
Define 
\begin{equation*}
\mu^k_{1} = \theta^k / (1+\theta^k) , \;\; 
\mu^k_{23} = 1 / (1+\theta^k) 
\end{equation*}
if $k = 1, \dots, m_1$, and 
\begin{equation*}
\mu^{k'}_{12} = \pi^{k'} / (1+\pi^{k'}) , \;\; 
\mu^{k'}_{3} = 1 / (1+\pi^{k'}) 
\end{equation*}
if $k' = m_1+1, \dots, m_1+m_2$. 
Then 
\[ 1 = \mu^k_{1} + \mu^k_{23} = \mu^{k'}_{12} + \mu^{k'}_{3} . \]

Observe that:
\begin{enumerate}[label=(\alph*),ref=(\alph*)]
\item $\underline{\theta} \leq  \bar{\theta} \Longrightarrow \underline{\mu}_{1} \leq \bar{\mu}_{1}$; 
\item $\underline{\pi} \leq \bar{\pi} \Longrightarrow \underline{\mu}_{12} \leq \bar{\mu}_{12}$; 
\item \label{it:mon1} $\underline{\theta} \leq  \underline{\pi} \Longrightarrow \underline{\mu}_{1} \leq \underline{\mu}_{12}$; and 
\item \label{it:mon2} $\bar{\theta} \leq \bar{\pi} \Longrightarrow \bar{\mu}_{1} \leq \bar{\mu}_{12}$. 
\end{enumerate}

Define $\underline{\mu}_2 = \underline{\mu}_{12} - \underline{\mu}_1$. Note that $\underline{\mu}_2 = \bar{\mu}_{23} - \bar{\mu}_3$ because 
\[ 1 = \underline{\mu}_{12} + \bar{\mu}_3 = \underline{\mu}_{1} + \bar{\mu}_{23} \Longrightarrow \underline{\mu}_{12} - \underline{\mu}_{1} = \bar{\mu}_{23} - \bar{\mu}_3 . \]
Note also that $\underline{\mu}_2 \geq 0$, as~\ref{it:mon1} implies that $\underline{\mu}_{12} \geq \underline{\mu}_1$. 
Similarly, if we define $\bar{\mu}_2 = \bar{\mu}_{12} - \bar{\mu}_1$. Then using~\ref{it:mon2} we obtain that 
\[ 0 \leq \bar{\mu}_2 = \underline{\mu}_{23} - \underline{\mu}_3 . \]
\end{proof}

\subsection{Approximate Revealed Preference Tests} 
\label{appendix:implementation_approximate} 

We can characterize $e$-price-perturbed SEU (Fact~\ref{fact:seu_approx_test} in Section~\ref{section:theoretical_background}) in the environment of the experiment. 

\begin{proposition}
\label{prop:p-seu_1}
Given $e \in \mathbf{R}_+$, a dataset $(x^k, p^k)_{k=1}^K$ is {\em $e$-price-perturbed SEU rational} if and only if there exist strictly positive numbers $v_s^k$, $\lambda^k$, $\mu_s$, and $\varepsilon_s^k$ such that: 
\begin{enumerate}
\item for all $(k, s), (k', s') \in \cup_{i = 1}^2 (K_i \times S_i)$, 
\begin{equation*}
\mu_s v_s^k = \lambda^k \varepsilon_s^k p_s^k 
, \quad 
x_s^k > x_{s'}^{k'} \Longrightarrow v_s^k \leq v_{s'}^{k'} , 
\end{equation*}

\item for all $i = 1, 2$, 
\begin{equation*}
\sum_{s \in S_i} \mu_s = 1 ,   
\end{equation*}

\item $\mu_{12} \geq \mu_1$ and $\mu_{23} \geq \mu_3$, 

\item $\mu_{12} - \mu_1 = \mu_{23} - \mu_3$, and 

\item for all $k, l \in K$ and $s, t \in S$, 
\begin{equation*}
\frac{\varepsilon_s^k / \varepsilon_t^k}{\varepsilon_s^l / \varepsilon_t^l} \leq 1 + e . 
\end{equation*}
\end{enumerate}
\end{proposition}

Following the same reasoning as Proposition~\ref{prop:seu_test_in_exp}, we obtain an equivalent characterization which is useful for setting up a linear programming problem. 

\begin{proposition} 
\label{prop:p-seu_2}
Given $e \in \mathbf{R}_+$, a dataset $(x^k, p^k)_{k=1}^K$ is {\em $e$-price-perturbed SEU rational} if and only if there exist strictly positive numbers $v_s^k$, $\tilde{\lambda}^k$, $\tilde{\mu}_s$, and $\varepsilon_s^k$ such that: 
\begin{enumerate}
\item for all $k_1 \in K_1$, 
\begin{equation}
v_1^{k_1} = \tilde{\lambda}^{k_1} \varepsilon_1^{k_1} p_1^{k_1} , 
\quad \text{and} \quad 
\tilde{\mu}_{23} v_{23}^{k_1} = \tilde{\lambda}^{k_1} \varepsilon_{23}^{k_1} p_{23}^{k_1} , 
\tag{$e$-FOC1} 
\label{eq:rSEU-foc-1}
\end{equation}

\item for all $k_2 \in K_2$, 
\begin{equation}
\tilde{\mu}_{12} v_{12}^{k_2} = \tilde{\lambda}^{k_2} \varepsilon_{12}^{k_2} p_{12}^{k_2} 
\quad \text{and} \quad 
v_3^{k_2} = \tilde{\lambda}^{k_2} \varepsilon_3^{k_2} p_3^{k_2} , 
\tag{$e$-FOC2} 
\label{eq:rSEU-foc-2}
\end{equation}

\item for all $(k, s), (k', s') \in \cup_{i = 1}^2 (K_i \times S_i)$, 
\begin{equation*}
x_s^k > x_{s'}^{k'} \Longrightarrow v_s^k \leq v_{s'}^{k'} , 
\tag{$e$-CON} 
\label{eq:rSEU-con}
\end{equation*}

\item $\tilde{\mu}_{23} \tilde{\mu}_{12} \geq 1$ ($e$-MON), and 

\item for all $k, l \in K$ and $s, t \in S$, 
\begin{equation*}
\frac{\varepsilon_s^k / \varepsilon_t^k}{\varepsilon_s^l / \varepsilon_t^l} \leq 1 + e . 
\end{equation*}
\end{enumerate}
\end{proposition}

\begin{definition} \normalfont 
Given a dataset $(x^k, p^k)_{k=1}^K$, {\it minimal~$e$} for SEU, $e_*^\text{SEU}$, is the solution to the following problem: 
\begin{equation}
\min_{(\tilde{\mu}_s, v_s^k, \tilde{\lambda}^k, \varepsilon_s^k)_{k, s}} \max_{k, l \in K, s, t \in S} \frac{\varepsilon_s^k / \varepsilon_t^k}{\varepsilon_s^l / \varepsilon_t^l} 
\tag{$\bigstar$}
\label{eq:e_problem}
\end{equation}
subject to \eqref{eq:rSEU-foc-1}, \eqref{eq:rSEU-foc-2}, \eqref{eq:rSEU-con}, and ($e$-MON).  
\end{definition}

\begin{remark} \normalfont 
Notice that in the objective function of the problem~\eqref{eq:e_problem}, 
\begin{equation*}
\frac{\varepsilon_s^k / \varepsilon_t^k}{\varepsilon_s^l / \varepsilon_t^l} , 
\end{equation*}
two states $s, t$ are fixed and observations $k, l$ are different. 
In our experimental setup, it means that either $k, l \in K_1$ or $k, l \in K_2$. 
\end{remark}

\begin{remark} \normalfont 
\label{remark:log-linear_obj}
By log-linearizing and substituting equality constraints in the objective function in the problem~\eqref{eq:e_problem}, we obtain 
\begin{equation*}
\begin{aligned}
\log \left( \frac{\varepsilon_s^k / \varepsilon_t^k}{\varepsilon_s^l / \varepsilon_t^l} \right) 
&= \log \varepsilon_s^k - \log \varepsilon_t^k - \log \varepsilon_s^l + \log \varepsilon_t^l \\
&= (\log \tilde{\mu}_s + \log v_s^k - \log \lambda^k - \log p_s^k) - (\log \tilde{\mu}_t + \log v_t^k - \log \lambda^k - \log p_t^k) \\
&\phantom{==} - (\log \tilde{\mu}_s + \log v_s^l - \log \lambda^l - \log p_s^l) + (\log \tilde{\mu}_t + \log v_t^l - \log \lambda^l - \log p_t^l) \\
&= (\log \tilde{\mu}_s + \log v_s^k - \log p_s^k) - (\log \tilde{\mu}_t + \log v_t^k - \log p_t^k) \\
&\phantom{==} - (\log \tilde{\mu}_s + \log v_s^l - \log p_s^l) + (\log \tilde{\mu}_t + \log v_t^l - \log p_t^l) \\
&= (\log v_s^k - \log p_s^k) - (\log v_t^k - \log p_t^k) - (\log v_s^l - \log p_s^l) + (\log v_t^l - \log p_t^l) . 
\end{aligned}
\end{equation*}
\end{remark}

\clearpage 
\section{Supplementary Figures and Tables} 
\label{appendix:additional_results} 

\subsection{Laboratory Data} 
\label{appendix:additional_lab} 

\begin{figure}[!h]
\centering 
\includegraphics[width=0.6\textwidth]{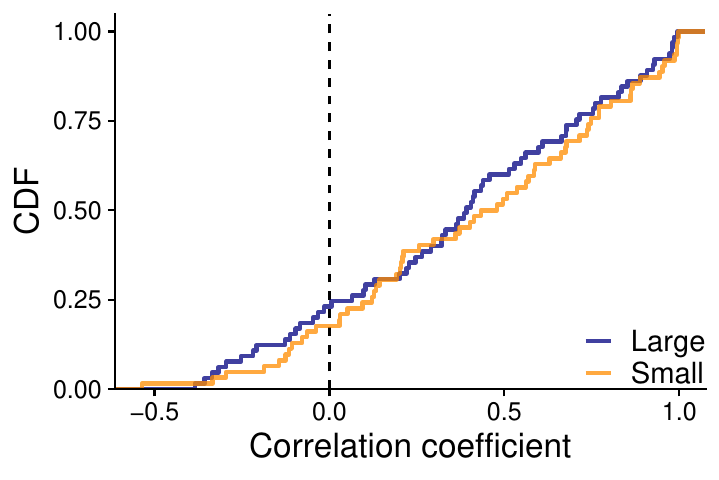}
\caption{Empirical CDF of within-subject (Spearman's) correlation between allocations in the market-stock task and the market-Ellsberg task. The two distributions are not significantly different (two-sample Kolmogorov-Smirnov test, $p = 0.91$).}
\label{fig:uci_allocation_corr_by_var}
\end{figure}

\begin{figure}[!h]
\centering 
\includegraphics[width=0.8\textwidth]{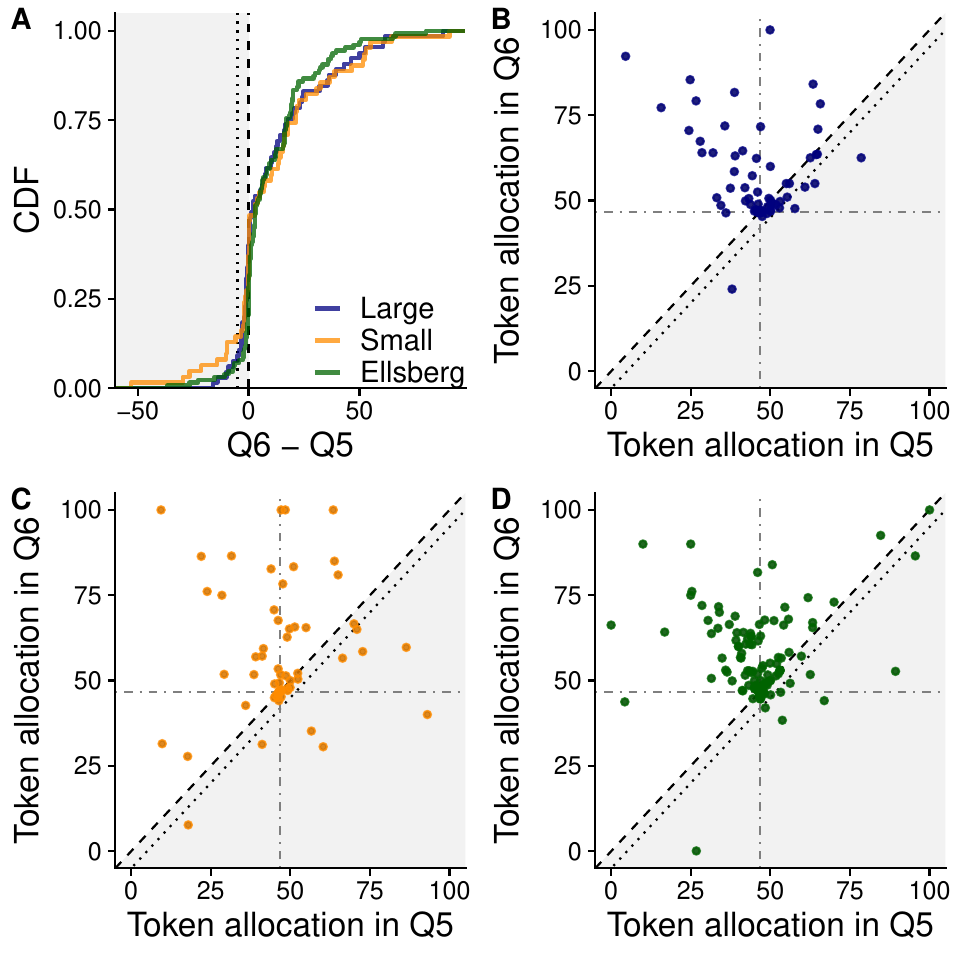}
\caption{Event monotonicity. (A) Empirical CDFs of token allocation difference. The dotted line represents a 5-token margin. No two pairs of distributions is significantly different (two-sample Kolmogorov-Smirnov test). (B-D) Token allocations in two questions. The dot-dashed lines at $46.67$ indicate the number of tokens which equalizes payouts in two events.}
\label{fig:uci_monotonicity}
\end{figure}

\begin{figure}[!h]
\centering 
\includegraphics[width=\textwidth]{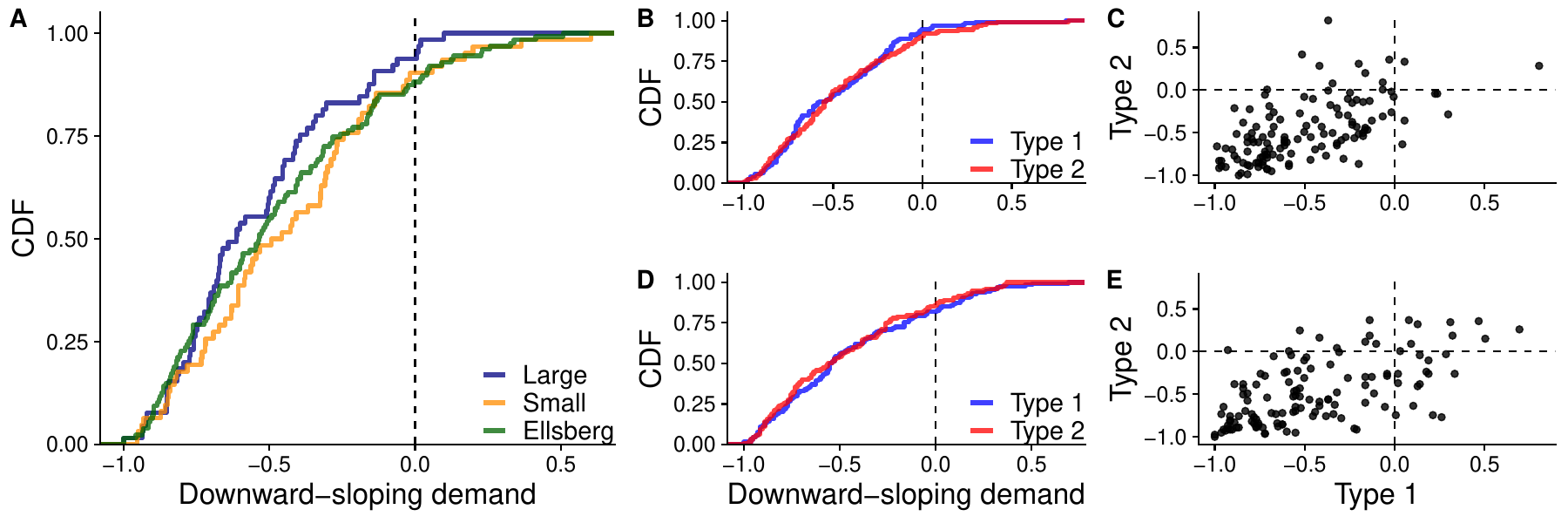}
\caption{Downward-sloping demand at the individual level (measured by $\rho^\text{dsd}$). (A) Empirical CDFs comparing across tasks. (BC) Comparing two question types in market-stock. (DE) Comparing two question types in market-Ellsberg.}
\label{fig:uci_dsd_individual}
\end{figure}

\begin{table}[!p]
\centering 
\caption{Pass rates. (C.f. Table~\ref{table:uci_pass_rates}.)}
\label{table:uci_pass_rates_all}
\begin{tabular}{l rrr rrr}
\toprule 
 & \multicolumn{3}{c}{GARP} & \multicolumn{3}{c}{PS} \\
 \cmidrule(lr){2-4} \cmidrule(lr){5-7} 
Task & Type~1 & Type~2 & Joint & Type~1 & Type~2 & Joint \\
\midrule 
Market-stock & 0.7638 & 0.6850 & 0.5827 & 0.7323 & 0.8110 & 0.4803 \\ 
Market-Ellsberg & 0.8268 & 0.7480 & 0.6693 & 0.8110 & 0.8346 & 0.6220 \\ 
\bottomrule 
\end{tabular}

\begin{tabular}{l rrr rrr}
\toprule 
 & \multicolumn{3}{c}{SEU} & \multicolumn{3}{c}{MEU} \\
 \cmidrule(lr){2-4} \cmidrule(lr){5-7}
Task & Type~1 & Type~2 & Joint & Type~1 & Type~2 & Joint \\
\midrule 
Market-stock & 0.0472 & 0.0157 & 0.0000 & 0.0472 & 0.0157 & 0.0000 \\ 
Market-Ellsberg & 0.0787 & 0.0315 & 0.0157 & 0.0787 & 0.0315 & 0.0157 \\ 
\bottomrule 
\end{tabular}
\caption*{\footnotesize {\em Note}: A subject satisfies GARP ``jointly'' if the subject passes GARP for both types. A subject is not inconsistent with PS ``jointly'' if the subject is not inconsistent with PS in the sense of Epstein for both types, and satisfies event monotonicity. Since \citeSOMpapersapos{epstein2000probabilities} condition is only necessary for probabilistic sophistication, the numbers reported here capture the fraction of the subjects who are not inconsistent with probabilistic sophistication.} 
\end{table}

\begin{figure}[!h]
\centering 
\includegraphics[width=\textwidth]{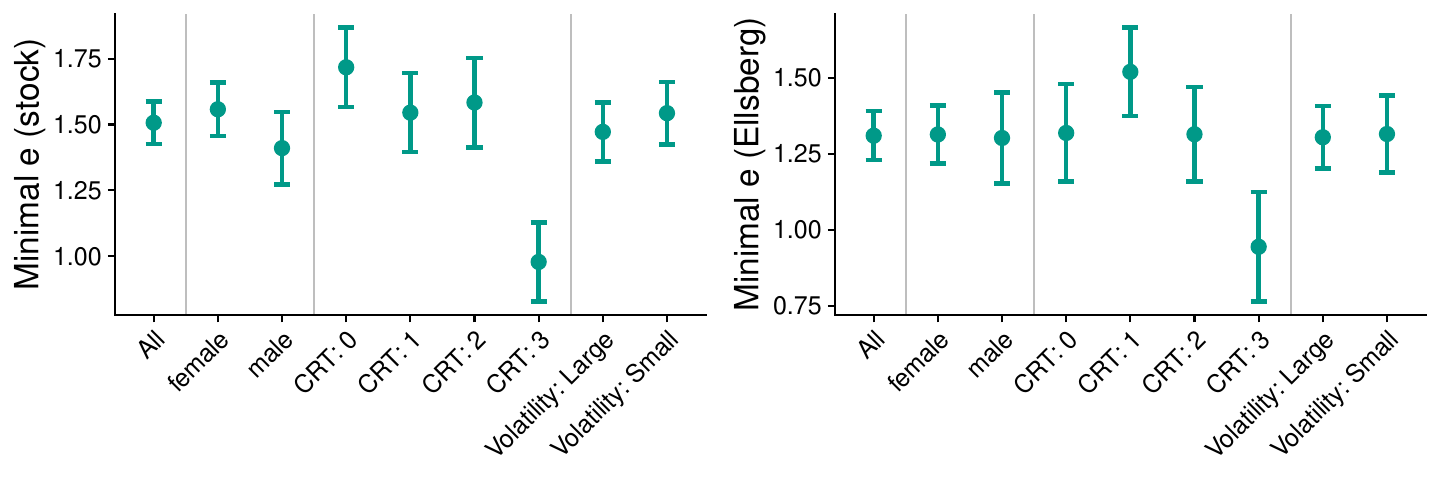}
\caption{Average $e_*$ for SEU in each group of subjects. {\em Notes}: Bars indicate standard errors of means.}
\label{fig:uci_e_demographics}
\end{figure}

\clearpage 
\subsection{Panel Data} 
\label{appendix:additional_panel} 

\begin{figure}[!h]
\centering 
\includegraphics[width=0.8\textwidth]{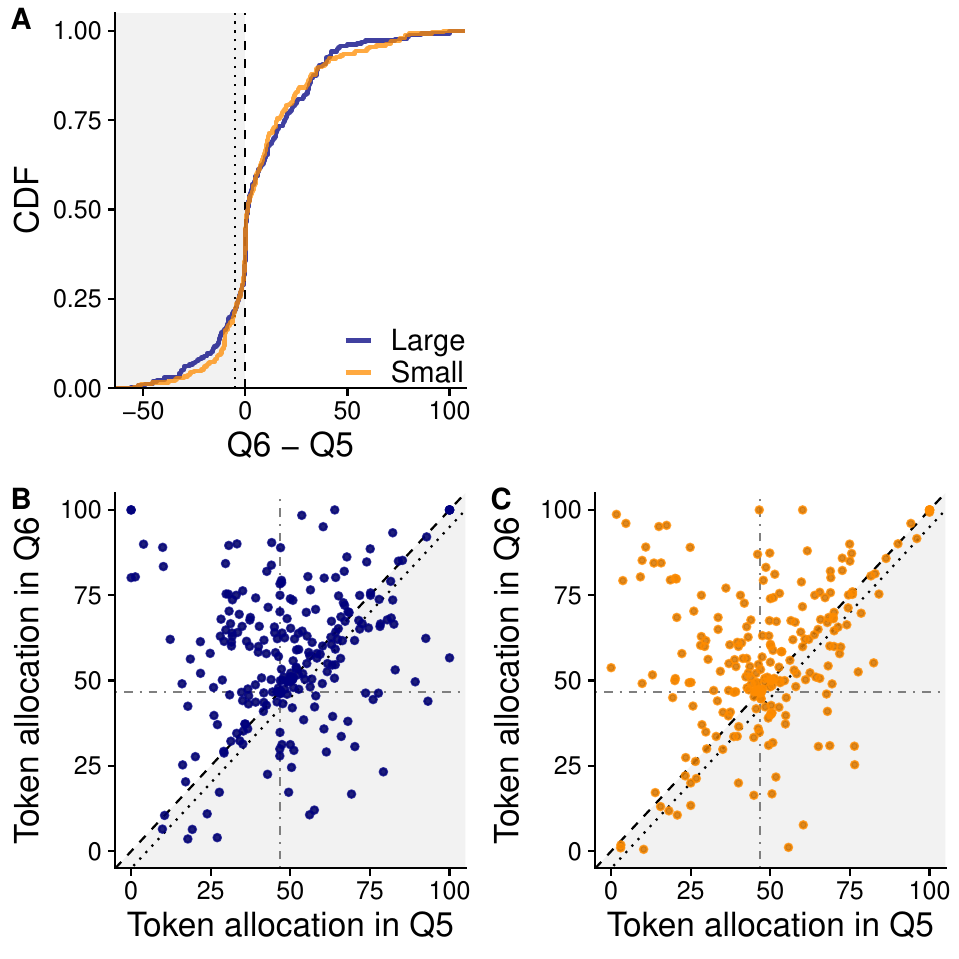}
\caption{Event monotonicity. (A) Empirical CDFs of token allocation difference. The dotted line represents a 5-token margin. No two pairs of distributions are significantly different (two-sample Kolmogorov-Smirnov test). (B-C) Token allocations in two questions. The dot-dashed lines at $46.67$ indicate the number of tokens which equalizes payouts in two events.}
\label{fig:uas_monotonicity}
\end{figure}

\begin{figure}[!h]
\centering 
\includegraphics[width=0.6\textwidth]{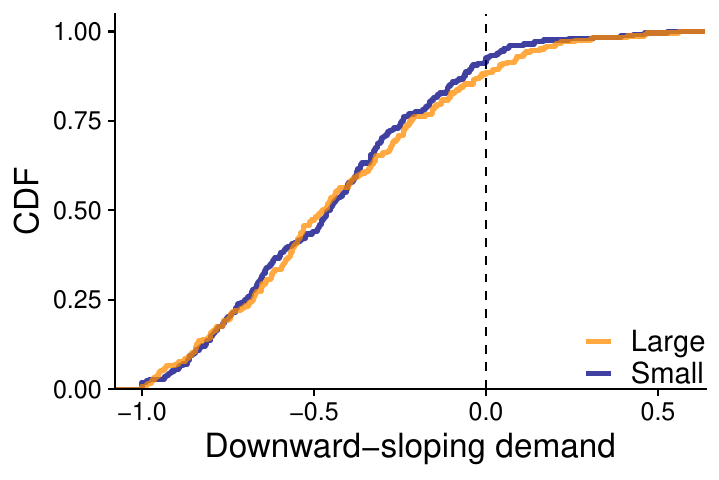}
\caption{Downward-sloping demand at the individual level (measured by $\rho^\text{dsd}$).}
\label{fig:uas_dsd_individual}
\end{figure}

\begin{figure}[!h]
\centering 
\includegraphics[width=0.8\textwidth]{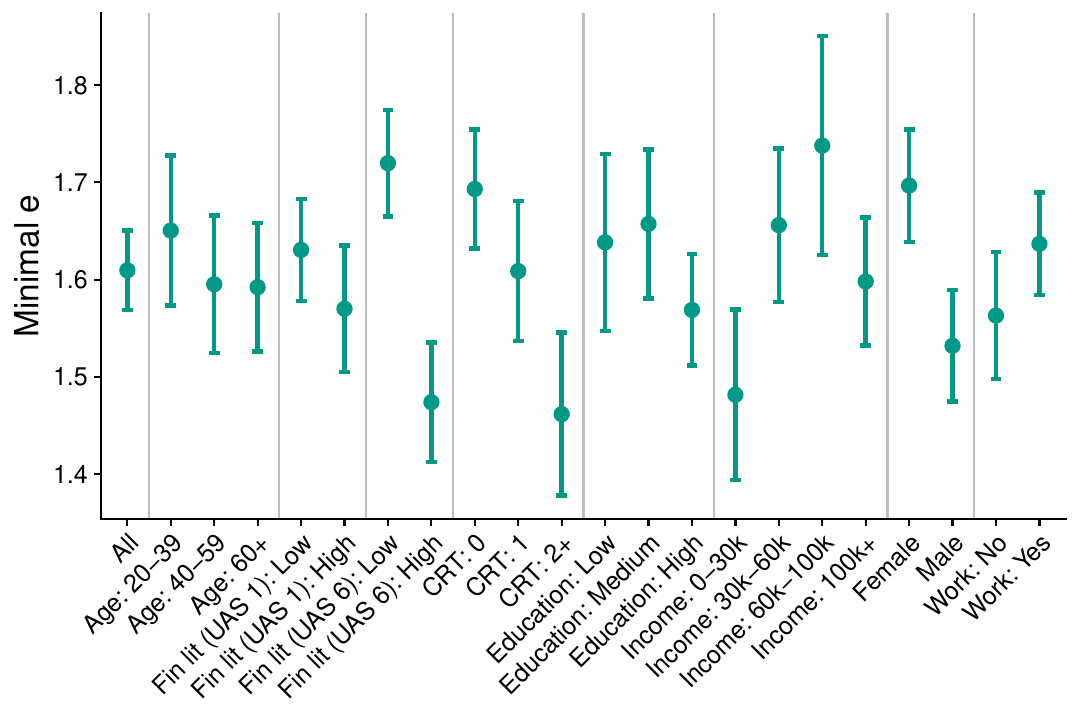}
\caption{Average $e_*$ for SEU in each group of subjects. {\em Notes}: Bars indicate standard errors of means.}
\label{fig:uas_e_demographics}
\end{figure}

\begin{table}[!h]
\centering 
\caption{Pass rates. (C.f. Table~\ref{table:uas_pass_rates}.)}
\label{table:uas_pass_rates_all}
\begin{tabular}{ll rrr rrr rrr rrr}
\toprule 
 & & \multicolumn{3}{c}{GARP} & \multicolumn{3}{c}{PS} \\
 \cmidrule(lr){3-5} \cmidrule(lr){6-8}
Treatment & $N$ & Type~1 & Type~2 & Joint & Type~1 & Type~2 & Joint \\
\midrule 
Large volatility & 245 & 0.6653 & 0.5918 & 0.4367 & 0.6776 & 0.8041 & 0.3959 \\ 
Small volatility & 256 & 0.6523 & 0.6016 & 0.4492 & 0.6680 & 0.7812 & 0.3945 \\ 
\midrule 
Combined & 501 & 0.6587 & 0.5968 & 0.4431 & 0.6727 & 0.7924 & 0.3952 \\ 
\bottomrule
\end{tabular}

\begin{tabular}{ll rrr rrr rrr rrr}
\toprule 
 & & \multicolumn{3}{c}{SEU} & \multicolumn{3}{c}{MEU} \\
 \cmidrule(lr){3-5} \cmidrule(lr){6-8}
Treatment & $N$ & Type~1 & Type~2 & Joint & Type~1 & Type~2 & Joint \\
\midrule 
Large volatility & 245 & 0.0735 & 0.0490 & 0.0122 & 0.0735 & 0.0490 & 0.0122 \\ 
Small volatility & 256 & 0.0508 & 0.0547 & 0.0234 & 0.0508 & 0.0547 & 0.0273 \\ 
\midrule 
Combined & 501 & 0.0619 & 0.0519 & 0.0180 & 0.0619 & 0.0519 & 0.0200 \\ 
\bottomrule
\end{tabular}
\caption*{\footnotesize {\em Note}: A subject satisfies GARP ``jointly'' if the subject passes GARP for both types. A subject is not inconsistent with PS ``jointly'' if the subject is not inconsistent with PS in the sense of Epstein for both types, and satisfies event monotonicity. Since \citeSOMpapersapos{epstein2000probabilities} condition is only necessary for probabilistic sophistication, the numbers reported here capture the fraction of the subjects who are not inconsistent with probabilistic sophistication.} 
\end{table}

\begin{table}[!t]
\centering
\caption{Sociodemographic information. Treatment assignments (small volatility or large volatility) are balanced.}
\label{table:uas_sociodemographic_by_treatment}
\begin{tabular}{l ccc r}
\toprule 
 & & \multicolumn{2}{c}{Treatment} \\
\cmidrule(l){3-4}
Variable & All & Small volatility & Large volatility & Test \\
\midrule 
\emph{Gender} \\
\hspace{3mm} Female & 0.496 & 0.487 & 0.505 & $\chi^2 (1) = 0.19$ \\ 
\hspace{3mm} Male & 0.504 & 0.513 & 0.495 & $p = 0.6666$ \\ 
\emph{Age group} \\
\hspace{3mm} 20-39 & 0.319 & 0.317 & 0.322 &  \\ 
\hspace{3mm} 40-59 & 0.353 & 0.366 & 0.340 & $\chi^2 (2) = 0.62$ \\ 
\hspace{3mm} 60+ & 0.327 & 0.317 & 0.338 & $p = 0.7342$ \\ 
\emph{Education level} \\
\hspace{3mm} Less than high school & 0.258 & 0.256 & 0.261 &  \\ 
\hspace{3mm} Some college & 0.219 & 0.235 & 0.202 &  \\ 
\hspace{3mm} Assoc./professional degree & 0.187 & 0.189 & 0.186 & $\chi^2 (3) = 1.51$ \\ 
\hspace{3mm} College or post-graduate & 0.336 & 0.320 & 0.351 & $p = 0.6806$ \\ 
\emph{Household annual income} \\
\hspace{3mm} -- \$25k & 0.211 & 0.247 & 0.173 &  \\ 
\hspace{3mm} \phantom{-- }\$25k -- \$50k & 0.258 & 0.245 & 0.271 &  \\ 
\hspace{3mm} \phantom{-- }\$50k -- \$75k & 0.202 & 0.196 & 0.207 &  \\ 
\hspace{3mm} \phantom{-- }\$75k -- \$150k & 0.262 & 0.240 & 0.285 & $\chi^2 (4) = 7.34$ \\ 
\hspace{3mm} \phantom{-- }\$150k -- & 0.068 & 0.072 & 0.064 & $p = 0.1188$ \\ 
\emph{Occupation type} \\
\hspace{3mm} Full-time & 0.497 & 0.474 & 0.521 & $\chi^2 (2) = 2.45$ \\ 
\hspace{3mm} Part-time & 0.102 & 0.098 & 0.106 & $p = 0.2935$ \\ 
\hspace{3mm} Not working & 0.401 & 0.428 & 0.372 &  \\ 
\emph{Marital status} \\
\hspace{3mm} Married/live with partner & 0.690 & 0.673 & 0.707 & $\chi^2 (1) = 0.92$ \\ 
\hspace{3mm} Other & 0.310 & 0.327 & 0.293 & $p = 0.3369$ \\ 
\midrule 
\# of observations in the sample & 764 & 388 & 376 \\
\bottomrule
\end{tabular}
\end{table}

\clearpage 
\subsection{Sample Comparison in the Panel Study}
\label{appendix:additional_panel_sample_comparison}

\begin{figure}[h]
\centering 
\includegraphics[width=0.8\textwidth]{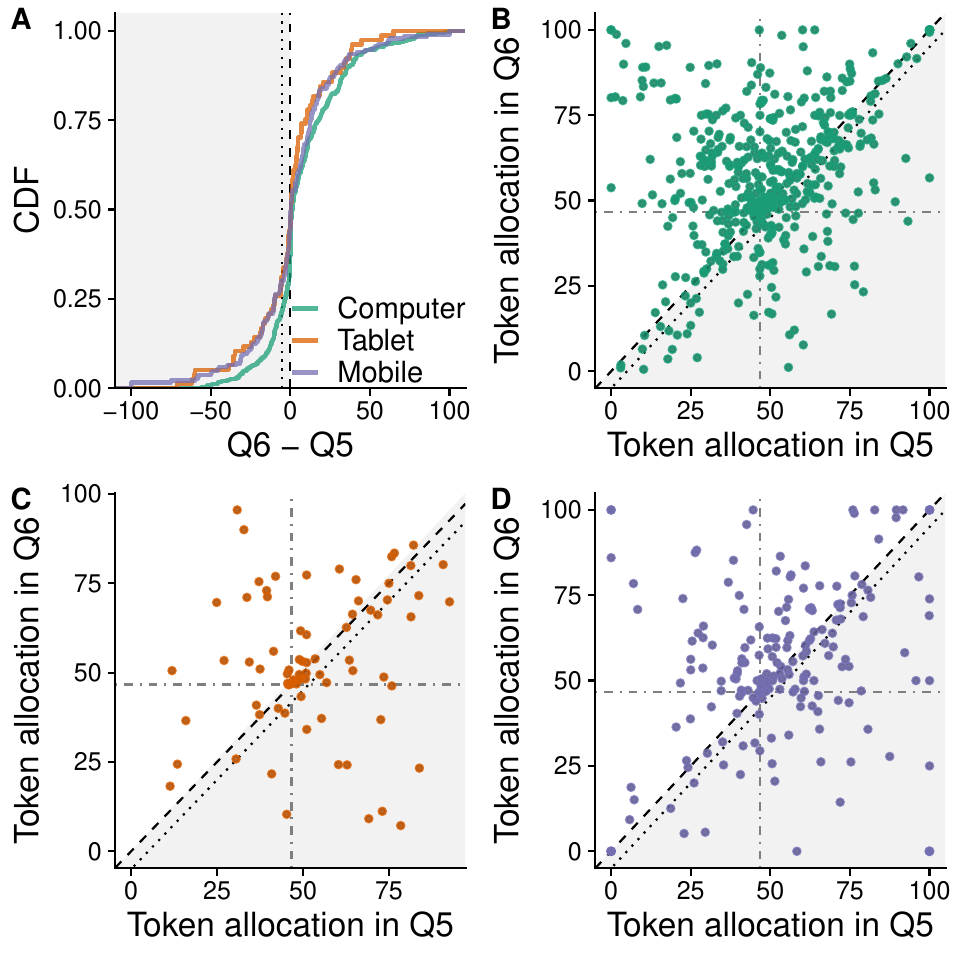}
\caption{(Panel) Event monotonicity. (A) Empirical CDFs of token allocation difference. The dotted line represents 5-token margin. Two-sample Kolmogorov-Smirnov test $p$-values: Computer vs. Tablet, $p = 0.116$; Computer vs. Mobile, $p = 0.044$; Tablet vs. Mobile, $p = 0.575$. (B-D) Token allocations in two questions. The dot-dashed lines at $46.67$ indicate the number of tokens which equalizes payouts in two events.}
\label{fig:uas_monotonicity_by_device}
\end{figure}

\begin{figure}[h]
\centering 
\includegraphics[width=\textwidth]{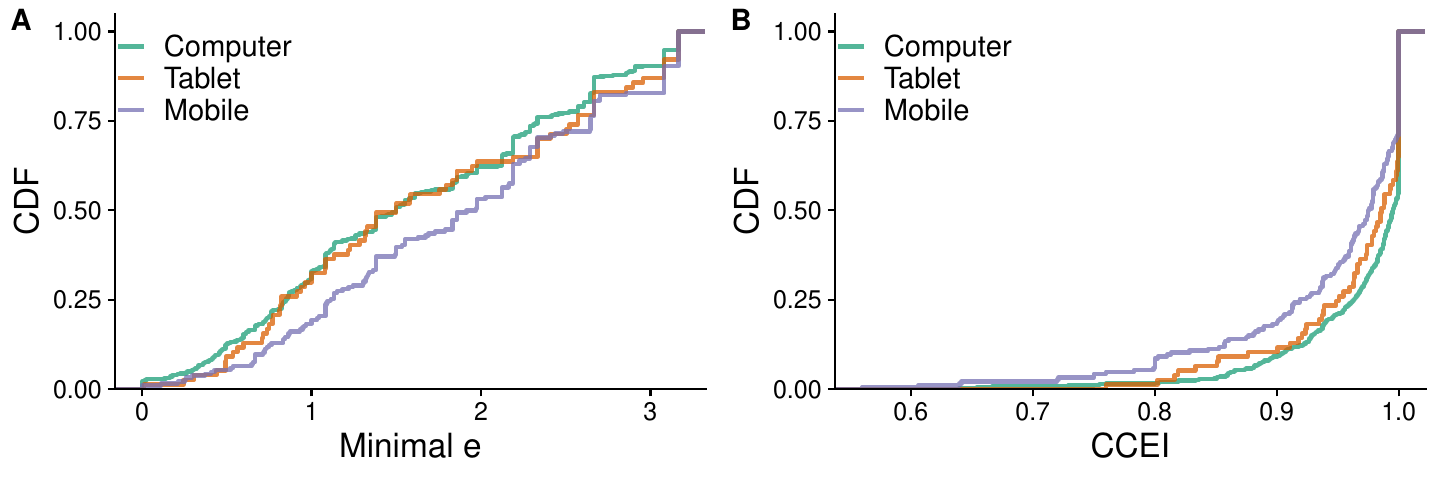}
\caption{Distribution of measures, by subjects using computer,tablet, or mobile. Two-sample Kolmogorov-Smirnov tests between the ``computer'' sample and the ``mobile'' sample: $p = 0.007$ for $e_*$ and $p < 0.001$ for CCEI.}
\label{fig:uas_measures_by_device}
\end{figure}

\begin{figure}[h]
\centering 
\includegraphics[width=0.9\textwidth]{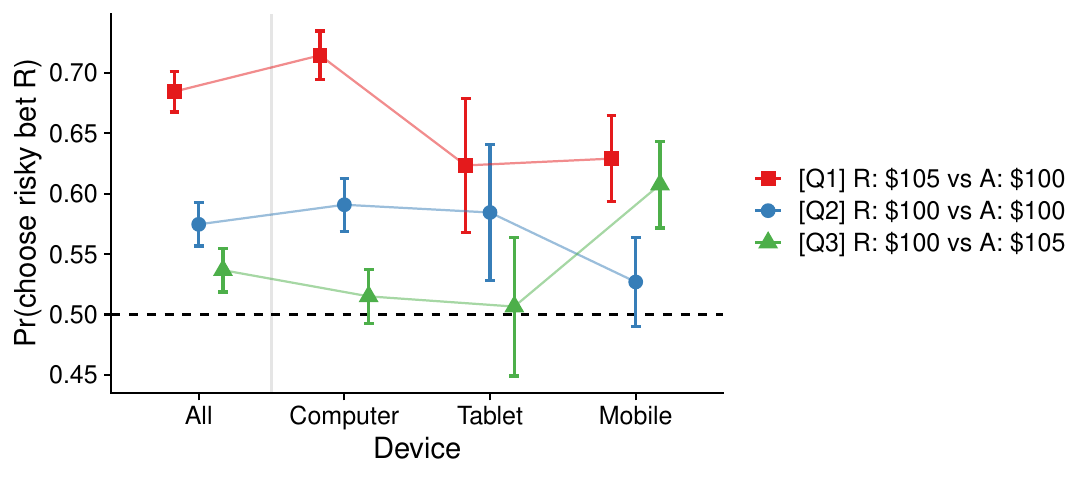}
\caption{(Panel) Probability of choosing a risky bet in each question in the standard-Ellsberg task, by subjects using computer, tablet, or mobile.}
\label{fig:uas_ellsberg_choice_freq_by_device}
\end{figure}

\clearpage 
\section{Power Calculation}
\label{appendix:power_calc}

It is well known that tests in revealed preference theory can have low power when used on certain configurations of budget sets. As a result, it is common to assess the power of a test by comparing the pass rates (the fraction of choices that pass the relevant revealed preference axiom) of the observed choice data from some benchmark behavior such as purely random choices.~\footnote{The idea of using random choices as a benchmark is first applied to revealed preference theory by \citeSOMpapers{bronars1987}. This approach is the most popular in empirical application: see, among others, \citeSOMpapers{andreoni2002}, \citeSOMpapers{fisman2007}, \citeSOMpapers{choi2007}, \citeSOMpapers{crawford2010habits}, \citeSOMpapers{beatty2011}, \citeSOMpapers{adams2014}, and \citeSOMpapers{dean2016}. For overview of power calculation, see discussion in \citeSOMpapers{andreoni2013power} and \citeSOMpapers{crawford2014empirical}.} 

We assess the power of the tests using two kinds of data-generating processes. In the first benchmark, we use the simple bootstrap procedure to look at the power from an ex-post perspective \citepSOMpapers{andreoni2002}. More precisely, for each budget set, we randomly pick one choice from the set of choices observed in the experiment. We repeat this to generate 10,000 synthetic choice data. In the second benchmark, we generate 10,000 datasets in which choices are made at random and uniformly distributed on the frontier of the budget set \citepSOMpapers[Method 1 of][]{bronars1987}. Table~\ref{table:test_power} report pass rates. The simulated choices almost always violate SEU and MEU. Pass rates for GARP test range from 0.23 to 0.68, depending on the underlying data-generating process. These numbers are higher than those reported in other studies \citep[e.g.,][]{choi2007,choi2014who}, but given that each type of problem has only 10 budgets, the configuration of budgets in our design has reasonable power to detect GARP violations.

\begin{sidewaystable}
\centering 
\caption{Pass rates under different data generating process. Simulation sample size $N = 10,000$.}
\label{table:test_power}
\resizebox{\textwidth}{!}{%
\begin{tabular}{ll rrr rrr rrr rrr}
\toprule 
 & & \multicolumn{3}{c}{GARP} & \multicolumn{3}{c}{PS} & \multicolumn{3}{c}{SEU} & \multicolumn{3}{c}{MEU} \\
 \cmidrule(lr){3-5} \cmidrule(lr){6-8} \cmidrule(lr){9-11} \cmidrule(lr){12-14} 
Data-generating process & Device & Type~1 & Type~2 & Joint & Type~1 & Type~2 & Joint & Type~1 & Type~2 & Joint & Type~1 & Type~2 & Joint \\
\midrule 
Bootstrap: Laboratory \\
\hspace{3mm} Market-stock (Large) & & 0.4591 & 0.5577 & 0.2562 & 0.5762 & 0.8742 & 0.3913 & 0.0007 & 0.0001 & 0.0000 & 0.0011 & 0.0001 & 0.0000 \\ 
\hspace{3mm} Market-stock (Small) & & 0.4769 & 0.4005 & 0.1934 & 0.6689 & 0.7269 & 0.4003 & 0.0002 & 0.0000 & 0.0000 & 0.0002 & 0.0000 & 0.0000 \\ 
\hspace{3mm} Market-Ellsberg & & 0.4712 & 0.5365 & 0.2560 & 0.7429 & 0.8373 & 0.5196 & 0.0002 & 0.0002 & 0.0000 & 0.0002 & 0.0004 & 0.0000 \\ 
Bootstrap: Panel \\
\hspace{3mm} Market-stock (Large) & All & 0.3642 & 0.3162 & 0.1172 & 0.4894 & 0.7299 & 0.2545 & 0.0001 & 0.0000 & 0.0000 & 0.0001 & 0.0000 & 0.0000 \\ 
\hspace{3mm} from Market-stock (Small) & All & 0.3184 & 0.2842 & 0.0881 & 0.4857 & 0.7136 & 0.2521 & 0.0001 & 0.0000 & 0.0000 & 0.0001 & 0.0000 & 0.0000 \\ 
\hspace{3mm} Market-stock (Large) & Desktop/laptop & 0.3940 & 0.3533 & 0.1389 & 0.4878 & 0.7319 & 0.2661 & 0.0002 & 0.0000 & 0.0000 & 0.0002 & 0.0000 & 0.0000 \\ 
\hspace{3mm} Market-stock (Small) & Desktop/laptop & 0.3314 & 0.2908 & 0.0967 & 0.5106 & 0.7245 & 0.2789 & 0.0001 & 0.0001 & 0.0000 & 0.0001 & 0.0001 & 0.0000 \\ 
\midrule 
Uniform random & & 0.2270 & 0.1440 & 0.0332 & 0.5002 & 0.5350 & 0.1925 & 0.0000 & 0.0000 & 0.0000 & 0.0000 & 0.0000 & 0.0000 \\
\bottomrule 
\end{tabular}
}
\end{sidewaystable}

\begin{table}[h]
\centering 
\caption{CCEI and $e^*$ calculated with randomly generated choice data. Simulation sample size $N = 10,000$.}
\label{table:test_power_distance}
\begin{tabular}{l rrr rrr rrr}
\toprule 
 & \multicolumn{3}{c}{CCEI} & \multicolumn{3}{c}{$e^* (\text{SEU})$} & \multicolumn{3}{c}{$e^* (\text{MEU})$} \\
 \cmidrule(lr){2-4} \cmidrule(lr){5-7} \cmidrule(lr){8-10} 
 & Type~1 & Type~2 & Joint & Type~1 & Type~2 & Joint & Type~1 & Type~2 & Joint \\
\midrule 
Mean & 0.9123 & 0.9194 & 0.8716 & 1.2476 & 1.3770 & 1.7121 & 1.2217 & 1.3770 & 1.7012 \\ 
Median & 0.9256 & 0.9436 & 0.8841 & 1.0833 & 1.5510 & 1.7639 & 1.0804 & 1.5510 & 1.7580 \\ 
SD & 0.0829 & 0.0817 & 0.0855 & 0.5968 & 0.6523 & 0.5293 & 0.6060 & 0.6523 & 0.5410 \\
\bottomrule
\end{tabular}
\end{table}

\clearpage 
\section{Design Detail}
\label{appendix:design_detail}

\subsection{The Set of Budgets}
\label{appendix:design_detail_budgets}

\begin{table}[h]
\centering 
\caption{The set of 20 budgets. The numbers indicate ``exchange value'' for each account $(z_1, z_2)$.}
\label{table:set_budgets}
\begin{tabular}{rclcccc}
\toprule 
 & & & \multicolumn{2}{c}{Lab} & \multicolumn{2}{c}{Panel} \\
 \cmidrule(lr){4-5} \cmidrule(lr){6-7}
 & Type & Order & Account 1 & Account 2 & Account 1 & Account 2 \\
\midrule 
1 & 1 & random & 0.30 & 0.18 & 3.0 & 1.8 \\
2 & 1 & random & 0.30 & 0.24 & 3.0 & 2.4 \\
3 & 1 & random & 0.38 & 0.30 & 3.8 & 3.0 \\
4 & 1 & random & 0.40 & 0.40 & 4.0 & 4.0 \\
5 & 1 & random & 0.50 & 0.12 & 5.0 & 1.2 \\
6 & 1 & random & 0.50 & 0.24 & 5.0 & 2.4 \\
7 & 1 & random & 0.50 & 0.34 & 5.0 & 3.4 \\
8 & 1 & random & 0.50 & 0.44 & 5.0 & 4.4 \\
9 & 1 & random & 0.60 & 0.30 & 6.0 & 3.0 \\
10 & 1 & fixed (5$^\text{th}$) & 0.32 & 0.28 & 3.2 & 2.8 \\
11 & 2 & random & 0.14 & 0.50 & 1.4 & 5.0 \\
12 & 2 & random & 0.24 & 0.50 & 2.4 & 5.0 \\
13 & 2 & random & 0.28 & 0.32 & 2.8 & 3.2 \\
14 & 2 & random & 0.30 & 0.36 & 3.0 & 3.6 \\
15 & 2 & random & 0.30 & 0.42 & 3.0 & 4.2 \\
16 & 2 & random & 0.30 & 0.56 & 3.0 & 5.6 \\
17 & 2 & random & 0.38 & 0.52 & 3.8 & 5.2 \\
18 & 2 & random & 0.40 & 0.50 & 4.0 & 5.0 \\
19 & 2 & random & 0.50 & 0.56 & 5.0 & 5.6 \\
20 & 2 & fixed (6$^\text{th}$) & 0.32 & 0.28 & 3.2 & 2.8 \\
\bottomrule 
\end{tabular}
\end{table}

\clearpage 
\subsection{Simulated Price Paths for the Market-Stock Task}
\label{appendix:design_detail_price_path}

In order to simulate price paths, we use a Geometric Brownian Motion (GBM): 
\begin{equation*}
dS_t = \mu S_t dt + \sigma S_t dW_t , 
\end{equation*}
where $W_t$ is a Wiener process, $\mu$ is a drift parameter, and $\sigma$ is a volatility parameter. To simulate trajectories of GBM, we calculate increments of $S$: 
\begin{equation*}
S_{t+h} = S_t \times \exp \left( (\mu - \sigma^2/2) h + \sigma \sqrt{h} Z \right) , 
\end{equation*}
with $Z \sim N (0, 1)$. 

As a first step, we generated $N$ paths of GBM, where each path $P^n = (P_0^n, P_1^n, \dots, P_T^n)$ has the common starting price $P_0$ and $T$ periods of prices. 
We then group these $N$ paths into several categories, based on several observable features: (i) absolute movement within $T$ periods; (ii) final price is higher than the initial price; (iii) final price is lower than the initial price; (iv) trends such as up-down, down-up, straight-gain, straight-loss, and cycle. 

After visually inspecting the pattern of each price path, we handpicked 28 paths and then asked workers on Amazon Mechanical Turk what they believed the future price of each path would be. We used a ``bins-and-balls'' belief elicitation task (also known as a histogram elicitation method) introduced by \citeSOMpapers{delavande2008eliciting} to elicit subjective belief distribution. The method, later refined by \citepSOMpapers{delavande2011measuring} and \citeSOMpapers{rothchild2012expectations}, is simple and easy to understand. It has been shown to work well in experiments conducted at developing countries \citepSOMpapers{delavande2011measuring} and online survey \citepSOMpapers{breunig2020standard}.

The idea of the task is as follows. First, the (continuous) state of the world (ranges of future prices) is partitioned into 20 disjoint and exhaustive bins. Second, subjects are asked to place 20 ``balls,'' each representing 5\% probability mass, into these bins. The subjects were then asked to express how likely they believed that the price to be in each or the 20 ranges. Figure~\ref{fig:design_bins-and-balls} illustrates the task. 

\begin{figure}[h]
\centering 
\includegraphics[width=0.95\textwidth]{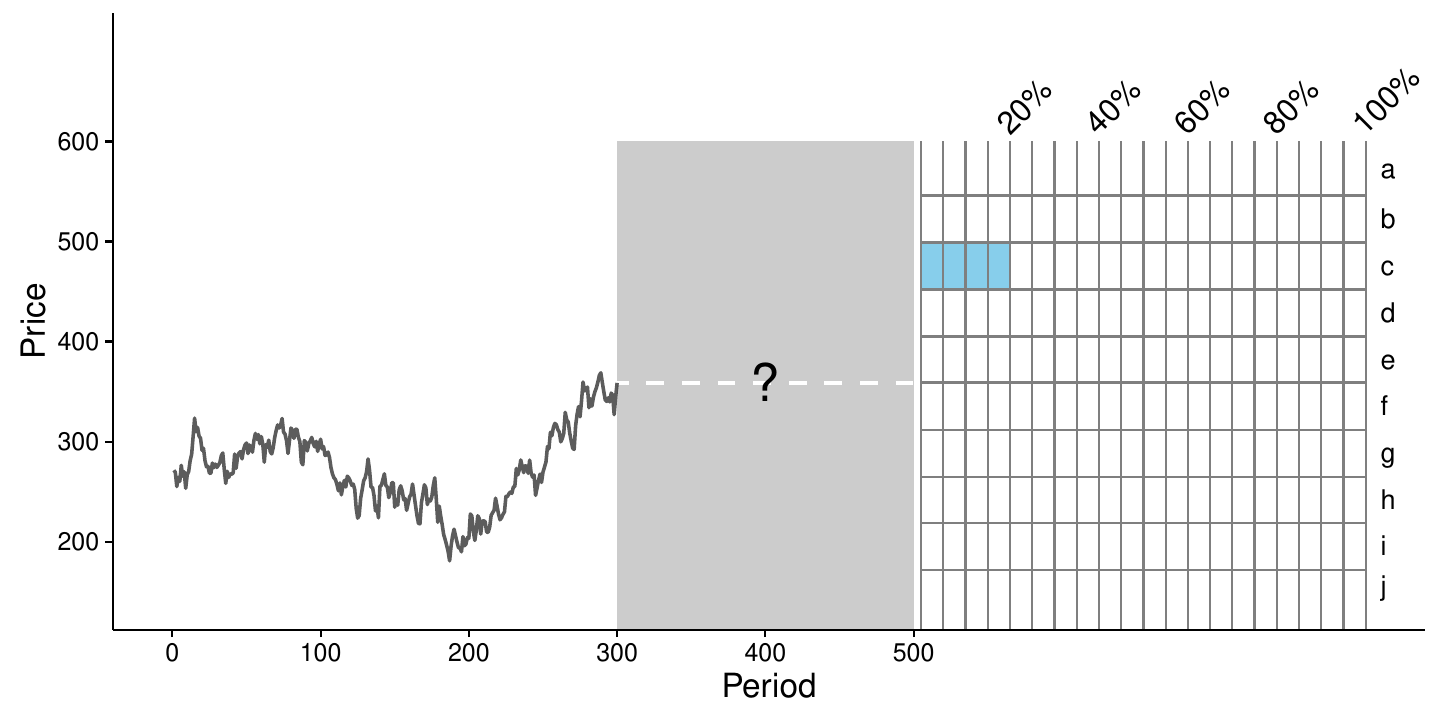}
\begin{tikzpicture}[scale=0.5]
\tikzstyle{vertex}=[circle,minimum size=10pt,inner sep=0pt]
\foreach \name/\x in {1/1, 2/2, 3/3, 4/4, 5/5, 6/6, 7/7, 8/8, 9/9, 10/10, 11/11, 12/12, 13/13, 14/14, 15/15, 16/16}
\node[vertex,fill=blue!25,draw=black] (G-\name) at (\x,0) {};
\foreach \name/\x in {17/17, 18/18, 19/19, 20/20}
\node[vertex,draw=black] (G-\name) at (\x,0) {};
\end{tikzpicture}
\caption{Illustration of the bins-and-balls belief elicitation task.}
\label{fig:design_bins-and-balls}
\end{figure}
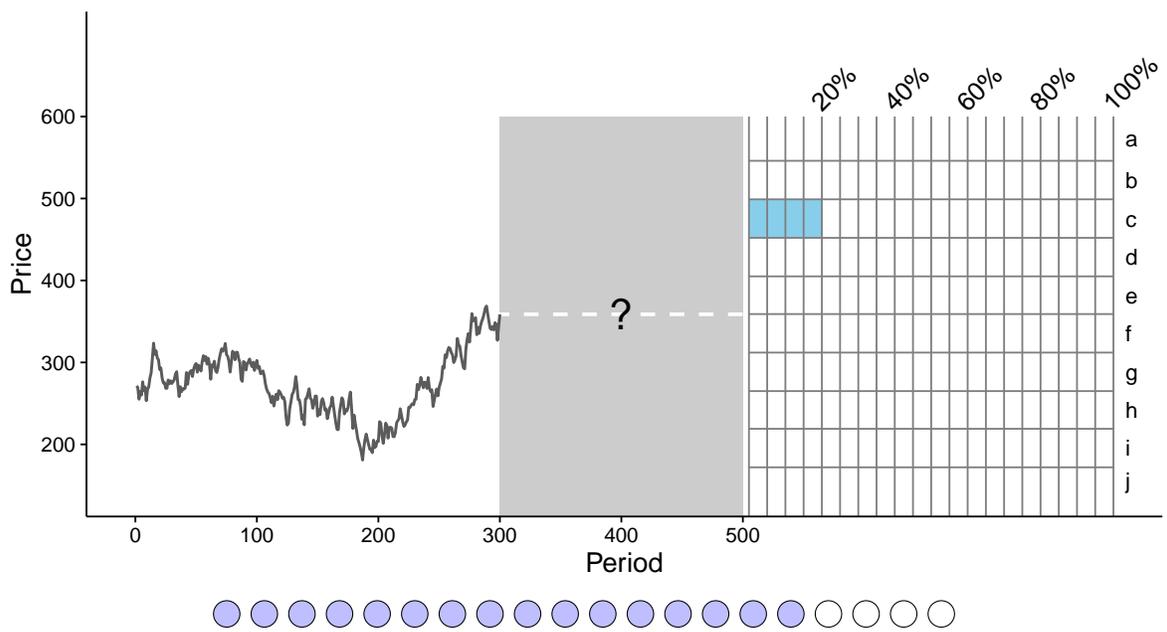

The elicited belief distributions were then averaged across subjects. Some price paths, especially those with clear upward or downward trend, tend to be associated with skewed distributions. Others have more symmetric distributions. We thus selected two relatively ``neutral'' ones from the latter set for the main experiment.

\clearpage 
\subsection{Post-Experiment Survey in the Laboratory Study}
\label{appendix:design_detail_post-exp_survey}

\paragraph{Demographic information.} 
\begin{enumerate}
\item What is your age? 
\item What is your gender? 
\item What is your ethnicity? 
\item What is your major? 
\end{enumerate}

\paragraph{Three-item cognitive reflection test.} 
\begin{enumerate}
\item If it takes 5 people 5 months to save a total of \$5,000, how many months would it take 100 people to save a total of \$100,000? 
\item A TV and a radio cost \$110 in total. The TV costs \$100 more than the radio. How much does the radio cost? 
\item In a lake, there is a patch of lily pads. Each day, the patch doubles in size. If it takes 48 days for the patch to cover the entire lake, how long would it take for the patch to cover half of the lake? 
\end{enumerate}

\clearpage 
\section{Instructions for the Experiments} 
\label{appendix:instruction} 

\newcommand{\numTask}{3}
\newcommand{\numSurvey}{2}
\newcommand{\showUpFee}{\$10}
\newcommand{\numPracticeAllocation}{3}
\newcommand{\numQallocationPath}{20}
\newcommand{\pathLengthPast}{300} 
\newcommand{\pathLengthForecast}{200} 
\newcommand{\pathLengthFuture}{500} 
\newcommand{\pathThreshold}{10\%}
\newcommand{\numQallocationBalls}{20}
\newcommand{\numBeliefBox}{20}
\newcommand{\numBeliefGrid}{10}
\newcommand{\percentBeliefBox}{five}

\begin{myparindent}{0pt}
\noindent\makebox[\linewidth]{\rule{\textwidth}{1pt}} 

\vspace{-1.5em}
\section*{Welcome!} 
Thank you for participating in today's experiment. 
\\

Please turn off all electronic devices, especially phones and tablets. 
During the experiment you are {\bf not} allowed to open or use any other applications on these laboratory computers, except for the interface of the experiment. 
\\

This experiment is designed to study decision making. 
You will be paid for your participation in cash privately at the end of the session. 
Please follow the instructions carefully and do not hesitate to ask the experimenter any questions by raising your hand. 
The experimenter will then come to your desk. 

\subsection*{Structure of the experiment} 
The experiment consists of \numTask{} tasks and a survey. 
We will hand out specific instructions for each of the tasks just before you are to perform that task. 
\\

At the front of this laboratory you will see several opaque bags labeled {\bf A}, {\bf B}, and so on, which we will use in some of the tasks during the experiment. 
Each of these bags contains colored chips. 
The exact composition of chips in each bag (for example, how many of them are red) may or may not be announced to you. 
If you wish, you can inspect these bags {\bf after} completing all sections of the experiment. 

\subsection*{Payment}
In order to determine the payment, {\bf one task} and {\bf one question from that task} will be randomly selected. 
Your payoff in the experiment will consist of the amounts you earned in the selected question plus a \showUpFee{} show-up fee if you complete the experiment as announced. 
The specific rules applied to determine payoffs for each section will be described in detail in the instructions for that part. 
\\

To select {\bf one task} and {\bf one question} that will determine your payment, the assistant rolled two fair dice for each participant. 
The assistant wrote down two numbers, one indicating the task and another indicating the question in that particular task that counts for payment. 
The note was placed into a sealed envelope. 
Please write your participant ID on the envelope once you receive it. 
Please do not open the envelope until you are instructed by the experimenter. 
\\

Remember that the question determining your payment is selected before you make any decisions in the experiment. 
This protocol of determining payments suggests that you should make a decision in each question as if it is the question that determines your payment. 

\subsection*{Important rules}
In the experiment we use a web browser. 
It is important that you ... 
\begin{enumerate}
\item do not close or refresh the browser, 
\item do not open other windows/tabs on the browser, 
\item do not exit the full screen mode, and 
\item do not open other applications and programs. 
\end{enumerate}
If you exit the full screen mode, please click the button at the top right corner to enter the full screen mode again. 
\\

Please raise your hand if you have any questions regarding the structure of the experiment. 

\noindent\makebox[\linewidth]{\rule{\textwidth}{1pt}}

\clearpage 
\noindent\makebox[\linewidth]{\rule{\textwidth}{1pt}} 

\vspace{-1.5em}
\section*{Task 1} 

\subsection*{Overview}
In this part of the experiment, you will be asked a total of~\numQallocationPath{} independent questions that share a common structure. 
Your goal is to invest tokens in two different accounts. 
The accounts pay off according to the value of a stock. 
\\

There is a {\bf hypothetical} company which we refer to as {\bf Company X}. 
We simulated a history of stock prices of this company using a model frequently used in financial economics. 
You will be presented a chart plotting the history of Company X's stock prices. 
The figure below is an illustration of such chart. Note that the price history presented in the image is meant to be an example, and is not the same one as you will see in the task. 

\begin{figure}[h]
\centering 
\includegraphics[width=0.5\textwidth]{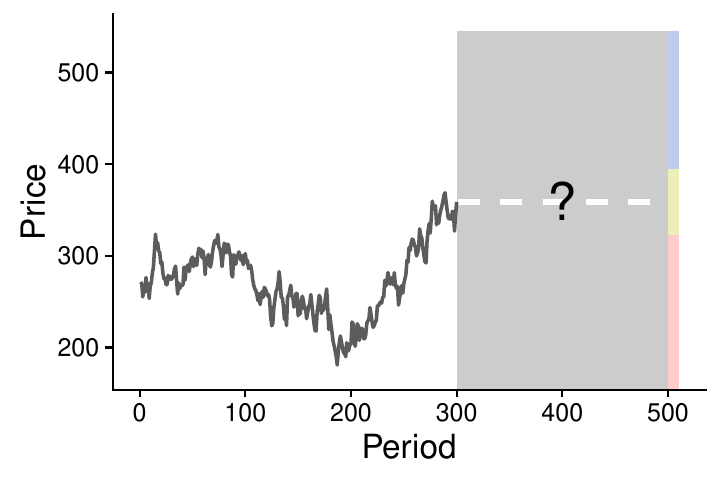}
\end{figure}

The chart shows stock prices from period~1 to~\pathLengthPast{} (imagine that period~\pathLengthPast{} is ``today''). 
You do not know the movement beyond period~\pathLengthPast{}, the area shaded in gray. 
Your payoff in this task depends on the ``future'' value of Company X's stock price at period~\pathLengthFuture{}, i.e., at the end of the chart. 
More precisely, it depends on whether the final price lies in the Blue area (increase by more than the threshold), Yellow area (change up to the threshold), or Red area (decrease by more than the threshold). 
In this example, the threshold is set to \pathThreshold{}. 

\subsection*{How it works}
Now we will explain the task in detail. 
\\

You will be asked a total of~\numQallocationPath{} independent questions. 
In each decision problem, you will be endowed with~100 tokens and asked to choose the portion of this amount (between~0 and~100 tokens, inclusive and divisible) that you wish to allocate between two accounts. 
Tokens allocated to each account may have different monetary values. 
Your payoff in this task will be determined by the following three components: 
\begin{enumerate}[label=(\roman*)]
\item the monetary value of tokens in each account, which is given in the question, 
\item your allocation of tokens in each of the two accounts, and
\item in which colored area Company X's stock price lies at period \pathLengthFuture{}. 
\end{enumerate}

\subsection*{Two types of questions}
There are two types of questions. 
\\

In Type 1 questions, two accounts are 
\begin{table}[h]
\begin{tabular}{lcl}
Account {\bf Blue-or-Yellow} & : & Stock price increases by a positive percentage \\
& & or decreases by at most \pathThreshold{}. \\
Account {\bf Red} & : & Stock price decreases by more than \pathThreshold{}. 
\end{tabular}
\end{table}

In Type 2 questions, two accounts are 
\begin{table}[h]
\begin{tabular}{lcl}
Account {\bf Blue} & : & Stock price increases by more than \pathThreshold{}. \\
Account {\bf Yellow-or-Red} & : & Stock price decreases by any percentage \\
& & or increases at most \pathThreshold{}. 
\end{tabular}
\end{table}

These two types of questions appear in random order. 
To understand the decision problem for the given trial correctly, it would be of your best interest to check the type of the question (1 or~2) on the right of the stock chart and also at the top of the ``allocation table'' which will be explained below. 

\begin{figure}[tp]
\centering 
\begin{subfigure}{0.48\textwidth}
\includegraphics[width=\textwidth]{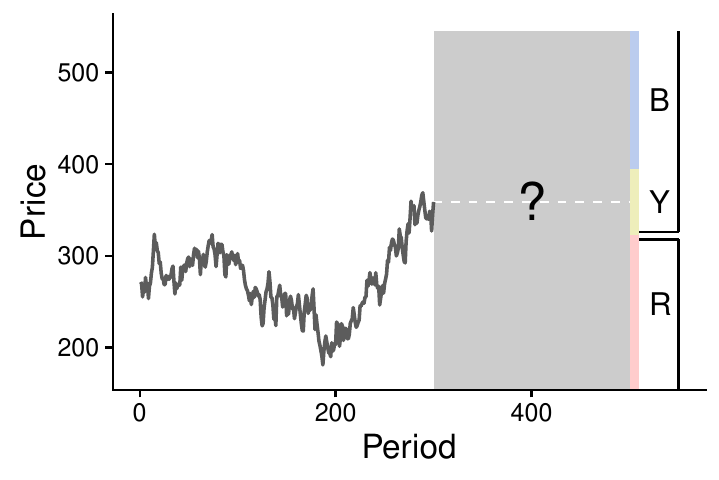}
\subcaption{Example: Type 1 question}
\end{subfigure}
\begin{subfigure}{0.48\textwidth}
\includegraphics[width=\textwidth]{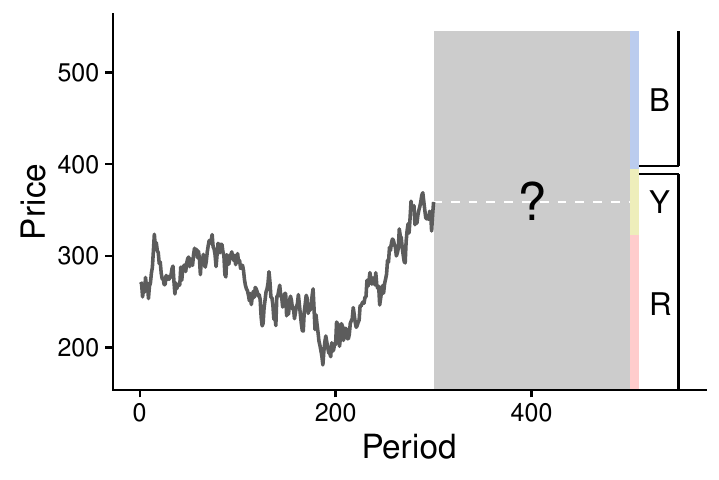}
\subcaption{Example: Type 2 question}
\end{subfigure}
\end{figure}

The ``allocation table'' at the bottom block of the screen shows information on the monetary values of tokens. 
In the example of Type~1 question below, each token you allocate to the {\bf Blue-or-Yellow} account is worth \$0.30 (30 cents), while each token you allocate to the {\bf Red} account is worth \$0.25 (25 cents). 
Notice that monetary values of tokens may change across questions. 

\begin{table}[h]
\newcolumntype{C}{>{\centering\arraybackslash}p{0.1\linewidth}}
\newcolumntype{D}{>{\centering\arraybackslash}p{0.2\linewidth}}
\centering 
\begin{tabular}{l C C p{0.02\linewidth} D}
& \cellcolor{myPaleBlue} B & \cellcolor{myPaleYellow} Y & & \cellcolor{myPaleRed} R \\
Token value & \multicolumn{2}{c}{\$0.30} & & \multicolumn{1}{c}{\$0.25} \\
Tokens & \multicolumn{2}{c}{75} & & \multicolumn{1}{c}{25} \\
Account value & \multicolumn{2}{c}{\$22.50} & & \multicolumn{1}{c}{\$6.25} \\
\end{tabular}
\end{table}

\subsection*{How to make a decision} 
You can allocate 100 tokens between two accounts using the slider. 
The table will be updated instantly once you move the slider, showing current allocations of tokens and their implied payment amounts if the stock ends up in the corresponding color region. 
No cursor appears at the start of the experiment---you need to click anywhere on the slider line to activate it. 
\\

Alternatively, you can allocate tokens by directly putting numbers in one of the boxes, or clicking {\bf up}/{\bf down} arrow (which appears when you mouse-over the box) to make small adjustments. 

\subsection*{How your payoff for this task is determined}
Your payoff is determined by the number of tokens in your two accounts, and the ``future'' value of stock X at the end of the chart (period \pathLengthFuture{}), which will be simulated after you complete all questions. 
\\

Suppose you allocated~75 tokens to account {\bf Blue-or-Yellow} and~25 tokens to account {\bf Red} as in the above example. 
If this question has been chosen to determine your payoff, your payoff will be determined by the price of company X's stock in period~\pathLengthFuture{}. 
If stock X hits blue or yellow area at period \pathLengthFuture{} (as in panel (a) below), then you will earn $75 \times \$0.30 = \$22.50$ (22 dollars 50 cents). 
On the other hand, if stock X hits red area (as in panel (b) below), then you will earn $25 \times \$0.25 = \$6.25$ (6 dollars 25 cents). 
Amount below one cent will be rounded up. 

\begin{figure}[h]
\centering 
\begin{subfigure}{0.48\textwidth}
\includegraphics[width=\textwidth]{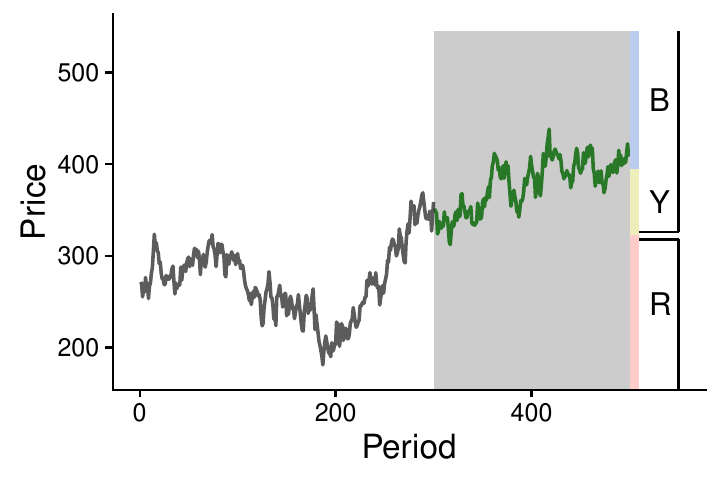}
\subcaption{Example: Stock price hits Blue area}
\end{subfigure}
\begin{subfigure}{0.48\textwidth}
\includegraphics[width=\textwidth]{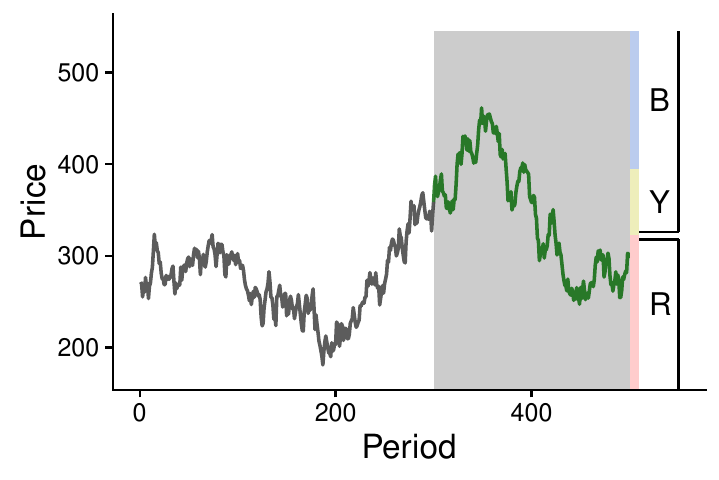}
\subcaption{Example: Stock price hits Red area}
\end{subfigure}
\end{figure}

\subsection*{Important}
\begin{itemize}
\item The history of stock prices of Company X up to period 300 (``today'') is the same throughout this task. 
\item You will not know ``future'' prices (between period~301 and period~500) until you complete all tasks in the experiment. 
\item Token values and question types can vary between questions. 
\item Once you hit the Proceed button, you cannot change your decision. 
You {\bf cannot} go back to previous pages, either. 
Note also that you {\bf cannot} change the question by refreshing the browser once it is displayed. 
\item Remember that the question that will determine your payment has already been selected at the start of the experiment. 
It is your best interest to treat each question as if it is the question that determines you payment. 
\end{itemize}

\thispagestyle{empty}
\subsection*{Hypothetical Stock Market}

As we mentioned before, we simulated a history of stock prices of this company using a model frequently used in financial economics. 
The following chart illustrates eight such simulated stocks in our ``hypothetical stock market''. 
\\

Let's imagine that we are at period 300 (``today'') and we do not know the ``future'' stock prices (periods 301 to 500). 
\\

The black solid line represents the price history of our Company X. 
You will see ONLY this price history during this task. 

\begin{figure}[h]
\centering 
\includegraphics[width=0.5\textwidth]{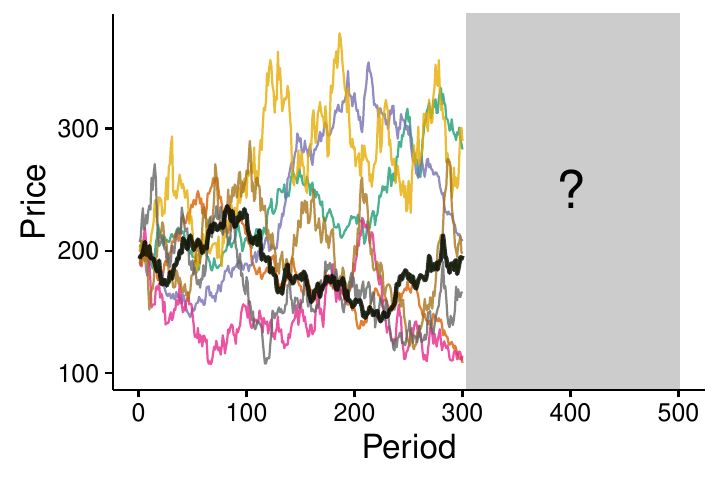}
\end{figure}

\noindent\makebox[\linewidth]{\rule{\textwidth}{1pt}} 

\clearpage 
\noindent\makebox[\linewidth]{\rule{\textwidth}{1pt}} 

\vspace{-1.5em}
\section*{Task 2} 

\subsection*{Overview}
In this part of the experiment, you will be asked a total of \numQallocationBalls{} independent questions that share a common structure. 
Your goal is to invest tokens in two different accounts. 
The accounts pay off according to the color of a chip drawn from a bag at the end of the experiment. 
\\

There is an opaque {\bf bag Z} which contains~30 colored chips. 
Each chip is either Blue, Yellow, or Red. 
The number of chips of each color is unknown to you: There can be anywhere from~0 to~30 Blue chips, anywhere from~0 to~30 Yellow chips, and anywhere from~0 to~30 Red chips, as long as the total number of Blue, Yellow, and Red chips sums to~30. 
Your payoff in this task depends on the color of a chip you will draw at the end of the experiment. 

\begin{figure}[h]
\centering 
\begin{tikzpicture}
\fill[myPaleBlue] (6.0,0.6) circle (0.2);
\fill[myPaleYellow] (6.0,0.0) circle (0.2);
\fill[myPaleRed] (6.0,-0.6) circle (0.2);
\node at (6.6,0.6) {$\times ?$};
\node at (6.6,0.0) {$\times ?$};
\node at (6.6,-0.6) {$\times ?$};
\draw[black, very thick] (5.6,1.0) -- (5.6,-1.0) -- (7.1,-1.0) -- (7.1,1.0);
\node at (6.35,1.4) {Bag Z: Total 30~chips};
\end{tikzpicture}
\end{figure}

The contents of bag Z has already been determined by an assistant at the beginning of the experiment. 
If you wish, you can inspect the contents of the bag {\bf after} completing the experiment. 

\subsection*{How it works}
Now we will explain the task in detail. 
\\

You will be asked a total of~\numQallocationPath{} independent questions. 
In each decision problem, you will be endowed with~100 tokens and asked to choose the portion of this amount (between~0 and~100 tokens, inclusive and divisible) that you wish to allocate between two accounts. 
Tokens allocated to each account may have different monetary values. 
Your payoff in this task will be determined by the following three components: 
\begin{enumerate}[label=(\roman*)]
\item the monetary value of tokens in each account, 
\item your allocation of tokens in each of the two accounts, and
\item the color of the chip you will draw from the bag at the end of the experiment. 
\end{enumerate}

\subsection*{Two types of questions}
There are two types of questions. 
\\

In Type 1 questions, two accounts are 
\begin{table}[h]
\begin{tabular}{ll}
Account {\bf Blue-or-Yellow} & : The color of chip drawn from the bag is either Blue or Yellow. \\
Account {\bf Red} & : The color of chip drawn from the bag is Red. 
\end{tabular}
\end{table}

In Type 2 questions, two accounts are 
\begin{table}[h]
\begin{tabular}{ll}
Account {\bf Blue} & : The color of chip drawn from the bag is Blue. \\
Account {\bf Yellow-or-Red} &: The color of chip drawn from the bag is either Yellow or Red. 
\end{tabular}
\end{table}

These two types of questions appear in random order. 
To understand the decision problem for the given trial correctly, it would be of your best interest to check the account structure at the top of the ``allocation table'' which will be explained below. 
\\

The ``allocation table'' at the bottom block of the screen shows information on the monetary values of tokens. 
In the example of Type~2 question below, each token you allocate to the {\bf Blue} account is worth \$0.36 (36 cents), while each token you allocate to the {\bf Yellow-or-Red} account is worth \$0.24 (24 cents). 
Notice that monetary values of tokens may change across questions. 

\begin{table}[h]
\newcolumntype{C}{>{\centering\arraybackslash}p{0.1\linewidth}}
\newcolumntype{D}{>{\centering\arraybackslash}p{0.2\linewidth}}
\centering 
\begin{tabular}{l D p{0.02\linewidth} C C}
& \cellcolor{myPaleBlue} B & & \cellcolor{myPaleYellow} Y & \cellcolor{myPaleRed} R \\
Token value & \multicolumn{1}{c}{\$0.36} & & \multicolumn{2}{c}{\$0.24} \\
Token & \multicolumn{1}{c}{30} & & \multicolumn{2}{c}{70} \\
Account value & \multicolumn{1}{c}{\$10.80} & & \multicolumn{2}{c}{\$16.80} \\
\end{tabular}
\end{table}

\subsection*{How to make a decision} 
You can allocate 100 tokens between two accounts using the slider. 
The table will be updated instantly once you move the slider, showing current allocations of tokens and their implied payment amounts. 
No cursor appears at the start of the experiment---you need to click anywhere on the slider line to activate it. 
\\

Alternatively, you can allocate tokens by directly putting numbers in one of the boxes, or clicking {\bf up}/{\bf down} arrow (which appears when you mouse-over the box) to make small adjustments. 

\subsection*{How your payoff for this task is determined}
Your payoff is determined by the number of tokens in your two accounts, and the color of a chip you will draw from the bag at the end of the experiment. 
\\

Suppose you allocated~30 tokens to account {\bf Blue} and~70 tokens to account {\bf Yellow-or-Red} as in the above example. 
If this question has been chosen to determine your payoff, your payoff will be determined by the color of a chip you will draw at the end of the experiment. 
If it is Blue, then you will earn $30 \times \$0.36 = \$10.80$ (10 dollars 80 cents). 
On the other hand, if it is Yellow or Red, then you will earn $70 \times \$0.24 = \$16.80$ (16 dollars 80 cents). 
Amount below one cent will be rounded up. 

\subsection*{Important}
\begin{itemize}
\item The composition of the bag (how many chips are blue, yellow, or red) is the same throughout this task. 
\item You will not know the actual composition until you complete all tasks in the experiment. 
\item Token values and question types can vary between questions. 
\item Once you hit the Proceed button, you cannot change your decision. 
You {\bf cannot} go back to previous pages, either. 
Note also that you {\bf cannot} change the question by refreshing the browser once it is displayed. 
\item Remember that the question that will determine your payment has already been selected at the start of the experiment. 
It is your best interest to treat each question as if it is the only question that determines you payment. 
\end{itemize}

\noindent\makebox[\linewidth]{\rule{\textwidth}{1pt}} 

\clearpage 
\noindent\makebox[\linewidth]{\rule{\textwidth}{1pt}} 

\vspace{-1.5em}
\section*{Task 3} 

There are two bags, bag A and bag B, each of which contains~30 chips. 
Each chip is either orange or green. 
The contents of each bag is as follows: 
\begin{itemize}
\item Bag A contains~10 orange chips and~10 green chips. 
\item Bag B contains 20 chips. Each chip is either orange or green. The number of chips of each color is unknown to you: 
There can be anywhere from~0 to~20 orange chips, and anywhere from~0 to~20 green chips, as long as the total number of orange and green chips sums to~20. 
\end{itemize}

\begin{figure}[h]
\centering 
\begin{tikzpicture}
\fill[Orange] (0.0,0.6) circle (0.2);
\fill[Green] (0.0,0.0) circle (0.2);
\node at (0.6,0.6) {$\times 10$};
\node at (0.6,0.0) {$\times 10$};
\draw[black, very thick] (-0.4,1.0) -- (-0.4,-0.4) -- (1.1,-0.4) -- (1.1,1.0);
\node at (0.35,1.4) {Bag A: Total 20 chips};
\fill[Orange] (6.0,0.6) circle (0.2);
\fill[Green] (6.0,0.0) circle (0.2);
\node at (6.6,0.6) {$\times ?$};
\node at (6.6,0.0) {$\times ?$};
\draw[black, very thick] (5.6,1.0) -- (5.6,-0.4) -- (7.1,-0.4) -- (7.1,1.0);
\node at (6.35,1.4) {Bag B: Total 20 chips};
\end{tikzpicture}
\end{figure}

The contents of bag B has already been determined at the beginning of the experiment. 
If you wish, you can inspect the contents of each bag {\bf after} completing the experiment. 
\\

You will now answer several questions, each of which offers you a choice between bets on the color of a chip that you will draw from one of two bags at the end of the experiment (if this section is chosen for payment). 
\\

\clearpage 
You will first be asked to choose one of the two colors. 
We will call this {\bf Your Color}. 
You will be paid only if a chip of this color is drawn from the bag at the end of the experiment. 
\\

You will then be asked to answer the following three questions. 
\begin{itemize}
\item Question: Please select a bet 
  \begin{enumerate}
  \item \$10.50 if a chip drawn from bag A is of {\bf Your Color} and \$0 otherwise. 
  \item \$10.00 if a chip drawn from bag B is of {\bf Your Color} and \$0 otherwise. 
  \end{enumerate}
\item Question: Please select a bet 
  \begin{enumerate}
  \item \$10.00 if a chip drawn from bag A is of {\bf Your Color} and \$0 otherwise. 
  \item \$10.00 if a chip drawn from bag B is of {\bf Your Color} and \$0 otherwise. 
  \end{enumerate}
\item Question: Please select a bet 
  \begin{enumerate}
  \item \$10.00 if a chip drawn from bag A is of {\bf Your Color} and \$0 otherwise. 
  \item \$10.50 if a chip drawn from bag B is of {\bf Your Color} and \$0 otherwise. 
  \end{enumerate}
\end{itemize}

\subsection*{How your payoff for this section is determined}
Suppose one of the three questions in this section is selected for determining your payment. 
If you chose bet~1 in that particular question, you will draw a chip from bag A. 
On the other hand, if you chose bet~2 in that particular question, you will draw a chip from bag B. 
In either case, you will get payment if the color of the drawn chip matches with Your Color. 
\\

\subsection*{How to make a decision}
You will see four questions on the screen. 
The first one asks which color you want to use as {\bf Your Color} and the following three questions ask which bet you would like to play. 
\\

For each question, you can make your selection by clicking on the check box for the option you would like to choose. 
You can change your selection as many times as you want, and there is no time limit. 
Once you make your selections for all three questions, you can submit them by clicking Proceed. 
You will not be able to change your decision after that. 

\noindent\makebox[\linewidth]{\rule{\textwidth}{1pt}} 

\end{myparindent}

\clearpage 
\bibliographystyleSOMpapers{ecta}
\bibliographySOMpapers{seu_exp}

\end{document}